\documentclass[a4paper,UKenglish,cleveref, autoref]{lipics-v2021}

\nolinenumbers
\hideLIPIcs
\usepackage{hyperref}
\usepackage{cite}

\usepackage{booktabs}   %% For formal tables:

\usepackage{subcaption} %% For complex figures with subfigures/subcaptions

\usepackage{ebproof}
\usepackage{mathtools}

\RequirePackage{ifthen}
\newcommand{\varcomment}{0} 
\newcommand{\version}{0}
\newcommand{\VarArX}{1}
\newcommand{\condinc}[2]{\ifthenelse{\equal{\varcomment}{0}}{#1}{\green {#2}}}
\newcommand{\SLV}[2]{\ifthenelse{\equal{\version}{0}}{#1}{\RED{#2}}}
\newcommand{\ArX}[2]{\ifthenelse{\equal{\VarArX}{0}}{#1}{#2}}

\usepackage[T1]{fontenc} % optionalF\sred

\usepackage{amsmath}
\usepackage{amsthm}
\usepackage{amsthm}

\usepackage{url}
\usepackage{mathtools}
\usepackage{relsize}

\usepackage{color}
\usepackage{graphicx}

\usepackage{cmll}
\usepackage{stmaryrd}
\usepackage{proof}
\usepackage{wrapfig}
\usepackage{amssymb}

\newcommand{\opair}[2]{\langle #1 : #2 \rangle}
\newcommand{\pwred}{{\xredx{\pw}{}\ }}

\newcommand{\alphabet}{\mathbb{A}}
\newcommand{\OutC}[1]{\mathtt{Out}_{#1}}

\newcommand{\liberal}{unbiased\xspace}

\newcommand{\bm}[1]{\boldsymbol{#1}}

\newcommand{\lam}{\lambda}
\newcommand{\CbN}{CbN\xspace}
\newcommand{\CbV}{CbV\xspace}

\usepackage{xspace}
\newcommand{\subs}[2]{ \{#2 /#1\} }
\usepackage{floatflt}
\usepackage[font=scriptsize,skip=4pt]{caption}

\newcommand{\timplies}{\mbox{ implies }}
\newcommand{\tand}{\mbox{ and }}

\newcommand{\op}{\mathbf{o}}
\newcommand{\obs}{\mathtt {obs}}

\newcommand{\print}[1]{\mathtt{out}_{#1}}
\newcommand{\tick}{\mathtt{tick}}
\newcommand{\ticks}{\mathsmaller{\surd}}
\newcommand{\dtick}{\Delta_{\ticks}}

\newcommand{\dplus}{\Delta_\oplus}
\newcommand{\dout}[1]{\Delta_{#1}}%{\Delta^{\mathtt{out}}_0}
\newcommand{\dz}{\Delta_z}
\newcommand{\pnf}[1]{\omega(#1)}
\newcommand{\PNF}{\Nnf_{\omega}}
\newcommand{\sem}[1]{\den{#1}}

\newcommand{\cpo}[1]{\bm #1}

\newcommand{\seqr}[1]{\langle#1_n\rangle_{n\in\Nat}}

\newcommand{\w}[1]{\mathtt \obs(#1)}
\newcommand{\BT}[1]{\mathtt{BT(#1)}}
\newcommand{\nnf}{\emph{nf}\xspace}
\newcommand{\hnf} {\emph{hnf}\xspace}

\renewcommand{\AA}{A}

\newcommand{\mdist}[1]{[#1]}

\newcommand{\ik}{{i\in \bm k}}

\newcommand{\ww}{\bm {\mathsf{W}}}   %{{\mathbf W}}
\newcommand{\cc}{\bm {\mathsf{C}}}  %{{\bf C}}

\newcommand{\hh}{\bm {\mathsf{H}}}

\newcommand{\leftc}{\bm {\mathsf{L}}}  %{{\bf L}} %contesto left
\newcommand{\rightc}{\bm {\mathsf{R}}} %  {{\bf R}} %contesto left

\newcommand{\fseq}[1]{\langle #1 \rangle}

\newcommand{\seq}[1]{\langle#1_n\rangle_{n}}
\newcommand{\cbv}{{\mathtt{cbv}}}
\newcommand{\cbn}{{\mathtt{cbn}}}
\newcommand{\n}{\mathtt n}

\newcommand{\Nnfv}{\Nnf_v}

\newcommand{\Real}{\mathbb{R}}
\newcommand{\Nat}{\mathbb{N}}

\newcommand{\pr}{\pn}
\newcommand{\pn}{{pn}}

\newcommand{\obsN}{\obs_{\mathsmaller \Nnf}}

\newcommand{\hole}[1]{\llparenthesis #1\rrparenthesis}
\newcommand{\tm}{t}
\newcommand{\tmt}{t}
\newcommand{\tms}{s}
\newcommand{\tmu}{u}
\newcommand{\tmr}{r}
\newcommand{\tmp}{p}
\newcommand{\tmq}{q}
\renewcommand{\to}{\xrightarrow{}}

\newcommand{\conv}[1]{\mkern-2mu{}^\infty_{#1}~}

\newcommand{\llred}{\uset{\llsmall~}{\red}}

\newcommand{\xfull}[1]{\usetfull{#1~}{\full}}

\newcommand{\llfull}{\xfull{\llsmall~}}

\newcommand{\xRedx}[2] {\mathrel{{\uset{#1}{\Red}}{}_{\mkern-5mu#2}}}

\newcommand{\llRed}{\uset{\llsmall~}{\Red}}

\newcommand{\iRed}{\uset{\neg\llsmall~}{\Red}}

\newcommand{\llRedx}[1]  {\mathrel{\llRed{}_{\mkern-6mu{#1}}}}

\newcommand{\iRedx}[1]  {\mathrel{\iRed{}_{\mkern-6mu#1}}}

\newcommand{\xbackredx}[2] {\mathrel{{\uset{#1}{\leftarrow}}{}_{\mkern-6mu#2}}~}

\newcommand{\nllred}{\uset{\neg \llsmall~~}{\red}}

\newcommand{\llredx}[1]  {\mathrel{\llred{}_{\mkern-6mu#1}}}

\newcommand{\parto}{\parRew{}}

\newcommand{\iparto}{\nmakepar{\ired~}}

\newcommand{\stred}{{\ered}}

\newcommand{\parmark}{\circ\mkern -0mu}

\newcommand{\makepar}[1]{~\parmark \mkern-15mu #1}
\newcommand{\nmakepar}[1]{~\parmark \mkern-13mu #1}

\newcommand{\parred}{{\makepar {\red}}}

\newcommand{\nleparred}{{\nmakepar  {\nlered} }}

\newcommand{\MPLambda}{\MDST{\PLambda}}

\newcommand{\PLambda}{\Lambda_\oplus}

\newcommand{\true}{\mathtt{T}}
\newcommand{\false}{\mathtt{F}}

\newcommand{\Val}{\mathcal V}
\newcommand{\Snf}{\mathcal S}
\newcommand{\Nnf}{\mathcal N}
\newcommand{\Hnf}{\mathcal H}

\newcommand{\two}{\frac{1}{2}}
\newcommand{\four}{\frac{1}{4}}

\newcommand{\Set}{\mathbb{S}}

\newcommand{\tr}{{}^{\mkern-6mu*}}

\newcommand{\ie}{\emph{i.e.}\xspace}
\newcommand{\eg}{\emph{e.g.}\xspace}
\newcommand{\ih}{\emph{i.h.}\xspace}
\renewcommand{\iff}{\emph{iff}\xspace}
\usepackage{xcolor}

\newcommand{\RED}[1]{{\color{red}{#1}}}

\newcommand{\green}[1]{{\color{green}{#1}}}

\newcommand{\multiset}[1]{[ #1 ]}
\newcommand{\mset}{\multiset}

\newcommand{\iI}{{i \in I}}

\makeatletter
\newcommand{\usetfull}[3][0ex]{%
	\mathrel{\mathop{#3}\limits_{
			\vbox to#1{\kern-6\ex@
				\hbox{$\scriptstyle#2$}\vss}}}}
\makeatother

\newcommand{\defeq}{\coloneqq}

\newcommand{\Lim}{\mathtt{Lim}}

\mathchardef\mhyphen="2D
\newcommand{\xLim}[1]{\Lim_{#1}}
\newcommand{\wLim}{\xLim{\obs}}

\newcommand{\tolim}[1]{#1 \conv{\obs}}

\newcommand{\DST}[1]{\mathcal{D}(#1)}

\newcommand{\MDST}[1]{\mathcal{M}(#1)}
\newcommand{\norm}[1]{\|#1\|}

\newcommand{\m}{\mathtt m}   
\newcommand{\ud}{\mathtt u}   
\newcommand{\td}{\mathtt t}   
\renewcommand{\r}{\mathtt r}     
\newcommand{\s}{\mathtt s}

\newcommand{\Red}{\Rightarrow} 
  
\newcommand{\Redbv}{\Red_{\betav} }  
\newcommand{\Redo}{\Red_{\oplus}}

\newcommand{\sRed}{\uset{\surf~}{\Red}}

\newcommand{\iRedbv}{\iRedx{\betav}}

\newcommand{\nsRed}{\uset{\neg\surf~}{\Red}}
\newcommand{\nsRedx}[1]  {\mathrel{\nsRed{}_{\mkern-6mu#1}}}
\newcommand{\nsRedbv}{\nsRedx{\betav}}

\newcommand{\full}{\rightrightarrows}

\newcommand{\den}[1]{\llbracket {#1} \rrbracket}

\newcommand{\betav}{\beta_v}

\newcommand{\betab}{\mathtt{b}}
\newcommand{\bbeta}{\betab}

\newcommand{\ex}{\mathsf {e}}
\renewcommand{\int}{\mathsf{i}}
\newcommand{\nex}{{\neg\ex}}
\newcommand{\head}{\mathsf{h}}
\newcommand{\weak}{\mathsf{w}}
\newcommand{\surf}{\mathsf{s}}
\newcommand{\lsym}{\mathsf{l}}
\newcommand{\rsym}{\mathsf{r}}
\newcommand{\lo}{\mathsmaller{\textsc{lo}}}
\newcommand{\lex}{\mathsf{le}}
\newcommand{\los}{\textsc{u}}
\newcommand{\out}{\textsc{u}}

\newcommand{\pd}{\textsc{//u}~}
\newcommand{\pw}{{\textsc{//}}\weak}

\newcommand{\llsmall}{\textsc{e}}
\newcommand{\EX}{\textsc{e}}

\newcommand{\xrevredx}[2] {\mathrel{{\uset{#1}{\leftarrow}}{}_{\mkern-3mu#2}}}
\newcommand{\revred}{\leftarrow}

\newcommand{\xredx}[2]{{\mathrel{\,\uset{#1\,}{\red}{}_{\mkern-6mu{#2} }}} }
\newcommand{\redx}[1]{\mathrel{\red{}_{\mkern-6mu#1}}}
\newcommand{\red}{\rightarrow}

\newcommand{\ered}{\uset{\ex}{\red}}

\newcommand{\nered}{\uset{\neg \ex~}{\red}}

\newcommand{\hred}{\uset{\head}{\red}}
\newcommand{\nhred}{\uset{\neg \head~}{\red}}

\newcommand{\hredx}[1]  {\mathrel{\hred{}_{\mkern-6mu#1}}}
\newcommand{\nhredx}[1]{\mathrel{\nhred{}_{\mkern-6mu#1}}}

\newcommand{\hredb}  {\mathrel{\hredx{\beta}}}
\newcommand{\nhredb}{\mathrel{\nhredx{\beta}}}

\newcommand{\wred}{{\xredx{\weak\,}{}}}
\newcommand{\nwred}{\xredx{\neg \weak\,}{}}
\newcommand{\wredx}[1]  {\xredx{\weak}{#1}}
\newcommand{\nwredx}[1]  {\xredx{{\neg\weak~}}{#1}}
\newcommand{\wredbv}{\wredx{\betav}}
\newcommand{\nwredbv}{\nwredx{\betav}}

\newcommand{\lred}{\uset{\lsym}{\red}}

\newcommand{\lredbv}{\mathrel{\xredx{\lsym}{\betav}}}

\newcommand{\sred}{\xredx{\surf}{}}
\newcommand{\nsred}{\xredx{\neg \surf}{}}
\newcommand{\sredx}[1]{\xredx{\surf}{#1}}
\newcommand{\nsredx}[1]{\xredx{\neg\surf}{#1}}

\newcommand{\sredbb}  {\mathrel{\sredx{\bbeta}}}
\newcommand{\nsredbb}{\mathrel{\nsredx{\bbeta}}}

\newcommand{\sredo}  {\mathrel{\sredx{\oplus}}}

\newcommand{\sredbv}  {\mathrel{\sredx{\betav}}}
\newcommand{\nsredbv}{\mathrel{\nsredx{\betav}}}

\newcommand{\redbv}{\mathrel{\red_{\beta_v}}}

\newcommand{\redb}  {\mathrel{\redx{\beta}}}
\newcommand{\redbb}  {\mathrel{\redx{\betab}}}

\newcommand{\loredb}  {\xredx{\lo}{\beta}}

\newcommand{\redo}{\rightarrow_{\oplus}}

\newcommand{\nlered}{\xredx{\neg\lex}{}}

\newcommand{\lsred}{\mathrel{\xredx{\los }{}}}
\newcommand{\nlsred}{\mathrel{\xredx{\neg\los~}{}}}
\newcommand{\lsredx}[1]{\xredx{\los }{#1}}
\newcommand{\lsredbb}{\lsredx{\betab}}
\newcommand{\nlsredbb}{\xredx{\neg\los~}{\betab}}

\newcommand{\lsredb}{\xredx{\los}{\beta}}
\newcommand{\nlsredb}{\xredx{\neg\los~}{\beta}}
\newcommand{\lsredbv}{\xredx{\los}{\betav}}
\newcommand{\nlsredbv}{\xredx{\neg\los~}{\betav}}

\makeatletter
\newcommand{\uset}[3][0ex]{%
	\mathrel{\mathop{#3}\limits_{
			\vbox to#1{\kern-6\ex@
				\hbox{$\scriptstyle#2$}\vss}}}}
\makeatother

\theoremstyle{plain}

\newtheorem{thm}{Theorem}

\newtheorem{Def}[theorem]{Definition}

\newtheorem{prop}[theorem]{Proposition}

\newtheorem{cor}[theorem]{Corollary}

\newtheorem{fact}[theorem]{Fact}

\newtheorem*{theorem*}{Theorem}
\newtheorem*{proposition*}{Prop}
\newtheorem*{lemma*}{Lemma}
\newtheorem*{ex*}{Example}

\newtheorem*{cor*}{Cor.}
\newtheorem*{prop*}{Prop}
\newtheorem*{Def*}{Def}

\newcommand{\Ie}{\textit{i.e.}\xspace}

\usepackage{ebproof}
\usepackage{mathtools}
\usepackage{thmtools}
\renewcommand{\paragraph}{\subparagraph}

\Crefname{section}{Sect.}{Sections}
\Crefname{theorem}{Thm.}{Thm.}
\Crefname{thm}{Thm.}{Thm.}
\Crefname{proposition}{Prop.}{Prop.}
\Crefname{prop}{Prop.}{Prop.}
\Crefname{Def}{Def.}{Def.}
\Crefname{figure}{Fig.}{Figs.}
\title{Strategies for Asymptotic Normalization}	
\author{Claudia Faggian}{ IRIF, CNRS---Universit\'e de Paris Cit\'e, F-75013 Paris, France}{}{}{Partly funded by
	the ANR project PPS, ANR-19-CE48-0014.}
\author{Giulio Guerrieri}{Huawei Research, Edinburgh Research Centre,
	Edinburgh, United Kingdom}{giulio.guerrieri@huawei.com}{https://orcid.org/0000-0002-0469-4279}{}

\authorrunning{C. Faggian and G. Guerrieri} %TODO mandatory. First: Use abbreviated first/middle names. Second (only in severe cases): Use first author plus 'et al.'
\Copyright{Claudia Faggian and Giulio Guerrieri} %TODO mandatory, please use full first names. LIPIcs license is "CC-BY";  http://creativecommons.org/licenses/by/3.0/

 %TODO mandatory; please add comma-separated list of keywords
\EventEditors{Amy P. Felty}
\EventNoEds{1}
\EventLongTitle{7th International Conference on Formal Structures for Computation and Deduction (FSCD 2022)}
\EventShortTitle{FSCD 2022}
\EventAcronym{FSCD}
\EventYear{2022}
\EventDate{August 2--5, 2022}
\EventLocation{Haifa, Israel}
\EventLogo{}
\SeriesVolume{228}
\ArticleNo{8}
\begin{document}

\maketitle

\begin{abstract}
	We present an abstract technique to study normalizing strategies when termination is asymptotic, that is, it appears as a limit. Asymptotic termination occurs in several  settings, such as effectful, and in particular probabilistic computation---where the limits are distributions over the possible outputs---or  infinitary lambda-calculi---where the limits are infinitary terms such as B\"ohm trees.
%	In order to abstract from specific details, our approach is formulated in a framework extending with quantitative information the usual notion of Abstract Rewriting Systems, which is the core of usual finitary rewriting. Namely, we work in a refinement of Ariola and Blom's Abstract Rewriting Systems with Information. The same abstract methods thus apply to a variety of instances.

	As a concrete application, we obtain a result which is of independent interest: a normalization theorem for Call-by-Value (and---in a uniform way---for Call-by-Name) probabilistic lambda-calculus.
\end{abstract}

% !TEX root = main_2022.tex

\renewcommand{\a}{\alpha}
\renewcommand{\c}{\gamma}
\newcommand{\ka}{|\alpha|}
\newcommand{\kc}{|\gamma|}

%%%%%%%%%%%%%%%%%%%%%
\renewcommand{\iparto}{\nleparred}
\renewcommand{\parto}{\parred}

\newcommand{\eq}{\, = \,}
\newcommand{\eqdef}{:=}
\newcommand{\deff}{\eqdef}

\newcommand{\Root}[1]{\mapsto_{#1}}
\newcommand{\opp}[1]{\op(#1)}

\newcommand{\lv}{\lambda_v}
\newcommand{\pair}[2]{\langle #1,#2 \rangle}

\newcommand{\id}{I}
\newcommand{\snf}{\textit{snf}\xspace}
%%%%%%%%%%%%%%%%%%%%%
\newcommand{\brang}[1]{\langle  #1\rangle}
\renewcommand{\brack}[1]{\big(  #1\big)}

\newcommand{\LambdaOp}{\Lambda_{\OpSet}}
\newcommand{\OpSet}{\mathcal {O}}
\newcommand{\Rules}{\mathtt{Rules}}
\newcommand{\Rule}{\rho}

\newcommand{\NDS}{\textsf{NDE}\xspace}
\newcommand{\NDE}{\textsf{NDE}\xspace}

\newcommand{\AC}[3]{{#1}\mhyphen\mathtt{C}(#2,#3)}

\newcommand{\bop}{\betab\op}
\newcommand{\bOp}{\betab\OpSet}
\newcommand{\bvo}{\betav\oplus}

\newcommand{\ddelta}[1]{\Delta_{#1}}
\newcommand{\LambdaOut}{\Lambda_{\texttt{out}}}
%%%%%%%%%%%%%%%%%%%%%
\newcommand{\zero}{\mathbf{0}}
\renewcommand{\AA}{\mathcal A}
\newcommand{\BB}{\mathcal B}

\section{Introduction}
%Computation is   a process that produces a result by gradually increasing the amount of information to be output.

Probabilistic computation is an example of computational  paradigm where  the notion of termination is \emph{asymptotic}, that is, it appears as a \emph{limit}, %  (almost sure termination), 
as opposed to reaching a normal form  in a \emph{finite} number of steps.  Streams,  infinitary $ \lam$-calculus,  algebraic rewriting systems, effectful computation,  are other examples: the notion of asymptotic computation is pervasive.
%Even when the probability
%that a term leads to a normal form is 1 (almost sure termination), that degree of certitude is
%typically not reached in any finite number of steps, but it appears as a limit.
%\pink{Other examples in which the results appears as a limit are  computations with  effects (outputs, )  streams \cite{},  infinitary $ \lam$-calculus,  algebraic rewriting systems.}
 %
 Here, we  investigate asymptotic normalization, and propose  
a  technique  to prove  that
 a strategy  
 is guaranteed to produce   a maximal or---ideally---the best  possible result.  Our  technique is abstract (in the sense of Abstract Rewriting Systems)
 % (\ie, it does not depend on the structure of the elements to rewrite)
   and so of general application.

Rewriting is a foundation for the operational theory of formal calculi and programming languages---$\lam$-calculus being  the paradigmatic example where rewriting is an abstract form of program execution. 
 Even if a  programming language is usually defined by a specific 
\emph{evaluation strategy},  to have a  general  rewriting theory
  allows  for \emph{program transformations, optimizations, parallel/distributed implementations}, and provides  a base on which to reason about program equivalence.
The $\lam$-calculus has a rich    theory that studies the properties of  reductions. 
Asymptotic computation   is much less understood from a  rewriting point of view, with the notable exception of  infinitary $\lam$-calculus, whose rewriting theory, pioneered in  {\cite{KennawayKSV95,BerarducciIntrigila96,KennawayKloopSleepDeVries97}},  %has given rise to en extensive body of work.}
has  been extensively studied.

The process of rewriting 
describes the  \emph{computation of a result}.
Normal forms, head normal forms, values, may or must termination, are all possible notions of result. 
For concreteness, let us  focus on normal forms.
 Operationally, key questions about a system  are the existence and uniqueness of normal forms, but also \emph{how} the result is computed. 
In a \emph{finitary setting} we would ask:
may a computation 
	produce a result (\emph{Existence} of normal forms)?
	 If so, is the result unique? Do different computations on the same input lead to the same result (\emph{Uniqueness} of normal forms)?
	 How to compute a result? Is there a reduction  strategy that is guaranteed to output  a result, if any exists (\emph{Normalizing strategy})?
In the\emph{ asymptotic case}, %when the result appears as a limit,  
such questions are still relevant, but  need to be opportunely formulated. To answer, we then need suitable tools and techniques, because those for finitary computation do not necessarily transfer (the key game-changer being that  asymptotic termination does not provide a well-founded order, 
see \cite{pars} for examples in a probabilistic setting). %---one such example is uniqueness of the notion of result, which we discuss below.
%In particular, we  introduce a set  of  abstract techniques to establish that a strategy is asymptotically complete and has a unique limit.
%Remarkably, both such    \emph{ infinitary properties} reduce to  a  finitary one, \emph{factorization}---a simple form of standardization---and to  some \emph{local, elementary  tests}. % (internal steps do not increase information, and a local  diamond property). 

%Factorization is the cornerstone of our construction.

%\blue{The standard strategy is what bridges between the foundational theory and a programming language}.

\paragraph{ Abstract Asymptotic Rewriting.}
%\RED{In order to abstract from specific details, our approach is formulated in a framework extending with quantitative information the usual notion of Abstract Rewriting Systems, which is the core of usual finitary rewriting. Namely, we work in a refinement of Ariola and Blom's Abstract Rewriting Systems with Information. The same abstract methods thus apply to a variety of instances.
%	}

%\blue{
%Instances of asymptotic computation are quite diverse, and the   syntax of each system  may be rather complex.
%Working abstractly makes possible to analyze   asymptotic  properties in a way  %that is 
%independent of specific syntax, and to develop \emph{general} proof techniques. 
%
%\emph{Quantitative Abstract Rewrite System} (QARS) are a framework  to study asymptotic rewriting \emph{abstractly},  
%as the theory of Abstract Rewrite Systems (ARS) does for finitary computation. 
%}

Our approach is to study asymptotic reduction strategies and properties of limits in an \emph{abstract} way (independent of the specific syntax of a calculus) as the theory of Abstract Rewrite Systems (ARS) does for finitary computation, so to isolate proof-techniques which are of general application.
For example, in infinitary lambda calculus, the limit is usually a (possibly infinite) limit term, while in probabilistic lambda calculus, the limit is a distribution over (finite) terms. The former is concerned with the depth of the redexes, the latter with the probability of reaching a result. The  abstract  notions of limit and normalization subsumes both, and so abstract results apply to either setting. 
A further,  conceptual advantage of an abstract approach, is to  display  the essence of the arguments, an to neatly discriminate between 
 those properties which  rely on  specific structure of a  concrete setting, and  
those which belong to any asymptotic notion of computation.

Specifically, we work in the setting of  \emph{Quantitative Abstract Rewrite System} (QARS) \cite{parsLMCS}, a framework  to study asymptotic rewriting \emph{abstractly}  which  refines Ariola and Blom's ARSI \cite{AriolaBlom02}.
%Abstract Rewriting Systems with Information.

\paragraph{From normal forms to limits.}

Intuitively, a possibly infinite reduction  sequence $\seq \tmt$  from $\tm =\tm_0$ expresses a computation whose 
\emph{result}  is   the maximal amount of information produced by that  sequence. This is formalized as a limit.
When the reduction is deterministic, it is standard to interpret such a limit as the meaning $\sem \tmt$ of $\tm$.
If however  $\tmt$ has \emph{several possible reduction  sequences}, each  can produce a different outcome (a different limit). It is then  natural  to define  the meaning $\sem t$ of a term  $t$ as the  \emph{greatest} element in the set of limits, if any.\footnotemark
%
%\footnotetext{\RED{We could have defined $\sem{t}$ as the lub of the set of limits, % of reduction sequences from $t$,
%	as done for B\"ohm trees in \cite{AmadioCurien}. 
%But our focus is on developing a \emph{operational} theory to identify \emph{normalizing strategies} that asymptotically compute the result, \emph{internally} to the calculus. 
%This is why we require that $\sem{t}$ is an element of the set of limits.}}
\footnotetext{One could also  define $\sem{t}$ as the lub of the set of limits, but  this    opens the question if   there is a 
		strategy that asymptotically computes  $\sem{t}$, \emph{internally} to the calculus. Since our focus is developing an \emph{operational} theory,  we require that $\sem{t}$ is itself a limit -- it is a result that can be (asymptotically) computed.}
%}
%
Intuitively, this means that the notion of ``greatest amount of information produced by any  reduction sequence'' is well defined.
To adopt such a 
	 notion demands care---for example, in the case of probabilistic and effectful computation, non-deterministic evaluation brings out issues  which do not appear in  pure lambda-calculus, not even when   infinitary.

 Given a term $t$  and a general reduction,  the notion of result  $\sem t$ is not necessarily  defined: the set of  limits for $\tm$ may  contain  different maximal elements, or it may  not even have any   maximal element
 % (think of the real interval $[0,1)$, which has no maximum).
 {(think of $\Nat$ or $[0,1)$, which have no maximum)}.
Maximal limits play a role similar to normal forms, and the  following  questions 
 are then  natural. 
\begin{enumerate}
	\item Is there a strategy that  produces a maximal amount of information (a maximal limit)? 
	\item Given a term $t$, is $\den t$---the result of computing $\tm$---well defined? 
	
\end{enumerate}
In  \Cref{sec:strategies}  we provide tools to answer these questions, %precisely 
in this order, as we discuss next.

\paragraph{On the workflow (and the limits of confluence).}

	The $\lambda$-calculus has two fundamental syntactical results: \emph{confluence}, which  implies \textbf{uniqueness of normal forms}, 
	and the \emph{standardization theorem,} which implies \textbf{normalization}, namely that a normal form can be reached by a computable strategy, which is  a standard reduction   (typically,  left-to-right). Uniqueness guarantees that the notion of  result is well defined, normalization provides a method
	to actually compute it.

A  common workflow when studying  $\lam$-calculi is to first prove uniqueness of normal forms (via confluence), then normalization (via standardization). 
However,  
in an  asymptotic setting 
%\SLV{}{\footnote{Such an issue does not appear in infinitary $\lam$-calculus.}}
\emph{confluence does not directly  imply} that the set of limits has a greatest element, but only that it has a least upper bound.
 So, even if  confluence is established, 
 one still needs to prove that the lub is itself a limit, 
  which may be   a non-trivial task. For example, in the probabilistic $\lam$-calculus \cite{pars, FaggianRonchi, parsLMCS}, such a proof  relies on  (technical)   properties of probability distributions.

%\paragraph{Solution: change the work flow  (forgo confluence)}. 	
In this paper, we \emph{reverse the workflow}, and   focus on  normalization. 
In the finitary setting, if a rewriting relation $\red$ has a  strategy ${\ered} \subset {\red}$ which satisfies a suitable completeness hypothesis and    uniqueness of normal forms,  %(which is trivially true if the strategy is deterministic),
 so does $\red$  (see \cite{Vrijer}). % {even if $\red$ is not confluent}.
%to establish that each term $\tm$ has a unique normal form, it suffices that the strategy $\stred$ has a unique normal form. 
With opportune definitions, this   lifts well to the asymptotic setting. 
Forgoing confluence  and  focusing on normalization yields  an efficient and uniform  method which is easy to apply and which  provides  simultaneously 
(1.) existence and uniqueness of maximal limits, and (2.) a strategy   to compute it.

\paragraph{Content and contributions.}
%We start by illustrating  asymptotic computation with examples (\Cref{sec:examples}).  
%Instances of asymptotic computation are quite diverse, and the   syntax of each system  may be rather complex.
%To study  rewriting \emph{abstractly},    as the theory of Abstract Rewriting Systems (ARS) does for finitary computation, makes possible to analyze   asymptotic  properties in a way  
%independent of specific syntax, and to develop \emph{general} proof techniques.
%In \Cref{sec:QARS}  we 
%present  the setting of Quantitative Abstract Rewriting Systems (QARS), which are ARSs enriched  with a notion of observation. QARS are  a natural refinement of ARSI \cite{AriolaBlom02}, already  adopted in  \cite{parsLMCS}.
% We illustrate examples of asymptotic computation in \Cref{sec:examples}.  
% Instances of asymptotic computation are quite diverse, and the   syntax of each system  may be  complex.
% We first study  them \emph{abstractly},    as the theory of Abstract Rewriting Systems (ARS) does for finitary computation, so as to analyze   asymptotic  properties in a way  
% independent of a specific syntax, and to develop \emph{general} proof techniques.
% In \Cref{sec:QARS}  we 
% present  the setting of Quantitative Abstract Rewriting Systems (QARS, already  adopted in  \cite{parsLMCS}), \ie~ARS enriched  with a notion of observation. QARS are  a natural refinement of ARSI \cite{AriolaBlom02}.
We start by illustrating  asymptotic computation with examples (\Cref{sec:examples}).  
Instances of asymptotic computation are quite diverse, and the   syntax of each system  may be rather complex.
To study  rewriting \emph{abstractly},    in the spirit of of Abstract Rewriting Systems (ARS), makes possible to analyze   asymptotic  properties in a way  
independent of specific syntax, and to develop \emph{general} proof techniques.
In \Cref{sec:QARS}  we 
present  the setting of Quantitative Abstract Rewriting Systems (QARS) \cite{parsLMCS}, which are ARS enriched  with a notion of observation. QARS are  a natural refinement of ARSI \cite{AriolaBlom02}.%\footnotemark
% \footnotetext{One moves from ARSI to QARS simply moving from partial orders to $\omega$-cpos as codomain of the observation function, 
% QARS allows us to deal uniformly with both finitary and asymptotic rewriting.}

\ArX{}{\footnotetext{The difference between ARSI to QARS is moving from partial orders to $\omega$-cpos as codomain of the observation function.
QARS allows us to deal uniformly with both finitary and asymptotic rewriting.}
}

Our  first  original  contribution, 
and the heart   of this paper,  is \Cref{sec:strategies}, which proposes  a   proof technique  to  study asymptotic  reduction strategies, and properties of   the limits.   
We first introduce \emph{asymptotic normalization}, which gives at the same time a tool to establish  the  existence of maximal limits---or of a greatest one---and a way to compute it. 
{It formalizes the intuition that  a  normalizing strategy  gradually computes (in a finite or infinite number of steps)  the/a maximal amount of information that an element $\tm$ can produce.} 
We then show  (\Cref{sec:tools}) that asymptotic normalization can be  established by proving 
 that a strategy is asymptotically complete and has a unique limit. 
Remarkably, such \emph{infinitary properties} reduce to  a  finitary one, \emph{factorization} (a simple form of \emph{standardization}) and to  some \emph{local, elementary  tests}, yielding a practical and versatile  proof-technique. 

We then  apply  our method to some representative case studies based on $\lam$-calculus. In order to do so, we first  
 revisit normalization for $\lam$-calculus---uniformly for  CbV and CbN---so as to have  \emph{a (novel) normalizing strategy} which is well-suited to  asymptotic normalization, and to deal with    probabilistic $\lam$-calculi (\Cref{sec:normalization}).
The application of  our method  to  \emph{probabilistic $\lam$-calculus} yields a result of independent interest,
 which was left as open question in \cite{FaggianRonchi} (Remark 27  there),  namely a theorem of asymptotic normalization for Call-by-Value probabilistic $\lambda$-calculus. We develop the CbV case explicitly  in \Cref{sec:PCbV}---the same results hold \emph{in a uniform way} for Call-by-Name. 
The same  technique applies to other monadic calculi such as calculi with output (as  we sketch in \Cref{sec:output}), but also  to the asymptotic computation of B{{\"{o}}hm Trees, which can be obtained as the limit of a normalizing strategy (we leave  this case  to Appendix \ref{sec:BT}). 

%Omitted proofs are available in the online extended version \cite{long}.

\subsection{Three  examples of Asymptotic Computation}
%\section{Asymptotic Computations and Non-Deterministic Evaluation}
\label{sec:examples}
%Computation is  --intrinsically-- a process that produces a result by gradually increasing the amount of information to be output.
We  illustrate  three diverse examples of asymptotic computation, where 
  the result of the computation is the limit of an \emph{infinitary} process.  All three examples are built on $\lam$-calculus.

\paragraph{Probabilistic computation.}
A probabilistic  program $P$ is a stochastic model generating a distribution  over all possible outputs of $P$. 
%\SLV{}{The intuition is that  the   program $P$  is executed, and random choices are made by sampling.}
%
Even  if  the   termination probability is $1$ (\emph{almost sure termination}), that degree of certitude 
is typically  not reached in a finite number of steps, but  \emph{as a limit}.  A standard example is a term $M$ that reduces to either a  normal form or $M$ itself, with equal probability $1/2$. After $n$ steps, $M$ is in   normal form with probability 
$\two + \frac{1}{2^2} + \dots + \frac{1}{2^n}$. % $\sum_{k:1}^n \frac{1}{2^k}$.
\emph{Only at the limit} this computation terminates with probability $1$.
%\pink{A \emph{probabilistic} program   generates a probability \emph{distribution over possible outputs}; it
%	\emph{terminates}  (on a given input) \emph{with a certain probability};
%	it  may have runs which take \emph{infinitely many steps}	even when termination has probability $1$.
%}
%
A direct way to model higher-order probabilistic computation is to endow   the  untyped $\lambda$-calculus 
 with a binary operator $\oplus$ which models fair, binary probabilistic choice: %Here $\oplus$ is  just flipping a fair coin: 
$M_1\oplus M_2$  reduces to either $M_1$ or $M_2$ with equal probability $1/2$; we write this as $M_1\oplus M_2 \red \mset{\two M_1, \two M_2}$. 
Intuitively,  \emph{the result} of evaluating a probabilistic term is a \emph{distribution} on its possible outputs.

\begin{example}\label{ex:proba}
	Let $\dplus=\lam x.\id \oplus (xx) $, where $\id=\lam x.x$. The term  $M:=\dplus\dplus$ has  the behavior we have described above, and evaluates to $\id$ with probability $1$ only at the limit.
\end{example}

\newcommand{\str}{\texttt s}
\paragraph{Computations with output.}
\label{sec:ex_output}
Consider a program that can print an output. %(we take this example from \cite{GavazzoPhD}).
%\RED{Following an example from \cite{GavazzoPhD}, let us consider a program that can print an output.}
Following \cite{GavazzoPhD}, we can represent this with a pair ${\str:M}$, where $\str$ is a string over an alphabet $\alphabet$, and $M$ is a term of  the $\lambda$-calculus extended with a set of  operators  $\print{}=\{\print c \mid c\in \alphabet\}$.
The term   $\print c (P)$ outputs $c$ , adding
it to the string, and continues as $P$.
That is,
%\begin{center}
	$\opair \str {\print c (P)} \red  \opair {c.\str} P$.
%\end{center}
%
\begin{example}\label{ex:output} Let $\alphabet =\{0,1\}$, and  
  $\dout 0:= \lam x.\print 0(xx)$. The computation from $\opair \epsilon {\dout 0\dout 0}$ (with $\epsilon$ the empty string) produces a stream: a  string of $0$'s whose length tends to infinity. 
\end{example}
\paragraph{Infinite Normal Forms.}\label{sec:ex_infinitary}
Infinitary $\lam$-calculi 
\cite{KennawayKSV95,BerarducciIntrigila96,KennawayKloopSleepDeVries97} model infinite structures in $\lambda$-calculi. 
Terms and reduction sequences need not be finite. An infinite reduction sequence is 
 \emph{strongly convergent} if 
 the depth of the contracted redex tends to infinity.
Based on different depth measures, in \cite{KennawayKloopSleepDeVries97}  eight
different infinitary $\lambda$-calculi are developed.
If the calculus is \emph{confluent}, the 
 \emph{infinite normal form of a term $N$ is unique}, and it is  the  \emph{meaning of $N$}.
Infinite normal forms are well-known   in $\lam$-calculus in the form of B{{\"{o}}hm tres \cite{Barendregt} or  L{\'{e}}vy-Longo  trees \cite{Levy78}.

%in?nite λ-terms constitute a syntactic approach to the semantics of ?nite λ-terms with (e.g.) β-reduction, in various forms, in particular the semantics given by the three families of in?nite λ-trees known as Böhm trees, Lévy-Longo trees, and Berarducci trees. Whereas

\begin{example}\label{ex:infinitary}
	Let $\dz:=\lam x. z(xx)$. In the  (infinite)  reduction sequence $\dz\dz\redb  z(\dz\dz) \redb z(z(\dz\dz)) \allowbreak\redb z(z(z(\dz\dz))) \dots$, the depth of the redex $\dz\dz$ tends to infinity.	
	It is  intuitively clear that $\dz\dz$ has an  \emph{infinite normal form} $z(z(z\dots))$.
\end{example}

\paragraph{Notation.}
From now on,  we use  
the following standard notations:  $\id=\lam x.x,  \quad \Delta= \lam x.xx$,
together with: $\dplus= \lam x. I \oplus (xx), \quad   \dout c:= \lam x.\print c(xx),  \quad \dz:=\lam x. z(xx) $.
% $\dplus= \lam x. I \oplus (xx) \quad   \dout c:= \lam x.\print c(xx)   \quad \dz:=\lam x. z(xx)$.
% The term $\Delta\Delta $ is the parametric diverging term. 
%	\[\id=\lam x.x  \quad \Delta= \lam x.xx   \quad  \true=\lam xy.x \quad  \false=\lam xy. y\]
%
%\[ \dplus= \lam x. I \oplus (xx) \quad   \dout c:= \lam x.\print c(xx)   \quad \dz:=\lam x. z(xx) \]

 \subsection{Motivations, and necessity, for non-deterministic evaluation.}
%\paragraph{Motivations, and necessity, for non-deterministic evaluation (\NDE).}
\label{sec:NDE}
  
In this paper we are concerned with evaluation towards a limit. 
%Because of this, we 
We allow the evaluation $\ered$ (the normalizing strategy)
to be \emph{non-deterministic}. 
Let us discuss the motivations.

A programming language which   is built on a $\lam$-calculus implements a specific evaluation strategy $\ered$ of the general reduction $ \red$. The evaluation strategy $\ered$ may or may not be deterministic, as long as \emph{all choices eventually yield the same result}. 
%
%Why not simply  fix    a deterministic  strategy? 
%This choice has several motivations. 
Non-deterministic evaluation (written \NDS) is a useful feature, which for example allows for parallel implementations, 
 but in some cases  is also a \emph{necessity}
 and a key  \emph{reasoning tool}, as we discuss.

\textbf{1.} \emph{\NDS subsumes different   evaluation policies}. A good  illustration   of this is in 
 Plotkin's Call-by-Value $\lam$-calculus, whose general reduction is $\redbv$.   \emph{Weak} evaluation   
%(where weak means no reduction 
(which does not reduce in the body of a function)  evaluates closed terms to  values.
There are  three main weak schemes (see \Cref{sec:weak}): reducing left-to-right, as defined by Plotkin  
\cite{PlotkinCbV}, right-to-left, as in  Leroy's ZINC abstract machine \cite{Leroy-ZINC}, %(resulting in a more efficient implementation), 
or in an arbitrary order. %, used for example  as a cost model in  \cite{LagoM08}. 
While left and right reduction are  deterministic, %and  incomparable,  
weak reduction  in arbitrary order is non-deterministic and \emph{subsumes} both. % thus a proof there establishes properties for all possible schedulings.
%This  insures the best of both world: one can reasons using any  sequential model, but also implement in a parallel way.  
%

\textbf{2.}  \NDS \emph{supports  parallel/distributed implementation}.
 {Non-deterministic evaluation} does not define an %\emph{actual} 
 abstract machine,
 but is %an inclusive definition of 
 includes \emph{all possible parallel implementations}. 

\textbf{3.} \emph{\NDE allows for  breadth-first scheduling.}  Left-to-right evaluation is inherently depth-first. \NDS allows for  breadth-first evaluation (favoring redexes  at minimal  depth),
which is a necessity when the reduction graph is infinitary.
% The key point here is that infinitary computation is the very way a result is produced. 
%Let us first  consider a
An example comes from  CbN $\lam$-calculus. Thinking of \Cref{ex:infinitary}, the terms
$  z(\Delta\Delta) (Iz) \tand z(\Delta\Delta) (\ddelta{z}\ddelta{z}) $ do deliver more information than $z(\Delta\Delta)(\Delta\Delta)$. 
Their respective infinite normal forms are   very different ($z \Omega z    \tand  z \Omega (zz\dots)$, respectively).
Still, for both left evaluation gets stuck at the leftmost  redex, $\Delta\Delta$.
%A similar  phenomenon appears with effectful computation: think of CbV $\lambda$-calculus with output, and  the terms 
Similar  phenomena appear with effectful computation: in  CbV $\lambda$-calculus with output of \Cref{ex:output}, the terms $ (\Delta\Delta)(\print 0 (I)) \tand (\Delta\Delta)(\ddelta{0}\ddelta{0}) $ %, whose behavior is similar to that of the previous terms.
{behave similarly} to the previous terms.
A breadth-first  approach  allows one to compute the ``best'' (in some sense) possible result  across all settings, uniformly.

\textbf{4. } \emph{\NDS facilitates reasoning and  proofs.}
 This point  
 is highly relevant when dealing with complex calculi, such as a  probabilistic $\lam$-calculus.
 Two examples from the literature are \cite{FaggianRonchi} and \cite{CurziP20}---in both cases moving from the usual deterministic 
  head reduction to its non-deterministic variant (given in \Cref{sec:head}) is crucial to the results. \SLV{}{We recall some of the issues in \Cref{app:proba}.}

 \SLV{}{
 \paragraph{Confluence is  not enough.}  
 Key to non-deterministic evaluation  is that despite the fact that there are many ways of evaluating
 a term,  \emph{all choices eventually yield the same result}. 
To this aim, confluence is not enough. 
 %Consider  a confluent reduction $\ered$; 
%Confluence  only implies  \emph{uniqueness} of normal forms. 
% If  strong normalization is guaranteed (for example, by a typing system) confluence implies (\#). 
% Otherwise,  
 The reduction of a term that has a normal form may still produce  diverging computations, which yield  \emph{no result } (think of $\beta$-reduction on the term $(\lam x.z)(\Delta\Delta)$).
 What we really want for a non-deterministic evaluation %$\ered$ 
 is that all its reduction sequences from the same   $\tm$ have the same behavior: if  $\tm$ has a  normal form, then 
 \emph{all} reduction sequences from $\tm$ eventually reach it (uniform normalization); 
 ideally, %all should do so 
 in the same number of steps. This latter  property  is known as 
  Random Descent \cite{Newman, Oostrom07,OostromT16}, and it is often  guaranteed in the literature of $\lambda$-calculus via a diamond-like property (see \Cref{sec:rewriting}).  These notions lift to the asymptotic setting.
 
}

% !TEX root = main_2022.tex

\section{Quantitative Abstract Rewriting Systems}
\label{sect:qars}\label{sec:QARS}
In this section we 
present   Quantitative Abstract Rewriting Systems (QARS) \cite{parsLMCS, AriolaBlom02}.
%which are a natural refinement of  ARSI \cite{AriolaBlom02}. 
QARS are Abstract Rewriting Systems (ARS) enriched  with a notion of observation, where we can  formalize   both finitary and asymptotic rewriting.    
We  first recall  some standard notions of rewriting (see  \cite{Terese03} or \cite{Book_AndAllThat}), in particular that of ARS and of \emph{normalizing strategy}.

\subsection{Basics in (Finitary) Rewriting}\label{sec:rewriting}

An \emph{abstract rewriting system (ARS)} is a pair $(\AA, \red)$ 
consisting of a set $\AA$ and a binary relation $\red$ on $\AA$ whose pairs are written  $t \to s$ and 
called \emph{steps}. %A \emph{$\red$-sequence} from $\tm$ is a sequence of $\red$-steps. 
We denote   $\red^*$ (resp. $ \red^= $, {$\red^+$}) the  transitive-reflexive  (resp. reflexive, {transitive}) closure of $\red$.
{We write $t \revred u$ if $u \red t$.}
%and use $\leftarrow$ for the reverse relation of $\red$, that is,  $\tm' \leftarrow \tm$ if  $\tm\red\tm'$.  
If $\red_1,\red_2$ are binary relations on $\AA$ then 
	$\red_1\cdot\red_2$ denotes their composition (\ie  $\tm \red_1\cdot\red_2 \tms$ if there exists  $\tmu\in \AA$ such that $\tm \red_1  \tmu \red_2 \tms$).
The relation $\red$  is \emph{confluent} if  
${\revred}^{*}\cdot {\red}^{*}
{~\subseteq~}{\red} ^{*}\cdot\, {\revred}^{*}$.
An element $u \in \AA$ is $\red$-\textbf{normal}, or  a $\red$-\emph{normal form} (\nnf) if there is no $t$ such that 
$u\red t$ (we also write $u\not\red$).

A  \textbf{$\red$-sequence} (or \textbf{reduction sequence}) from $\tm$ is a possibly infinite sequence $\tm = \tm_0, \ \tm_1, \ \tm_2, \dots$ such that $\tm_i\red \tm_{i+1}$. Notice that $\tm\red^*\tms$ holds exactly when there  is a \emph{finite} sequence from $\tm$ to $\tms$---we often write $\tm\red^*\tms$ to indicate a finite $\red$-sequence.
A  $\red$-sequence from $\tm$ is \emph{maximal} if it is either infinite or ends in a $\red$-\nnf.
We write $\seq \tm$ to  indicate a maximal $\red$-sequence from $\tm_0$; by convention, if $\tm_i  =\tmu\not\red$ then $\tm_k=\tmu$ for all $k\geq i$.

\paragraph{Normalization.}\label{sec:RD}
%Let  $(\AA,\red)$ be an ARS.
In general,   $\tm\in \AA$ may or may not reduce to a normal form. %Even 
And if it does, not all reduction sequences necessarily lead to normal form. 
%A term is %called 
%\emph{weakly} or \emph{strongly normalizing}, depending on if it  may or must reduce to normal form.
%If a term $\tm$ is strongly normalizing, any choice of steps will eventually lead to a normal form. 
%However, if $\tm$ is weakly normalizing, how do we compute a normal form? This is the problem tackled by \emph{normalization}
%:  by repeatedly performing \emph{only specific  steps},  a normal form will be computed, provided that $\tm $ can reduce to~any.
%%
$(\AA,\red)$ is strongly (weakly, uniformly) normalizing if every $\tm\in \AA$ is, where the    normalization notions are as follows.
%Let us recall three important notions of normalization:
\begin{itemize}
	\item $\tm$ is  \emph {strongly $\red$-normalizing}: every maximal $\red$-sequence from $\tm$ ends in a normal form; 
	%(\ie, $\tm$ has no infinite $\red$-sequence).
	\item 	 $\tm$ is   \emph{weakly $\red$-normalizing}: there is a $\red$-sequence from $\tm$ which ends in a normal form;
		\item $\tm$ is  \emph{uniformly $\red$-normalizing}:  $\tm$  weakly $\red$-normalizing implies $\tm$   strongly $\red$-normalizing.
\end{itemize}
Untyped $\lambda$-calculus is not strongly normalizing. 
How do we compute a normal form, or  \emph{test} if any exists? This problem is tackled by \emph{normalizing strategies}.
By repeatedly performing \emph{only specific  steps $\ered$}, we are guaranteed that  a normal form, if any, will eventually be computed.

A reduction $\ered$ is a  \emph{one-step} (resp.~\emph{multi-step}) \emph{strategy for $\red$}  if   ${\ered}\subseteq {\red}$  (resp. ${\ered}\subseteq {\red^+}$), and  it  has the same normal forms as $\red$. It is a \textbf{normalizing strategy} for $\red$ if, moreover,  whenever  $t$ has a  $\red$-normal form, then
\emph{every} maximal $\ered$-sequence from $t$ ends in a $\red$-normal form.
%
%$\ered$ is a  \emph{one-step} (resp.,\emph{ multi-step}) subreduction of  $\red$  if   $\ered\subseteq \red$  (resp, $\ered\subseteq \red^+$).
%It is a strategy for $\red$ if it has  the same normal forms as $\red$. It is \textbf{a normalizing strategy} if whenever  $t$ has a  $\red$-normal form, then \emph{every} maximal $\ered$-sequence from $t$ ends in a normal form.
%
%It is  is  a (one-step or multi-step) \textbf{normalizing strategy} for $\red$ if     it  has the same normal forms as $\red$ and whenever  $t$ has a  $\red$-normal form, then \emph{every} maximal $\ered$-sequence from $t$ ends in a normal form.
% 
%Let $\ered$ be a strategy for $\red$.
Note that %this definition takes into account that 
$\red$ may not have the property of unique normal forms.
\begin{remark}\label{rem:weakCBN}
	A familiar  example of  calculus where terms may not have a unique normal form is  Call-by-Name Weak $\lam$-calculus (weak means no reduction under $\lam$), studied by Abramsky and Ong \cite{AbramskyOng93}. The term $M=(\lam xy.x)(II)$ has two distinct normal forms, $N_1=\lam y.II$ and $N_2=\lam y.I$. Weak head reduction is a normalizing strategy for it. 
	However the strategy is not complete, in the sense that  it produces the normal form $N_1$, but it cannot reach $N_2$.
\end{remark}

A normalizing strategy $\ered$   need not   be deterministic (a reduction $\red$ is deterministic if for all $t\in \AA$ there is at most one $s\in \AA$ such that $t\red s$). 
%{($\ered$ is \emph{deterministic} if, for each $t\in \AA$, there is at most one $s\in \AA$ such that $t\ered s$)}.
However, %it must be
$\ered$ is required to be \emph{uniformly normalizing} on all elements $\tm$, \ie, all reduction sequences have the same behavior.

A  property of  $\ered$  which guarantees uniform normalization is Newman's \emph{Random Descent (RD)} \cite{Newman}:  for each $\tm\in \AA$, all maximal sequences from $\tm$ have the same length and---if it is finite---they all end in the same element.
The following  property suffices to establish it.
\begin{fact}[Newman]\label{fact:diamond} 	
	If reduction $\ered$  %has the following property, 
	is RD-diamond, then it  has Random Descent, where
	
	%\begin{equation*}
	\emph{RD-diamond}:\quad	$ 		( \tm_1 \ \xbackredx{\ex}{} \ \tm \ \ered  \  \tm_2) \mbox{ implies } (\tm_1=\tm_2 \mbox{ or }  \exists \tmu. ~\tm_1 \ \ered  \   \tmu  \  \xbackredx{\ex}{}  \ \tm_2) $.
	%\end{equation*}
	
	%then for each $\tm\in \AA$, all maximal sequences from $\tm$ have the same length; %and all end in the same normal form, if  any exists.
	%moreover, if the length is finite, they all end in the same element.
	
\end{fact}

\subsection{QARS} \label{sec:QARSdefs}
	Ariola and Blom \cite{AriolaBlom02} have introduced the notion of Abstract Rewrite Systems
	with Information content (ARSI); a rewrite system is
	associated with a partial order that expresses the ``information content'' of the elements.
	ARSI however are tailored   to  infinite normal forms in the sense of   B{\"o}hm and Levy-Longo trees: limits are there given by the ideal completion \cite[Prop. 1.1.21]{AmadioCurien} of the partial order.
	%Moving from partial orders to \emph{$\omega$-complete partial orders} ($ \omega $-cpo's) is enough to capture also effectful computation, such as the probabilistic one. It also open new questions, as we illustrate below.
	QARS \cite{parsLMCS}  move from partial orders to \emph{$\omega$-complete partial orders} ($ \omega $-cpos)---this is enough  to capture also effectful computation, such as the probabilistic one. We illustrate the key notions with several examples, including the calculi from \Cref{sec:examples}.

Computation is  a process that produces a result by gradually increasing the amount of  available information---the standard structure to express a result in terms of partial information is that of an $\omega$-cpo. 
Recall that a partially ordered set $\Set=(\Set,\leq)$  is
an \textbf{$\omega$-complete partial order}  (\textbf{$\omega$-cpo}) if  every $\omega$-chain $\cpo s_0\leq \cpo s_1\leq \dots$ has a supremum. We assume that $\leq$ has  a \emph{least element} $\bot$. 
The elements of $\Set$ are denoted by bold letters $\cpo s,\cpo p,\cpo q$.

Let $(\AA,\red)$ be an ARS.
With each $\tmt\in \AA$ is associated  a notion of  (partial) information, called \emph{observation},  by means of a function  from $\AA$  to an $\omega$-cpo.   \Cref{def:qars} formalizes this idea. 
\begin{Def}[QARS]\label{def:qars}
	A \emph{quantitative  ARS}  (\emph{QARS})  is an ARS $(\AA,\red)$ with a function  $\obs \colon \AA \red \Set$ (where  $\Set$ is  an $\omega$-cpo) such that for all $\tm,\tms\in \AA$, if $\tmt \red \tms$ then  $\obs (\tmt)\leq \obs (\tms)$.
%	\begin{center}
%		$\tmt \red \tms$ implies  $\obs (\tmt)\leq \obs (\tms)$.
%	\end{center}
	%	The notion  generalizes to a family of functions $\obs_\alpha: A \red \Set_\alpha$.
\end{Def}
Intuitively, the   function $\obs$ observes a specific property of interest about $\tm\in \AA$, and 
    indicates how much stable information $t$ delivers: the information content is monotonically increasing during computation.
%   A  reduction sequence which terminates---if any---should carry  a maximal amount of information. 
%
Notice that $\obs$ may take numerical values, but needs not.
\newcommand{\obsx}[2]{\obs_{#1}(#2)}%{\chi_{#1}(#2)}

%The following are some examples of QARS. The function $\obs$ needs not   take numerical values.
\begin{example}\label{ex:nfs}\label{ex:cpo_num} 
	\begin{enumerate}
		\item\label{p:ARSnf} $\lam$-calculus: let %$\Set$ the natural order on   $\{0,1\}$, and the indicator  function 
		$\Set = \{0<1\}$ and $ \obs_{n}(\tm) =1$ if $\tm$ is  normal, % $\obs (\tm) =0$ otherwise.
	$0$ otherwise.
	
	%	\item  Probabilistic $\lam$-calculus:  take  $\Set=([0,1], \leq_{\Real})$, and   $ \obs_{\pn} (\m)$ the   probability to be in normal form, 	as in  \Cref{fig:obs_proba}.
		\item  Probabilistic $\lam$-calculus:  take  $\Set=([0,1], \leq_{\Real})$, and for $\obs$ the  probability to be in normal form
		(we will formalize this in \Cref{sec:PLambda} , see  $\obs_{\pn} (\m)$ in  \Cref{fig:obs_proba}.)

			\item Infinitary $\lam$-calculus: take  $\Set = \Nat^\infty = \Nat \cup \{\infty\}$ with the usual order, and for $\obs$ the   function % $\ll$ 
			which associates with any term $\tm$ the minimal depth $k$ of any redex in $\tm$.% That is, the term $t$ is in normal form up to depth $k$.
		%	We develop this example  more formally in \refsec{llQARS}.
	\end{enumerate}
\end{example}

\begin{example}[Non-numerical $\obs$]\label{ex:cpo_nf}
\begin{enumerate}
	\item $\lam$-calculus: take for $\Set$ the flat order on normal forms, and define  $\obsN  (\tmu) =  \tmu$ if $\tmu$ is normal, 
	$\obsN( \tmu )=  \bot $ otherwise.
	
		\item Probabilistic $\lam$-calculus:  take  for $\Set$ the $\omega$-cpo of the subdistributions on normal forms $\DST{\Nnf}$
		(we will  formalize this 
		in  \Cref{sec:CbVproba}, see \Cref{fig:obs_proba}).

	\item Infinitary $\lam$-calculus: take the  $\omega$-cpo of the partial normal forms that are associated with $\lam$-terms 
(see \cite{AmadioCurien} page 52, and \Cref{sec:BT}).
\end{enumerate}
\end{example}

%For ARS, to reach a result %(normal form, value, head normal form)
 %is a yes/no property. For QARS, to reach a result is a \emph{quantitative} notion.
\newcommand{\QQ}{\mathcal Q}

\paragraph{Limits as   Results.}\label{sec:limits} From now on,  let $\QQ=( (\AA,\red),\obs )$ be an arbitrary but fixed QARS.
%
%\blue{Intuitively, a possibly infinite reduction  sequence $\seq \tmt$  from $\tm =\tm_0$ expresses a computation, whose 
 %\emph{result}  is   the maximal amount of information produced by the  sequence.}
%
By definition, given a $\red$-sequence $\seq \tmt$, its \emph{limit} $ \sup_{n} \{\w{\tmt_n} \} $ with respect to $\obs$
%\begin{center}
%	$ \sup_{n} \{\w{\tmt_n} \} $
%\end{center}
always exists,  because    $\Set$ is an $\omega$-cpo.
If $\red$ is deterministic---hence any $\tm$   has a unique maximal $\red$-sequence---it   is standard to interpret  the   limit as the \emph{meaning} of   $\tmt$.
%, however this definition  makes sense only when $\red$ is a deterministic reduction.
In a QARS,  $\tmt$ has \emph{several possible reduction  sequences}, and so can produce several outcomes (limits). Following \cite{pars}:
% whose set we denote $\wLim\seqt $.
 % If $\wLim$ has a greatest element, it means that \emph{the result is well defined}. 
\begin{Def}[$\obs$-limits] Let $\tmt \in \AA$. We write
%	  If exists   a $\red$-sequence $\seq \tm$ from  $\tmt\in \AA$   whose limit is $\cpo p$, we write
%	\[	t\red\conv{\obs} \bm{p}\]  	
\begin{itemize}
	\item 		$ 	t\red\conv{\obs} \bm{p} $, if  there exists   a $\red$-sequence $\seq \tm$ from  $\tmt$   whose limit $\sup_{n} \{\w{\tmt_n} \} = \cpo p$;

\item  $\wLim (t, \red) $ is the set $ \{\bm p\mid t\red\conv{\obs} \bm{p}\}$ of limits from $t$; % is  the set of limits of $t$ (w.r.t. $\obs$).
 
 \item ${\den t}$ denotes the greatest element of  $\wLim (t, \red)$, if it  exists.
\end{itemize}

The notations  omit the subscript $\obs$ when the function $\obs $ is clear from the context.
\end{Def}
%Informally, with $\tmt  $ is associated a well-defined result, which we denote $\sem \tmt$, if it is well defined the notion of greatest amount of information produced by any  reduction sequence.
  Intuitively, $\den \tmt$ is well defined if different reduction sequences from $\tmt$ do not produce \emph{essentially different} results: if  $\cpo q \not =  \cpo p$ then they  both approximate  a  same result $\cpo r$
(\ie, $\cpo q,  \ \cpo p  \  \leq \ \cpo r$).

Thinking of usual rewriting, consider  $\obsN$ as  in  \Cref{ex:cpo_nf}, point 1: here   to have a greatest limit  exactly corresponds to uniqueness of normal forms.
\begin{example}\label{ex:convergence}Let us revisit \Cref{ex:cpo_num} pointwise, using the same notations.
\begin{enumerate}
	\item $\lam$-calculus: 	consider $\tm = (\lam x. z)(\Delta\Delta)$. This term has  infinite  possible    $\redb$-sequences.
	The set of  limits w.r.t. $\obs_{n}$ contains two elements:  $\xLim{\obs_{n}}(\tm, \redb)=\{0, 1\}$
%	\begin{itemize}
%	 \item The set of $\nnf$-limits is  $\xLim{X}(\tm, \redb)=\{\bot, z\}$, which has a greatest element.
%	 		\item  The set of $X$-limits is  $\xLim{X}(\tm, \redb)=\{0, 1\}$.
%	\end{itemize}
	
\item Probabilistic  $\lam$-calculus: consider   the term $I\oplus \Delta\Delta$.  It  has only one reduction
sequence $\m = \mset{I\oplus \Delta\Delta}  \Red 	\mdist{
	\two I, \two \Delta\Delta}\Red	\mdist{
	\two I, \two \Delta\Delta} \Red \dots$. Here % $\m\Red\conv {\obs_{\pn}} \two$ and
 $\xLim{\obs_{\pn}}(\m, \Red)=\{\two\}$.

\item Infinitary $\lam$-calculus: consider the reduction sequence in \Cref{ex:infinitary}. The depth of the redex $(\dz\dz)$ tends to $\infty$, which is the limit. 
\end{enumerate}
\end{example}

Note that maximal elements of $\wLim(t, \red)$   need not  be maximal elements of $\Set$.   For instance, in \Cref{ex:convergence}.2, the term $I\oplus (\Delta\Delta)$ converges with probability $\two$ (rather than $1$). %; such terms are natural in an untyped setting like $\lam$-calculus. 
As a consequence,  \emph{the  set of limits may or may not %\footnote{\RED{We given an example in the Appendix.}}
	have maximal elements}.
 The fact that $\wLim(t,\red)$ may have a lub but not a maximum---similarly to $\Nat$ in $\Nat^\infty$ or the real interval  $[0,1)$--- is also easy to realize. \ArX{}{We give an example in \Cref{app:QARS}, see  \Cref{ex:nomax}.}

 Even if   $\wLim(\tm, \red)$  has maximal  elements,  a greatest limit  does not necessarily exist:  different reduction   sequences may lead to different limits. The probabilistic $\lam$-calculus and the $\lam$-calculus with output provide several  natural examples. 
  Point 2 in \Cref{ex:strings} below   shows moreover that the set of limits is---in general---\emph{uncountable}.
\begin{example}[Output $\lam$-calculus]\label{ex:strings} Consider  the   calculus  sketched in \Cref{ex:output}. 
Let	 $\OutC \alphabet=(\alphabet^*\times \Lambda_{\print{}},~\wred)$, where reduction  is  \emph{CbV  and weak}, with the obvious definitions.
%	 	%The pairs  $\str:M$ of a string on the alphabet $  \{0,1\} $ and a term in $\LambdaOut$,  as in \Cref{ex:output}. 
	Let  $\Set$ be the $\omega$-cpo of strings, and   let $\obs (\opair \str M)=\str$. 
	Clearly,  $(\OutC \alphabet, \obs)$ is  a QARS.
	%while $\obs_2({\str:M})$ is $1$ or $0$, depending if $M$ is a value or not. 
%
  \begin{enumerate}
	\item  Let $\m=\opair \epsilon {\print 0 (\id) \print 1 (\id)}$.  $\xLim{\obs}(\m,\wred)$ contains    two limits, $10$ and $01$, both  maximal, because  $\m \wred \opair 0 {\id \print 1 (\id)} \wred  \opair {10} {\id\id} $, but also 
	 $\m \wred \opair 1  {\print 0 (\id)  \id} \wred  \opair {01}  {\id\id} $. 
\item  Let $\m'= \opair \epsilon {M'}$  for ${M'= (\dout 0\dout 0)(\dout 1\dout 1)}$. This  produces all possible sequences on the alphabet $\{0,1\}$. So $\xLim{\obs}(\m',\wred)$ has uncountable many elements, all maximal.
\end{enumerate}
\end{example}
We are interested  in the case when a greatest limit exists. {The reason  is that if $\wLim(t,\red)$  has a sup $\bm s \in \Set$ which does not belong to $\wLim(t,\red)$, no reduction sequence converges to $\bm s$; that is, we cannot compute $\bm s$ {internally to the calculus}}.

\section{Strategies and Asymptotic Normalization}\label{sec:strategies}\label{sec:Acompleteness}
The question  of whether   the result $\sem \tmt$ of computing an element $\tmt$ is well defined is   natural. Equally natural is to wonder if there is a strategy that is guaranteed to compute  $\sem \tmt$. These two questions are at the core of this section.
%By definition, $\tmt $ has a \emph{greatest } element 	if and only if   $\sem \tmt$ is defined. 
%
The existence of unique normal forms is  independent of that of a normalizing strategy (see \Cref{rem:weakCBN}).
However, the computationally interesting  case is  (often) when both hold, so we will  focus on this case.

\newcommand{\bp}{\cpo \tmp}
\newcommand{\bq}{\cpo \tmq}
\newcommand{\br}{\cpo \tmr}
%

%\subsubsection{Asymptotic Completeness}\label{sec:Acompleteness}
 We say that a	 reduction ${\stred}\subseteq \red$ is (asymptotically) normalizing  if 
\emph{each}  ${\stred}$-sequence  from a given $\tmt$ converges \emph{maximally}. 
We decompose this property in two properties: completeness and uniformity, which we discuss after the formal definition.
\begin{Def}[Asymptotic properties] Given a QARS $ \brack{(\AA,\red),\obs}$, a subreduction ${\stred} \subseteq {\red}$ 
	 is  \textbf{asymptotically normalizing} for $\red$ (or $\obs$-normalizing)  if it is both asymptotically complete and uniform, where 
	\begin{enumerate}
		\item $\stred$ is  \textbf{asymptotically complete} (or $\obs$-complete)  if 
		\begin{center}
		$(\forall t\in \AA)$ : 	$\tmt \tolim{\red} \bq$ implies  $\tmt \tolim{\stred} \bp$	for some $\bp$ such that $\bq \leq \bp$;
		\end{center}
	
	\item $\ered$ is   \textbf{asymptotically uniform} (or $\obs$-uniform) if 
\begin{center}
	$(\forall t\in \AA)$ :	    all elements in $\wLim(t, \stred)$ are maximal in $\wLim(t, \stred)$.
\end{center}

	\end{enumerate}
All definitions adapt to $\ered$ multistep subreduction of $\red$.
\end{Def}
 Let us discuss all components, comparing with their  ARS  analog.
%	The above  provides 	an asymptotic analog for the notion of normalizing strategy. 
\begin{itemize}
\item \emph{Completeness} guarantees that the strategy $\ered$ is as good as $\red$ in the amount of information it produces.  
\item \emph{Completeness  is not enough:} an  asymptotically complete strategy is not guaranteed to find a/the ``best'' result: in \Cref{sec:PCbV} we will study a  reduction $\llRed$ which  is  complete, but need not converge to the greatest limit (\Cref{rem:complete}). Let us first see a classical example.
\end{itemize}

\begin{example} In  the usual $\lam$-calculus (as in \Cref{ex:convergence}.1), the  term  $M=(\lam x.I)(\Delta\Delta)$  has a $\redb$-sequence which reaches $I$, and a diverging one.   The leftmost-outermost strategy always produces $I$ (it is complete \emph{and} normalizing). Notice that  $\redb$ is \emph{trivially a complete strategy}  for $ \redb$,  but it is not  normalizing, because  $M$  has a diverging $ \redb$-sequence. Indeed, $\redb$ is \emph{complete, but  not uniform}.
\end{example}

\begin{itemize}
	\item \emph{Asymptotic uniformity} expresses that   all  $\stred$-sequences from a term behave the same way. 
This corresponds to the ARS notion of \emph{uniform normalization}: the reduction sequences from a term  either all diverge, or  all terminate (not necessarily in the same normal form).

\item  \emph{Normalizing strategies.} If we consider usual ARS, and assume   $\obs$ as in   \Cref{ex:cpo_num}.\ref{p:ARSnf}, expressing whether $\tm$ is or is not normal, then 
	  a strategy for $\red$ that is $\obs$-normalizing   is exactly a normalizing strategy for $\red$ in the  usual sense.
	
%		\item Assume  $(\AA, \red)$ is an ARS, and  $\obs_n$ as in   \Cref{ex:cpo_num}. point 1, expressing that $\tm$ is or not normal.
%		Then  a strategy for $\red$ which is $\obs_n$-normalizing   is exactly a normalizing strategy for $\red$ in the  usual sense.
%		

\end{itemize}

%Our general method here is to reduce the test for a property of $\red$, to a test on a  subrelation $\stred\subseteq \red$, which is simpler to study. 
If  $\stred \subseteq \red$ is $\obs$-\emph{complete},  then  $\wLim(t, \red)$  has  maximal elements (resp. a greatest element) if and only if 
$\wLim(t, \stred)$ does. So we can reduce 
testing  such  properties for  $\red$, to  testing  the same properties for $\ered$, which is  often  simpler to study. 
In particular, if we are able to find  a reduction   $\stred\subseteq \red$ which is complete and moreover  has a unique limit, then necessarily  $\red$  has a greatest limit. That is, we can  simultaneously answer both of our questions: whether 	$\den t$  is well defined, and if some  strategy is guaranteed to compute it.
%if there is a strategy which is guaranteed to compute it.

\begin{prop}[Main, abstractly]	\label{thm:main}	
	If the following hold
	\begin{enumerate}
		\item[i.] $\stred$ is asymptotically complete for $\red$;
		\item[ii.] $\wLim(\tmt,\stred)$ contains a unique element (\ie $\wLim(\tmt,\stred)=\{\cpo p\}$, for some $\cpo p$).
	\end{enumerate}
Then: (1.)  $\sem \tmt$ is  defined, and (2.) $\tmt \tolim \stred  \sem \tm $, for each $\ered$-sequence.
\end{prop}
 Notice that condition (ii.)  means that   \emph{all} $\stred$-sequences from the term  $\tmt$ have the \emph{same} limit.
 
\begin{remark}[Asymptotically normalizing strategies]\label{rem:Anormalizing}
	If   a QARS is such that 
	$\sem \tm$ is defined for each $\tm$, then the two notions---to be an $\obs$-normalizing strategy and to satisfy the conditions in \Cref{thm:main}---coincide. Indeed, any $\obs$-normalizing strategy for $\red$, if it exists,  is forced to have a \emph{unique} limit,  {that is, $\wLim(\tmt, \stred) = \{\sem \tmt\}$. }
\end{remark}

\subsection{A proof technique for Asymptotic Normalization}\label{sec:tools}
The two conditions in \Cref{thm:main} give  a method to prove normalization. The crucial step is to prove  asymptotic completeness. Remarkably, as we show in this section, 
this  can be  reduced to prove a \emph{finitary} property (\emph{factorization}) and an elementary  one-step test (\emph{neutrality}). 

The other condition in \Cref{thm:main}, namely 
 uniqueness of limits, is trivial if the strategy is deterministic. Otherwise,  random descent (opportunely formulated \cite{parsLMCS}) is a property that guarantees it, and that can also be established via a local test, as  we recall below. While  it is only a sufficient  criterion, it often suffices to deal with non-deterministic evaluation strategies in $\lam$-calculus, and in particular it suffices to deal with strategies in probabilistic $\lam$-calculus.

%As may be expected, the proof of such \emph{infinitary properties} reduces to the proof of \emph{finitary} ones. 
%Remarkably while such properties are  \emph{ infinitary properties}, to establish them  it suffices to prove a set of properties which are not only finitary but local ---in the sense that we only need tests on one-step reduction . 
% 
%First, we reduce  \emph{two infinitary properties} to a finitary property of $\red$ (factorization) and two local properties, \ie properties which can be checked in one-step reduction.
%We then  go even further: 
%we  show that to  study  a compound system such as a $\lam$-calculus endowed with an extra operator -- as is the case of probabilistic $\lam$-calculus--  to test  factorization only requires to check an elementary, local diagram (\refsec{factorization}). 
%
%
%Uniqueness of the limit can be established via a local,  diamond-like test. 
%
%Asymptotic completeness is reduced to a local test (neutrality), and a finitary property (factorization). In turn,  the latter can be established 

\paragraph{Asymptotic Completeness via Factorization.}\label{sec:ACompl}

%In $\lam$-calculus, standardization provides  a strategy which is complete for all reduction
%sequences, i.e., for every reduction sequence $t \red^* s$, there
%is a standard reduction sequence from $t$ to $s$. Here we are interested in the asymptotic analog of this property: if  
%$ t \red\conv{\obsa}  \bm p$,  then there is a strategy which from $t$ converges to $\bm p$.

%Remarkably, \emph{finitary} properties allow us to establish \emph{asymptotic} completeness. 

The following theorem assumes a partition of the $\red$-steps into two classes: essential steps $\ered$ and internal steps $\nered$.
Point (i) states that every sequence $\red^*$ factorizes into a $\ered$-sequence followed by a $\nered$-sequence.
Point (ii) states that the internal steps $\nered$ do not increase the information content.

\begin{theorem}[Asymptotic completeness criterion]
	\label{thm:w-essential}\label{thm:ACompl}Given   $((A,\red), \obs )$ a  QARS, and a subrelation  $ \ered \subseteq \red$, assume :
	 %  $\red = \ered \cup \nered$ and
\begin{enumerate}
	\item[i.]  \emph{$\ex$-factorization}:  if $\tmt \to^* \tmu$
	then $ \tmt \ered^*\cdot \nered^*\tmu$;
	\item[ii.] \emph{$\nex$-neutrality}:    $\tmt \nered \tms$ implies  $\w{\tmt}=\w{\tms}$.
\end{enumerate}	
Then: % for each $\tm$:
\quad \quad
  $\tm \red\conv{\obs}  \cpo p  \timplies  \tm \ered\conv{\obs} \cpo p.$

\end{theorem}

\begin{proof}

Let $\seq \tmt$ be a $\red$-sequence such that $t=t_0$ and  $\sup_{n} \{\w{\tmt_n} \} = \cpo p$. From $t$, we  inductively build a $\ered$-sequence $\seq s$ with $s_0=t$ and such that, for every $k \in \Nat$, there is an index $j(k)$ such that $t\ered^* s_{j(k)}$ and $s_{j(k)}\nered^* t_{k}$.
Case $k=0$ is trivial (set $s_{j(0)} := t$). 

Assume  the claim holds for $k\geq 0$, so  $t\ered^* s_{j(k)}$. Observe that we have a sequence $s_{j(k)} \nered^*t_k\red t_{k+1}$. By applying assumption (i.) to it, we have 
$s_{j(k)} \ered^*   u \nered^* t_{k+1}$. We concatenate  $t\ered^* s_{j(k)}$ and $s_{j(k)} \ered^*   u $ to obtain
 $t\ered^* s_{j(k)} \ered^*  s_{j(k+1)}:= u $, as desired.
By assumption (ii.),  $s_{j(k)}\nered^* t_{k}$ implies    $\w {t_k }=\w {s_{j(k)}}$. 
The claim easily  follows.
\end{proof}

\paragraph{Uniqueness of the  limit  via  Random Descent.}\label{sec:wRD}

To establish that a strategy has a unique limit, 
  Random Descent \cite{Newman,Oostrom07,OostromT16}   has already been shown to adapt well and naturally in a probabilistic and asymptotic  setting \cite{pars, parsLMCS}. 
 
 The property  $\obs$-RD  below states that if  $\tm$ has  different reduction sequences, they are all \emph{indistinguishable} if regarded through the lenses of $\obs$.  Namely,  all reduction sequences $\seq \tmt$ starting from $\tmt$ induce the same  
 $\omega$-chain  $\langle \obs(t_n)\rangle_n$. Thus, they all have the same $\obs$-limit.

\begin{Def}[Weighted Random Descent]\label{def:RD} \label{def:WRD}
	Let $((\AA,\red),\obs)$ be a QARS. The relation $\ered \subseteq \red$   satisfies the following properties if they hold for each $\tm \in \AA$. 
	\begin{enumerate}		
		\item {$\obs$-RD}:  for each pair of $\ered$-sequences   $\seq r$, $\seq s$ from  $t$,  $\w {r_n} = \w {s_n} $ for all  $n$.
		%
	%	\item \textbf{local} $\obs$-RD:   if  $r_1 \leftarrow c \red s_1$, then there exists a pair of sequences $\seq r$ from $r_1$ and 
	% $\seq s$ from $s_1$ such that 	$ \w {s_n}= \w{r_n}$, ~$\forall n$.
		
		\item $\obs$-diamond: $\ered$ satisfies RD-diamond, and if  
			$\tm  \xbackredx{\ex}{} m \ered s$ then $ \w {s}= \w{t}$.	
	\end{enumerate}
\end{Def}
\begin{prop}[\cite{parsLMCS}]\label{thm:diamond}With the same notation as in \Cref{def:WRD}:\\
	($\obs$-diamond) $ \Rightarrow $ ($\obs$-RD)  $ \Rightarrow $  $\wLim(\tmt,\ered)$ contains a unique element.
\end{prop}

%We give some concrete examples in  \Cref{app:WRD}.

\begin{example}[CbV Weak reduction] Let us consider  Call-by-Value $\lam$-calculus with   weak  reduction $\wred$, where weak means no reduction in the scope of $\lam$-abstractions.
	The following are two different  $\wred$-sequences from the term $(II)(Ix)$:\\
	$(II)(Ix) \ \wred \ I(Ix) \ \wred \ Ix  \ \wred \ x ~~~\tand ~~~
	(II)(Ix) \ \wred \ (II)x \ \wred \ Ix \ \wred \ x.$
	
	The observations of interest are values. Let $\obs_v: \Lambda \to \{0,1\}$	be $1$ if the term is a \emph{value} (i.e. a variable or an abstraction), $0$ otherwise. Through the lenses of $\obs_v$, both sequences appear as
	$ \langle	0,~ 0 ,~ 0,~  1 \rangle $.
	%	\begin{enumerate}
		%		\item $(II)(Ix) \wredbv I(Ix) \wredbv  Ix \wredbv x$
		%		\item $(II)(Ix) \wredbv (II)x \wredbv Ix \wredbv  x$
		%	\end{enumerate}
	
	%	 Despite the fact that the choice of redexes is different, both have the same length, and both end in the same value $x$.
\end{example}

\section{Normalization in CbV and CbN \texorpdfstring{$\lambda$}{lambda}-calculi}
In the rest of the paper, we study asymptotic normalization in the setting of $\lam$-calculi---in particular we are interested in probabilistic  $\lam$-calculus (\Cref{sec:PLambda}). 

In this section, after recalling the general syntax of $\lam$-calculus, 
 we  define  a novel, flexible  normalizing  strategy, which is uniformly defined  for \emph{Call-by-name} (\CbN) and \emph{Call-by-Value} (\CbV)  $\lam$-calculi.
Its features---in particular the fact that it support breadth-first reduction---make it suitable to then be extended to asymptotic normalization, in different settings.

\subsection{Call-by-Name and Call-by-Value (applied) $\lambda$-calculus }
\label{sec:lambda}

%\subsubsection{Syntax}
We recall the basics of $\lam$-calculus. Our syntax admits    %constants and 
\emph{operator symbols} \cite{HindleySeldin86, PlotkinCbV}, \ie~\emph{constants} with a fixed arity for their arguments.
%\subsubsection{ $ \lam $-calculus }\label{sec:lambda} 
\emph{Terms} and \emph{values} are defined by  the grammars below.
%\footnote{For compactness we omit 0-ary constants in this paper. They are straightforward to add.}.
			\[	\begin{array}{rcllr}
			%	\cc & ::= & \hole{~}   \mid M\cc\mid \cc M \mid \lambda x.\cc 			& (\textbf{contexts})\\		
			M & \Coloneqq & x   \mid  \lambda x.M \mid MM  \mid  \op(M,\dots,M) & ( \text{\emph{Terms,} } \LambdaOp  )\\		
			V & \Coloneqq & x \mid \lambda x. M & (\text{\emph{Values,} } \Val)\\
				\end{array}\]
where $x$ ranges over a countable set of \emph{variables},  and 
$\op$  over a disjoint (possibly empty)
set $\OpSet$ of operator symbols. \SLV{If $\OpSet$ is empty, the calculus is \emph{pure} and we set $\Lambda \defeq \Lambda_{\OpSet}$.}{
	If the set of constants is non empty, the calculus is called \emph{applied}, and  the set of terms is often  indicated as $\Lambda_{\OpSet}$.
 Otherwise, the  
calculus is  \emph{pure}, and the sets of terms is   $\Lambda$. }
Terms are identified up to renaming of bound variables, where $\lambda x$ is the only binder constructor. % with abstraction being the only binder. 
 $ P \subs x Q  $ is the capture-avoiding substitution of $Q$ for the free occurrences of $x$ in $P$.

\textbf{Contexts} (with  an hole $\hole{\ }$) are defined by the grammar below. $\cc\hole{N}$ 
stands for the term obtained from $\cc$ by replacing the  hole with $N$ (possibly capturing the free variables~of~$N$).
\[\cc ::= \hole{~}   \mid M\cc\mid \cc M \mid \lambda x.\cc \mid \op(M, \dots, \cc,\dots, M) \qquad (\textit{Contexts}) \]
%\begin{center}
%	{\footnotesize {\begin{minipage}[c]{0.45\textwidth}
%				$
%				\begin{array}{rcllr}
%				\cc & ::= & \hole{~}   \mid M\cc\mid \cc M \mid \lambda x.\cc \mid \op(M, \dots, \cc,\dots, M)
%				& (\textbf{Contexts})\\		
%				\end{array}
%				$
%	\end{minipage}}}
%\end{center}
% 

\paragraph{Rules and Reductions.}

A \textbf{rule}  $ \Rule$ is  a binary relation on $\LambdaOp$, which we also denote   $\Root{\Rule}$, writing  $R \Root{\Rule} R'$. $R$ is called a $\Rule$-\emph{redex}.
	%rather than $(\tm,\tm')\in \rho$.
The best known rule is $\beta$:  $(\lambda x.M)N ~\mapsto_{\beta}~ M\subs{x}{N} $.

A \textbf{reduction step}  $\red_{\Rule}$ is %the binary relation on $\LambdaOp$ defined as
the  closure under context $\cc$ of $\Rule$.
\SLV{}{Explicitly, $T \red T'$
holds if  $T = \cc\hole{R}$, $T' = \cc\hole{R'}$,  and  $ R \Root{\Rule} R'$.
%The term $R$ is called a $\Rule$-\emph{redex}. 
%The set of   $\Rule$-\emph{redexes} is denoted by $\R_{\Rule}.$
}

%\paragraph{Beta Rule}The reduction $\redb$ is the closure under context $\cc$ of the  beta rule $\mapsto_{\beta}$ (Fig.~\ref{fig:brules}).

%Operators are intended to model some  form of operation, as soon as rules are specified; with slight abuse of language we call 
%\emph{\oterm} any  term of shape $\op (M_1, \dots, M_n)$.

\newcommand{\Lcbn}{\Lambda^\cbn}
\newcommand{\Lcbv}{\Lambda^\cbv}

\SLV{\paragraph{CbN and CbV Calculi.}}{\paragraph{CbN and CbV Calculi.}}
%The key rule in $\lam$-calculus is the $\beta$-rule.
The  (pure) \textbf{Call-by-Name} calculus  $\Lcbn=(\Lambda,\redb)$ is the set of terms equipped with the contextual closure of the $\beta$-rule, as described \eg  in \cite{Barendregt}.
% ( Fig.~\ref{fig:brules}).
 %
The (pure) \textbf{Call-by-Value} calculus $\Lcbv=(\Lambda,\redbv)$ is the same set   equipped with the contextual   closure  of the $\betav$-rule: 
$(\lambda x.M)V ~\mapsto_{\beta_v}~ M\subs{x}{V} ~ \mbox{ where } V \in \Val$, as introduced by Plotkin \cite{PlotkinCbV}.

%\paragraph{Operators rules.} 
%To endow  $\lam$-calculus with

\textbf{\CbN and \CbV applied} calculi are obtained  by associating 
to operators  (the contextual closure of) a family of rules of the form 	 
\SLV{$\op (M_1, \dots, M_k)\mapsto_{\op} N$.}{\begin{center}
	 $\op (M_1, \dots, M_k)\mapsto_{\op} N$.
\end{center}}
This is a standard way to enrich $\lam$-calculus with new computational features, such as probabilistic choice or output.

%With slight abuse of language we call \emph{\oterm} any  term of the form $\op (M_1, \dots, M_n)$.
%where to $ \op (M_1, \dots, M_k) $ one can associate several rules. %with a single operator $\op$. 
%In this section, we  do not consider any specific rule associated with any operator.
%We will see examples of such  rules (modeling a choice) in the next sections. 
% 

\SLV{}{
\paragraph{Restricted Reductions.}
The contextual closure of rules may be restricted. 
% Given a rule  $\mapsto_{\gamma}$, we  denote  its  closure   under a specific context $\ss$ by  $\xredx{\ss~}{~\gamma}$.
For example, in CbN,  \emph{head reduction}  is the closure of the $\beta$ rule under head context, \emph{leftmost-outermost} reduction
only reduces the leftmost-outermost redex. We recall  the non-deterministic variants of weak and head reduction.
}

%\subsection{Subreductions: head and weak}

\SLV{\paragraph{Weak reductions in CbV.}\label{sec:weak}}{\subsubsection{Weak reductions in CbV}\label{sec:weak}}
\SLV{In CbV $\lam$-calculus,  various restrictions of $\redbv$   are studied.
	If the result of interest are values, the  
	reduction is \emph{weak}, that is, it does not reduce in the body of a function. There are  three main weak schemes: left, right and in arbitrary order. 
}{
In the literature on the CbV $\lam$-calculus,  various restrictions of $\redbv$   are studied.
If the result of interest are values, the  
reduction is \emph{weak}, that is, it does not reduce in the body of a function. 
As we mentioned, there are  three main weak schemes: reducing from left to right, as originally done by Plotkin 
\cite{PlotkinCbV}, from right to left, as in  Leroy's ZINC abstract machine \cite{Leroy-ZINC}, or in an 
arbitrary order, used for example in  \cite{LagoM08}.
Left and right reduction  are  \emph{deterministic}.  Weak reduction  in arbitrary order   \emph{subsumes} both. It is non-deterministic in the choice of the redex, but such a choice is irrelevant w.r.t. reaching a value and the number of steps to do so,   because $\wredbv$ satisfies the diamond property  of  \Cref{fact:diamond}.
{This  insures the best of both world: one can reasons using a (any) sequential model, but implement in a parallel way.}
}
\emph{Left} contexts  $\leftc$, \emph{right} contexts $\rightc$, and   (arbitrary order) \emph{weak} contexts $\ww$ 
are   defined by
\[	\leftc ::= \hole{~}  \mid  \leftc M \mid  V \leftc   \quad\quad
	\rightc ::= \hole{~}  \mid   M \rightc \mid   \rightc V  \quad\quad
	  \ww ::= \hole{~}  \mid  \ww M \mid   M \ww \]
 Given a rule  $\Root{}$  on $\Lambda$,
\emph{weak  reduction} 
$\xredx{\weak}{}$  is the closure of  $\Root{}$ under context $\ww$.
A step $T\redx{} S$ is  \emph{non-weak}, noted 
$T\nwredx{} S$  if it is not weak.
Similarly for left  ($\xredx{ \lsym}{}$  and $\xredx{\neg \lsym}{}$), and right  ($\xredx{ \rsym}{}$  and~$\xredx{\neg \rsym}{}$).
\SLV{}{\begin{remark}$\lred \subsetneq \wred$ and $\xredx{\rsym}{} \subsetneq \wred$. For example $(xx)(II)\wred (xx)(I) $ but $(xx)(II)\not \lred$. % Indeed $\wred$ \emph{subsumes} both, left and right steps.
%	Similarly,  $N=(xx)(\Delta\Delta)\wred (xx)(\Delta\Delta) $  while $N$ is $\lsym$-normal.
\end{remark}}
\SLV{Left and right reduction  are  \emph{deterministic}.  
	Reduction $\wredbv$   \emph{subsumes} both. The choice of a redex is 
	 non-deterministic, but  irrelevant w.r.t. reaching a value and the number of steps to do so,   because $\wredbv$ is RD-diamond  (\Cref{fact:diamond}). 
	We can fire any arbitrary  redex in weak position---or  \emph{all of them} {in parallel}.  A \emph{parallel variant} can easily be defined. % \RED{(see \Cref{app:normalization} in the Appendix)}.
}

Weak factorization  holds for the three reductions: 
$\redbv^*   \ \subseteq   \  \sredbv^*  \cdot \nsredbv^*$, for   $\surf\in \{\weak,\lsym,{\rsym}\}$. 

\SLV{
}{
\paragraph{Weak Factorization.} 
 Let $\surf\in \{\weak,\lsym,{\rsym}\}$
	\begin{itemize}
		\item \emph{weak factorization of $\redbv$:}  ~~	 $\redbv^*   \ \subseteq   \  \sredbv^*  \cdot \nsredbv^*$.
		\item \emph{$\Val$-completeness:}~~   $T \redbv W (W\in \Val)$ 
		if and only if $T\sredbv^*V$ ($V\in \Val$)
	\end{itemize}
}

\SLV{}
{\paragraph{Parallel Weak Reduction.} With $\wred$, any  redex in weak position  can be fired---or we can fire \emph{all of them} in parallel. A parallel variant can easily be defined as in \Cref{fig:weak_parallel} and shown to be sound and complete.
For simplicity of definitions, we only consider pure CbV and closed terms, so that the normal forms of $\wredbv$ are exactly the closed values.
Rule $val$ makes the relation reflexive on values and \emph{only} on values---this is just   an harmless shortcut in order to give a compact and neat formulation.
%it  can be replaced by giving three subcases for the application: $P_1P_2,~ VP_2, ~P_1V$. 

	\begin{figure}
\small{
	\[\infer[rdx]{M \xredx {\pw}{} M'}{M\Root{\betav} {M'}} \qquad
		\infer[val]{ V \xredx{\pw}{} V}{} \qquad
%	\infer[wp2]{M \xredx{\pw}{} M}{M \not\wred} \quad
  \infer[app]{ P_1P_2 \xredx{\pw}{}   P_1'P_2'}{ P_1P_2\not\Root{\betav}  & P_1\xredx{\pw}{}   P_1' & P_2\xredx{\pw}{}   P_2'} \]
}

%\footnotesize{
%	\[\infer[]{M \xredx {\pw}{} \brang 1 \ M'}{M\Root{\betav} {M'}} \quad
%	\infer[]{ V \xredx{\pw}{}\brang 0 \ V}{} \quad
%	%	\infer[wp2]{M \xredx{\pw}{} M}{M \not\wred} \quad
%	\infer[]{ P_1P_2 \xredx{\pw}{}\brang{n_1+n_2}  P_1'P_2'}{ n_1+n_2 >0  & P_1\xredx{\pw}{}\brang{n_1}  P_1' & P_2\xredx{\pw}{} \brang{n_2}  P_2'} \]
%}	
	%	
%			\[\infer{M \xredx {\pw}{} M'}{M\Root{} {M'}} \quad
%		\infer{P_1P_2 \xredx{\pw}{}   P_1'P_2'}{ P_1\xredx{\pw}{}   P_1' & P_2\xredx{\pw}{}   P_2'}\quad
%		\infer{P_1V \xredx{\pw}{}   P_1'V}{ P_1\xredx{\pw}{}   P_1'}\quad
%		\infer{VP_2 \xredx{\pw}{}   VP_2'}{ P_2\xredx{\pw}{}   P_2'}
%		\]
%\caption{Parallel weak reduction in $(\Lambda,\redbv)$. $M$ closed term.}
\caption{$(\Lambda^\bullet,\redbv)$: Parallel weak reduction on closed terms.}\label{fig:weak_parallel}
	\end{figure}

	\begin{prop}$M \wredbv\tr V$  $\Leftrightarrow$ $M\xredx{\pw}{\betav}\tr V$.

	\end{prop}
\begin{proof}$ \Rightarrow $: by induction on the number of   $\wredbv$ steps, observing that if $P_1P_2\wredbv^* V$, then necessarily 
	$P_1P_2\wredbv^* V_1V_2$.	
	
		  $ \Leftarrow $:  immediate.
	
\end{proof}
}

\newcommand{\hBar}{\mathfrak{h}}
\SLV{\subparagraph{Head reduction in CbN.}\label{sec:head}}{\subsubsection{Head reduction in CbN}\label{sec:head}}
Head reduction \cite{Barendregt}\SLV{}{---here written $ \xredx{\hBar}{\beta}$---} is the closure of $\beta$ under head context
$\lam x_1\dots x_n.\hole {~}M_1\dots M_k $. 
\emph{Head normal forms} (\hnf), whose set is  denoted by $\Hnf $, are its   normal forms.
The literature of linear logic often uses a variant of head context which  includes  the standard one, and induces  exactly the same set $\Hnf$ of  normal forms. Given a rule ${\Rule}$, we write $\hredx{\Rule}$ for its closure under context $\hh$.
\[\hh ::= \hole{~}  \mid  \lambda x.\hh \mid  \hh M \qquad (\emph{Head contexts}) \]
\SLV{}{Clearly, $ \xredx{\hBar}{\beta} \subseteq \hredb$. Since 
$\hredb$ has the diamond  property of \Cref{fact:diamond}, and since $\hredb$ subsumes ``usual'' head reduction, 
there is no essential difference, neither in reaching a \hnf nor in the number of steps to do so.}
\SLV{Head factorization (see %$\redb^*   \ \subseteq   \  \hredb^*  \cdot \nhredb^*$, 
	\cite[Lemma 11.4.6]{Barendregt})  and head normalization (see \cite[Thm.~8.3.11]{Barendregt}) are classical results, which hold  also when the calculus includes constants, \ie for $(\Lambda_{\OpSet},\redb)$.}{
	\begin{remark}
		$ \xredx{\hBar}{\beta}    \subsetneq \hred$. For example $(\lam x.Ix )z \hredb (\lam x.x)z$ but $M=(\lam x.Ix )z  \not{\xredx{\hBar}{\beta}} (\lam x.x)z$
	\end{remark}
	%As we discussed in the Introduction, $\hred$ allows for simpler proofs in the probabilistic case  \cite{FaggianRonchi}  \cite{CurziPagani}.

	%\RED{given a $\beta$-redex $(\lam x. P) Q$, spine head reduction first evaluates the body of $P$ and then evaluates the outermost redex.}
	
	\paragraph*{Head Factorization}
	Every $\redb$-reduction sequence can be re-organized/factorized
	as to first reducing  redexes in head position, and then everything else. Head  factorization \cite[Lemma~11.4.6]{Barendregt} and head normalization are classical results, which hold  also when the calculus includes operators, \ie for $(\Lambda_{\OpSet},\redb)$.
}
\SLV{}{
	\begin{theorem}[Head Factorization] \hfill
		\begin{itemize}
			\item \emph{Head  Factorization:}  ~~	 $\redb^* \ \subseteq \ \hredb^*  \cdot \nhredb^*$.
			\item \emph{Head Normalization:} $M$ has \hnf if and only if \\ $M\hredb^* S~ (\text{for some }S \in \Hnf)$ .
		\end{itemize}
	\end{theorem}
	%Such results also hold when the calculus includes constants, and so for  $(\Lambda_{\OpSet},\redb)$.
	%\footnote{In CbN, $\op(P_1,cdot,P_k)$ can be seen as sugar for $\op P_1\dots P_k$, where  $\op$ is simply a constant.}.
}

% !TEX root = main_2022.tex
\subsection{A strategy for finitary  normalization in CbV and CbN \texorpdfstring{$\lambda$}{lambda}-calculus} 
\label{sec:normalization}
We revisit normalization for $\lam$-calculus---uniformly for  CbV and CbN ---and define a  strategy which is well-suited to be  extended to  probabilistic $\lam$-calculi, and  to asymptotic normalization. It  supports non-deterministic head and weak reduction
 %the non-deterministic variants of head and weak reduction
(as needed in  the probabilistic case) and breadth-first evaluation of redexes (as needed to deal with infinitary reduction graphs).

\SLV{}{\paragraph{From  Surface Normal Forms to Normal Forms}} 
\SLV{We call \emph{surface reduction} weak reduction in CbV and  head reduction in CbN, because they only fire redexes at \emph{depth 0},
%	 (\ie outside arguments for CbN and ouside abstractions for CbV).
where in CbV the  depth of a redex $R$ is  the number  of \emph{abstractions} in which $ R $ is nested, and  in CbN is the number of arguments.}{}
\SLV{}{We call head reduction in CbN and weak reduction in CbV \emph{surface}, because they only fire redexes at \emph{depth 0}, where  the CbN (resp. CbV) depth  of a redex $ R $ in a term is the number of \emph{arguments } (resp. \emph{abstractions}) in which $ R $ is nested.
If the language of terms includes operators, each nesting in the scope of an operator also increases the depth, matching the intuition that the evaluation of a term $\opp{p_1,\dots, p_k}$ first performs the operation $\op$.
Fixed a surface reduction $\sred$, we call \emph{surface normal forms } its normal forms, and write $\Snf$ for their set. 
Note that in CbV the values are exactly  the surface normal form of closed terms, but in general  $  \Val \subsetneq \Snf $ (for example $(x(\lam z.\Delta z)) (\lam y.\Delta I) \in \Snf$).
When interested in normal forms, we cannot restrict the attention only to closed terms, because  the evaluation process will need to act  in the scope of abstraction. For example, the  body $\Delta z$ of $ (\lam z.\Delta z)  $ is an open term.}%
%
%A $\beta$-normal form (resp. a $\betav$-\nnf) can be computed by iterating the  relevant surface reduction.
Normal forms for $\beta$ and $\betav$ can be computed by iterating surface reduction in a suitable way, as~we~show~below.
	\SLV{}{Head and weak normal forms can so be seen as  "stage one"  in the process leading to the full normal form.}

\paragraph{Normalizing strategies.}

In $\Lambda^\cbn$, a paradigmatic normalizing strategy is leftmost-outermost reduction. 
It can be described as: first apply head reduction $\hredb$ until  \hnf,   and then  iterate the process, in left-to-right order.   
Normalization in $\Lambda^\cbv$ is less established: one can iterate $\lredbv$ left to right 
 (as in Plotkin's   standard reduction \cite{PlotkinCbV}), but also iterate  $\xredx{\rsym}{\betav}$  right to left, as in  
 Gr{\'{e}}goire  and Leroy's implementation \cite{GregoireL02}.
In all cases, once  a head or weak normal form is reached (think of  $xM_1\dots M_k$ in CbN) no interaction is possible among the subterms $M_i, \dots, M_k$, so in fact  the process can be iterated in any \emph{arbitrary order}.

We define  
a rather liberal normalizing 
strategy, 
\emph{uniformly} for CbN and CbV, and \emph{parametrically} in the choice of surface reduction. Unlike leftmost-outermost reduction, which is sequential and \emph{inherently depth-first},  the \emph{unbiased} reduction $\lsred$ is non-deterministic in the choice of the outermost redex, and can  support a breadth-first reduction policy.
%$\osym    \los$
%
	It  persistently performs surface steps, as long as it is possible, and then iterates the process in the subterms, in 	{arbitrary} order.
%	Said differently, it can reduce \emph{any} redex that is not contained in other redexes.
	 %Intuitively, it selects a redex at minimal {depth}.
%	\SLV{It is \emph{unbiased} w.r.t. the order, and we denote it $\lsred$.}{}
	
	 % It iterates  $\sred$-reduction in  \emph{unspecified} order.

%	The reduction   $\lsred$ persistently performs surface steps, as long as it is possible, and then iterate the process in the subterms, in 	\emph{unspecified} order. Intuitively, it chooses one (any) outermost redex

%\begin{figure}
%	\begin{tabular}[c]{lll}
%	CbN:&$ \betab \eq \beta  $&  $\surf \eq \head$\\
%	CbV: & $\betab \eq \betav$ &  $\surf \in \{  \weak, \lsym, \rsym \}$ 
%\end{tabular}	
%	\begin{itemize}
%	\item CbN: $\betab \eq \beta,  \ss \eq \hh,  \surf \eq \head$
%	\item CbV:$\betab \eq \betav,   \ss \in \{\ww, \leftc, \rightc \},   \surf \in \{  \weak, \lsym, \rsym \}$ 
%\end{itemize}
%\caption{CbN and CbV surface reduction}\label{fig:surfaces}
%\end{figure}

\begin{Def}[Unbiased iteration of surface reduction]\label{def:liberal} \label{def:U}
	Given 
	   $(\Lambda_\OpSet, \red)$, where %$\OpSet$ is a  set of operator symbols, and 
	   $\red$ is the contextual closure of  a rule 
	   $\betab\in \{\beta,\betav\}$, let  $\sred \subseteq\, \red$ be as follows:
\[ \sred \eq \hred 	\mbox{ if } \bbeta=\beta \text{ (CbN) }      \quad  \quad 
 \sred\in \{\wred,\lred,\xredx{\rsym}{}\}  	\mbox{ if } \bbeta=\betav \mbox{ (CbV) }.\]
The  relation $\lsred\, \subseteq\, \red$ is inductively defined as follows:
	\begin{itemize}
		\item if $M\sred M'$ then $M\lsred M'$;
		
			\item if $M \not\! \sred$ then $M \lsred M'$ is defined according the rules below.
	\end{itemize}
{\small$	 \infer{(\lam x.P) \lsred (\lam x.P')}{P \lsred P'} \quad
			\infer{PQ \lsred P' Q }{P \lsred P'}\quad
			\infer{PQ \lsred PQ'}{ Q \lsred Q'   } \quad
			 	\infer{\opp {P_1, \dots, P_i, \dots, P_k} \lsred \opp {P_1, \dots, P'_i, \dots, P_k}}{ P_i \lsred P'_i  } $}
% 		\item if $N \not\! \sred$ then 
% 	\end{itemize}
% {\scriptsize$	 \infer{N:=(\lam x.P) \lsred (\lam x.P')}{P \lsred P'} \quad
% 			\infer{N:= PQ \lsred P' Q }{P \lsred P'}\quad
% 			\infer{N:= PQ \lsred PQ'}{ Q \lsred Q'   } \quad
% 			 	\infer{N:= \opp {P_1, \dots, P_i, \dots, P_k} \lsred \opp {P_1, \dots, P'_i, \dots, P_k}}{ P_i \lsred P'_i  } $}

\vskip 8pt
\noindent The same  definition of $\lsred\, \subseteq\, \red$ still applies if  $\red$ is the contextual closure of  
$\Root \betab \cup   \Root\Rule $, 
% $\Root \betab \cup \big(\bigcup_{\rho \in \Rules} \Root\Rule\big)$ 
\Ie of the  rule $\Root \betab$ extended  with  some other rule $\Root\Rule$ on $\Lambda_\OpSet$.

\end{Def}

We study $\lsredbb$. It is RD-diamond (see \Cref{fact:diamond}) and   is a normalizing strategy  for both  
CbN and CbV $\lam$-calculi. Note that in CbN,  $\lsredb$ \emph{subsumes} usual  \emph{leftmost-outermost reduction}.
\begin{proposition}[$\los$-Factorization]\label{thm:Ufactorization} 
	Let $\betab\in \{\beta,\betav\}$. 
	\begin{equation*}\tag{$\los$-Factorization }
		M\redbb^*  N  \timplies M \lsredbb^* \cdot  \nlsredbb^* N
	\end{equation*}
\end{proposition}
%\begin{theorem}[$\los$-Factorization] Let $\betab\in \{\beta,\betav\}$. \\
%	\label{thm:Ufactorization} 
%	$ 	M\redbb^*  N  \timplies M \lsredbb^* \cdot  \nlsredbb^* N $.
%\end{theorem}
%
\begin{restatable}{proposition}{Uproperties}\label{lem:Uproperties}With the same assumptions as in \Cref{def:liberal}, let $\betab\in \{\beta,\betav\}$. Then:
\begin{enumerate}
	\item % \emph{RD-Diamond.} 
	$\lsredbb$  is    RD-diamond.
	\item %\emph{Same normal forms.}  
	$\lsredbb$ has the same normal forms as $\redbb$. 
	\item %\emph{Completeness.} 	
	Let $N$ be  $\betab$-normal. $M \redbb\tr N$ implies $ M\lsredbb\tr  N $.
%	\item $\los$-Factorization:	$M\redbb^*  N  \timplies M \lsredbb^* \cdot  \nlsredbb^* N$
\end{enumerate}
\end{restatable}

 Normalization for both CbN and CbV follows from the  points above.
\begin{theorem}[Normalization]\label{thm:Unormalization} For $\betab\in \{\beta,\betav\}$, 
		$\lsredbb$ is a normalizing strategy for $\redbb$.
\end{theorem}
\SLV{}
{\begin{theorem}[Normalization] \hfill\\
	\emph{CbN:}
	$\lsredb$ is a normalizing strategy for $\redb$\\
	\emph{CbV:}	$\lsredbv$ is a normalizing strategy for $\redbv$.
\end{theorem}
}
\SLV{}{
%\begin{remark}[Diamond and Optimality ]
\paragraph*{Leftmost-outermost and Optimality.}
Note that in CbN,  $\lsredb$ \emph{subsumes} usual  \emph{leftmost-outermost reduction}.
\footnote{Still, it carry also  all the advantages of non-deterministic $\hred$ which we have recalled in \Cref{sec:NDE}} $\loredb$. 
	The RD property implies that
	\SLV{if $\ered \subseteq \lsred$ is \emph{another strategy}, $\ered$ is normalizing if and only if $\lsred$ is normalizing. Moreover,  the cost of computing a \nnf is the same.}{ if a term $M$ has normal form $N$, then every  $\lsred$-sequences from $M$ end in $N$, in the same number of steps. Hence in particular,if $\ered \subseteq \lsred$ is \emph{another strategy}, $\ered$ is normalizing if and only if $\lsred$ is normalizing. Moreover, there is no difference in the cost of computing a normal form, in term of performed steps.  }
}	
\SLV{}{
	\paragraph*{Factorization}
	We observe here also a property of factorization which we will use later
	\begin{prop}[$\out$-Factorization]\label{prop:out-factorization}	Given $(\LambdaOp, \redbb)$, where  $\redbb\in \{\beta,\betav\}$, let  $\lsredbb$ be as in \Cref{def:liberal}, Point 1.
		Then:
		\[M \redbb^* U  ~~\mbox{ implies } ~~   M \lsredbb^* \cdot \nlsredbb^* U  \]
	\end{prop}
\SLV{}{	\begin{proof}
		Again, by  induction on $U$, using surface factorization and shape preservation.
	\end{proof}
}
%\end{remark}
}

\paragraph{Depth-first vs Breadth-first.}
 Leftmost-outermost reduction fires redexes  in a  depth-first way. 
 Instead, 
 $\lsred$ evaluates in a breadth-first  style, which is more suitable to deal with possibly infinitary reductions. 
For example, in CbN think  of $z(\Delta\Delta)(\dz\dz)$. Leftmost-outermost reduction never leaves the redex $\Delta\Delta$, while $\lsred$ can also fire $ (\dz\dz) $    yielding  $z(z(z \dots))$.

\paragraph{A parallel variant.} Once a term is  $\sred$-normal,  the process can be iterated  in any arbitrary order, or in \emph{parallel}. 
\SLV{Parallel  (multi-step) reduction $ \xredx {\pd}{}$ is easily defined (\Cref{sec:parallelU}).}{
	Given $(\Lambda,\red)$ and $\sred$ as  in \Cref{def:liberal}, the parallel  (multi-step) reduction $ \xredx {\pd}{}$ is easily defined. 
\begin{itemize}
	\item If  $M\sred{} M'$ then $M \xredx {\pd}{} M'$ ($M$ surface redex)
	\item If $M \not \red$ then $M \xredx{\pd}{\surf} M$ ($M$ \nnf)
	\item Otherwise,  for  $M $ which is $\lam x. P, \  P_1P_2 $, or $  \opp{ P_1 \dots P_k}$:
{\small \begin{gather*}
%	\infer{M \xredx{\pd}{\surf} M}{M \not \red} \quad
\infer{\lam x. P \xredx{\pd}{}   \lam x.P'}{P \xredx{\pd}{} P'} \quad
\infer{P_1P_2 \xredx{\pd}{}   P_1'P_2'}{ P_1\xredx{\pd}{}   P_1' & P_2\xredx{\pd}{}   P_2'}\quad
\infer{\opp{ P_1 \dots P_k}\xredx{\pd}{}   \opp{ P_1' \dots P_k'}}{ (P_i\xredx{\pd}{}  P_i')_{\ik}} 
\end{gather*}}
\end{itemize}

%\begin{prop}$\xredx{\pd}{}  $ is a (multi-step) normalizing strategy for $\red$.
%\end{prop}

The parallel $\xredx{\pd}{}$ is deterministic, and it is guaranteed to reach the $\red$-\nnf, if any exists.
\begin{prop}[$\Nnf$-completeness] Let $\betab\in\{\beta,\betav\}$ and $N$ $\betab$-normal.  ~$M\redbb^*N$  if and only if $M\xredx{\pd}{\betab}^*N $
\end{prop}
\SLV{}{\begin{proof}$ \Leftarrow $:  immediate.  $\Rightarrow$: by induction on $N$. The claim is an 
easy consequence of Surface Factorization and the fact that non-surface step preserve the shape of terms.  By the former $M\sredbb^* U \nsredbb^* N$.  By the latter, $U$ has the same shape as $N$; we then  examining the possible cases, and   conclude by \ih.
\end{proof}
}
}

% !TEX root = main_2022.tex

\renewcommand{\Nnfv}{\Nnf}

\section{Probabilistic \texorpdfstring{$\lambda$}{lambda}-calculi and Asymptotic Normalization}\label{sec:CbVproba}\label{sec:proba}\label{sec:PLambda}
A standard way to model \emph{probabilistic  choice} (a fair coin) is by means of  a binary  operator $\oplus$.  % as given in \refex{choice}.
We write $M \oplus N$ for $\oplus(M,N)$.  
Intuitively, $M\oplus N$   reduces to either $M$ or $N$,
\emph{with equal probability} $\two$.
Reduction is then defined not simply on terms but on (monadic) structures representing probability distributions over terms.
Here we  follow \cite{FaggianRonchi}, which defines both a CbV and a CbN calculus 
$\PLambda^\cbv$ and $\PLambda^\cbn$, where $\beta$ or $\betav$ reduction are  ``as usual'', so if a term contains no probabilistic operator, it behaves the same as  in the usual $\lam$-calculus (\ie the extension is \emph{conservative}).
Probabilistic reduction instead  needs to be constrained in order to have good  properties such as confluence (see \cite{FaggianRonchi},  and  \cite{deLiguoroP95,LagoZ12} for a  discussion of the issues). %We recall a standard example in \Cref{ex:confluence_failure}.

%
%
%Being interested in rewriting, in this paper we consider probabilistic $\lam$-calculi (\cite{Leventis19},\cite{FaggianRonchi}) whose  reduction is not limited to a deterministic strategy. 
%%which are  conservative extensions of the full (CbV or CbN) $\lam$-calculus,
%Probabilistic reduction needs to be constrained---in order to have good operational properties---but 
%  $\beta$ or $\betav$ reduction are  "as usual". That is, if a term contains no probabilistic operator, it behaves the same as  in the usual $\lam$-calculus (\emph{conservative} extension).
%Here we  follow \cite{FaggianRonchi}, which defines both a CbV and a CbN calculus:
%$\PLambda^\cbv$ and $\PLambda^\cbn$.
%We focus on  CbV, which is the most relevant 
%paradigm for calculi with effects. % the  definitions for CbN are similar
%%\RED{We refer to \cite{FaggianRonchi} for  details.}

\SLV{}{
\begin{example}[Failure of commutations]	
The need for constraints on the probabilistic operator is exemplified by the 
following well-known  counter-example to confluence, which  is a counter-example also to left-to-right standardization in its simpler form, head factorization.
In a CbN setting, 
	consider the duplicator $\Delta$, an effectful term  $ P = (I\oplus \Delta\Delta) $, and the following  two sequences from $\Delta P$:
	{\small \begin{itemize}
			\item $\mset{\Delta P }\nhred \mset{\two \Delta I, \two \Delta (\Delta\Delta)} \hred \mset{ \two II, \two (\Delta\Delta)(\Delta\Delta)}=\m$
			\item  $\mset{\Delta P} \hred   \mset{PP}=\m' $
	\end{itemize}}	
	\emph{Confluence failure}: $\m$ and $\m'$ have no common reduct.
	
		\emph{Head factorization failure}: from $\Delta P$ it is not possible to reach $\m$ performing head steps first.
\end{example}
}
%\subsubsection{Basics }
%A \emph{discrete probability space} is given by a pair $(\Omega, \mu)$,
%where  $\Omega$  is  a \emph{countable} set, and $\mu$ is  a \emph{discrete probability distribution}
%on  $\Omega$, \ie\  is a 
%function from $\Omega$ to $[0,1]\subset  \Real$ such that  $\norm \mu:= \sum_{\omega\in \Omega} \mu(\omega) = 1$.    Every  subset  $s \subseteq \Omega$  is given a  measure as $\mu(s)=\sum_{\omega\in s} \mu(\omega)$.
%%In the language of probability theory, a subset of $\Omega$ is called an \emph{event}.
%
%Any \emph{function} $F : \Omega \to \Delta$, where  $\Delta$ is another countable set, 
%	\emph{induces a  probability distribution} $\mu^F$ on $\Delta$ -- such a  function is called a \emph{random variable}.
%}{
% Any \emph{function} $F : \Omega \to \Delta$, where  $\Delta$ is another countable set, 
%\emph{induces a  probability distribution} $\mu^F$ on $\Delta$
%by composition: $\mu^F(d \in \Delta) := \mu (F^{-1}(d))$ \ie\ $\mu\{\omega \in \Omega: F(\omega) = d\}$. 
%In the language of probability theory, $F$ is called a \emph{discrete random variable} on $(\Omega, \mu)$.
%}

\subparagraph{Discrete Probability Distributions.}\label{notation:dist} Given    a countable set $\Omega$, a
function $\mu \colon\Omega\to[0,1]$ is a probability \emph{subdistribution} if 
$\norm \mu := \sum_{\omega\in \Omega} \mu(\omega)\leq 1$ (a \emph{distribution} if $\norm \mu=1$).
%With a slight abuse of language, we will use the term distribution  also for   subdistribution. 
Subdistributions allow us to deal with partial results.
We write  $\DST{\Omega}$ for the set
of   subdistributions on $\Omega$,  equipped with the pointwise order on functions:  $\mu \leq \rho$ if
$\mu (\omega) \leq \rho (\omega)$ for all $\omega\in \Omega$. 
$\DST{\Omega}$ has a bottom  element (the subdistribution $\zero$) and maximal elements (all distributions)}.

\subparagraph{Multi-distributions.}\label{sec:multi} 
We use multi-distributions \cite{Avanzini} to syntactically represent distributions, %we  rely on the notion of multi-distribution  \cite{Avanzini}.
A \emph{multi-distribution} $\m=\mdist{p_iM_i}_{ i\in I}$ on the set of terms $\LambdaOp$   is a finite multiset of pairs of the 
form
$pM$, with $p\in]0,1]$,  $M\in \LambdaOp$, and  $\sum_i p_i\leq 1$.
  The set of
all multi-distributions on $\LambdaOp$ is $\MDST \LambdaOp$.
The sum of multi-distributions is noted  $+$.  
The product $q\cdot \m$ of a scalar $q$ and a multi-distribution $\m$ is
defined pointwise $q\mdist{p_iM_i}_{\iI} := \mdist{(qp_i)M_i}_{\iI} $. We   write $\mdist{M}$ for   $\mdist{1M}$.

%\subsubsection{A CbV and a CbN probabilistic $\lam$-calculus}\label{sec:PLambda}
\paragraph{Syntax.}
% rewrite system where the objects to be rewritten are not terms, but
%monadic structures on terms, namely multi-distributions [8]. Intuitively, a multi-distribution
%represents a probability distribution on the possible reductions from a term.

%
%
%To guarantee good properties---such as confluence and factorization-- the contextual closure of the  probabilistic rule $\oplus$  needs to be restricted 
%{(see \cite{FaggianRonchi} and the Appendix for a discussion)}.
%The $\beta$-rule (resp. $\betav$-rule) instead is closed under \emph{arbitrary} context.
%\SLV{}{---this  guarantees that  the new calculus is a \emph{conservative 
%	extension} of CbN (resp. CbV) $\lam$-calculus. }

%and then 
%give a new proof of weak factorization, using our technique  and obtaining a neat, compact proof of 
%factorization, which only requires a few lines.

%\subsubsection{$\PLambda^\cbv$ }
\emph{Terms ($\Lambda_\oplus$)} and  \emph{values} are as in \Cref{sec:lambda}, with the operator $\op$ being here $\oplus$.

\subparagraph{Call-by-Value.} The  \emph{calculus} $\PLambda^\cbv$ is the rewrite system  $(\MPLambda, \Red)$ where   
 $\MPLambda$ is the set of  \emph{multi-distributions} on $\PLambda$ and 
 the relation $\Red\subseteq \MDST{\PLambda}\times \MDST{\PLambda}$ is defined in  
\Cref{fig:reductions} and \Cref{fig:lifting}.  First,  define   one-step reductions from terms to 
multi-distributions---so for 
example, $M\oplus N \red \mdist{\two M, \two N}$. Then,  lift the definition of reduction
to a binary relation on $\MDST{\PLambda}$,  in the natural way---for instance $\mdist{\two (\lam x.x)z, \two (M\oplus 
	N)} \Red $
$ \mdist{\two z, \four  M, \four N}$.
Precisely:
\begin{enumerate}
	\item  The reductions 	
	$\redbv,\redo\subseteq \PLambda \times \MPLambda$ are defined  in 	Fig.~\ref{fig:reductions}.  \emph{Contexts $\cc$ and    $\ww$}  are  as in \Cref{sec:lambda}.
	Note that $\betav$ is closed under arbitrary context, while the $\oplus$ rule---probabilistic choice---is  closed under  weak 
	contexts $\ww$ (no reduction  in the scope of $\lambda$ or $\oplus$).
	%body of a function or in the scope  of an operator $\oplus$). 
	We write $\sredbv$ for  the closure of $\betav$ under  context $\ww$.
%
%	Instead, the $\betav$ rule  is closed under	arbitrary contexts.  
	The relation  $\red$  is  $\redbv \cup \redo$.
%	\footnote{In \cite{FaggianRonchi},  a \emph{weak} 
%		reduction (resp. weak context) is called \emph{surface}, and hence noted  $\xredx{\mathsf{s}}{}$.} 
%
	\emph{Surface  reduction} is  $\sred \eq \sredbv \cup \redo$.  
	 A $\red$-step which is not surface is noted $\nsred$.

	\item 
	The lifting of a relation
	$\red_r \subseteq  \Lambda_\oplus \times \MDST{\Lambda_\oplus}$ to a reduction on multi-distributions is defined  in Fig.~\ref{fig:lifting}. In particular, 
	$\red, \redbv,\redo, \sred,\nsred$ lift to 
	$\Red,\Redbv,\Redo, \sRed, \nsRed$.  
	
\end{enumerate}
%Note that $\redo$ is restricted to weak contexts.
%
%\paragraph*{Normal Forms.} 
A term $M$ is   $\red$-\textbf{normal} if there is no $\m$ such that $M\red\m$.  We also write $M\not\red$. 
We denote by $\Nnfv$ the set of the \textbf{normal forms} of $\red \eq (\redbv\cup \redo)$. 
\begin{figure}
	\fbox{	\footnotesize{ \begin{minipage}{0.41\textwidth}
			\begin{align*}   
			\cc \hole{(\lambda x.M)V } &\redbv \mdist{\cc\hole{ M \subs x V} } 
			\\
			\ww\hole{M\oplus N} &\red_{\oplus} \mdist{\frac{1}{2}\ww\hole M, \frac{1}{2}\ww\hole N}
			\\[4pt]
		 	\red\,:=\, \redbv \!\cup &\redo \quad
			\sred \!:=\! \sredx{\betav} \!\cup \redo 
			\end{align*}
	\caption{$\red$-steps for the calculus $\PLambda^\cbv$  }\label{fig:reductions}
	\end{minipage}}
 	\quad
	{\footnotesize 
	\begin{minipage}[c]{0.52\textwidth}
			$
			\dfrac{}{\mdist{M}\Red \mdist{M}} \quad 
			\dfrac{M\red\m}{\mdist{M}\Red \m} \quad   
			\dfrac{(\mdist{M_i} \Red \m_i)_{\iI} }{ \multiset{p_{i}M_{i}\mid i\in I} \Red  \sum_{\iI} {p_i \!\cdot \m_i}} 
			$  	
			\caption{Lifting of $\red$}\label{fig:lifting}
			\vskip 4pt
			{	$
				\dfrac{M\not \red}{\mdist{M}\full \mdist{M}} \quad 
				\dfrac{M\red\m}{\mdist{M}\full \m}  \quad   
				\dfrac{(\mdist{M_i} \full  \m_i)_{\iI} }{ \multiset{p_{i}M_{i}\mid i\in I} \full  \sum_{\iI} {p_i \!\cdot \m_i}} 
				$}\caption{Full lifting of $\red$ }\label{fig:lifting_full}
			\end{minipage}}}
\end{figure}

\subparagraph{Call-by-Name.} The  calculus  $\PLambda^\cbn$  is defined in a similar way, by replacing $\betav$ with $\beta$ and weak contexts with head contexts $\hh$ (as defined in \Cref{sec:lambda}).

\SLV{\subparagraph{Observations on multi-distributions.}\label{sec:observations} }
{\subsubsection{Observations on multi-distributions}\label{sec:observations}  }

\SLV{}{
A multi-distribution $\m\in \MDST{\PLambda} $ is a \emph{syntactical representation}  of a space %(the space of all programs runs),
where to  each element  is associated a probability and a term of $\Lambda$ and contains a large amount of information. 
We analyze it  by defining random variables that \emph{observe specific properties of interest}. 
%We analyze it  w.r.t. \emph{specific properties of interest},  such as normal forms ($\NF$), and values ($\Val$).
}

%Let us give an example which we will use through the paper.  
In \textbf{CbV}, events of interest are the set  $\Val$ of values  and  the set $\Nnfv$ 
of  $\red$-normal forms (for $\red \eq \redbv \cup \redo$). 
Focusing on $\Nnfv$, we can  define:
\begin{figure}
		\fbox{
		\begin{tabular}{r@{} l@{\quad} l@{\ } c}
			%&{\small   $\MDST{\PLambda} \ni ~ \m=\mdist{p_iM_i}_{ i\in I}$}&&	\\ [-2pt]
			$ \obsN \colon$ &$\MDST{\PLambda} \to  \DST{\Nnfv} $  &$ \mdist{p_iM_i}_{ i\in I} \mapsto \mu $  &
			where 
			$ \forall N\in {\Nnfv}, \ \mu(N) = \sum_{\iI}p_i \mbox{ s.t. }   M_i=N
%			\left\{
%			\begin{array}{ll}
%				p& \mbox{ if } p=\sum_{\iI}p_i \mbox{ s.t. }   M_i=U\\
%				0&\mbox{otherwise;}
%			\end{array}
%			\right.
			$\\[4pt]
			$\obs_{\pn} \colon $& $\MDST{\PLambda} \to  [0,1] $   &$ \mdist{p_iM_i}_{ i\in I} \mapsto \norm{\mu}$ &
		\end{tabular}
		}
	\caption{CbV observations on multi-distributions}\label{fig:obs_proba}
\end{figure}
%\begin{figure}
%%	\fbox{
%		\begin{tabular}{l c l c}
% %&{\small   $\MDST{\PLambda} \ni ~ \m=\mdist{p_iM_i}_{ i\in I}$}&&	\\ [-2pt]
%	$ \obsN:$ &$\MDST{\PLambda} \to  \DST{\Nnfv} $  &$ \mdist{p_iM_i}_{ i\in I} \mapsto \mu $  &
%where 
%$( \forall U\in {\Nnf}):~~ \mu(U) = 
%\left\{
%\begin{array}{ll}
%	p& \mbox{ if } p=\sum_{\iI}p_i \mbox{ s.t. }   M_i=U\\
%	0&\mbox{otherwise;}
%\end{array}
%\right.
%$\\[4pt]
%	$\obs_{\pn}: $& $\MDST{\PLambda} \to  [0,1] $   &$ \mdist{p_iM_i}_{ i\in I} \mapsto \norm{\mu}$ &
%	\end{tabular}
%%}
%	\caption{CbV observations on multi-distributions}\label{fig:obs_proba}
%%
%\end{figure}
\begin{itemize}
	\item $ \obsN$
   extracts  from $\m=\mset{p_iM_i}_{\iI}$ a \emph{subdistribution $\mu$ over normal forms}.
For example, if 
		 $\m=\mset{\four \true, \frac{1}{8} \true, \frac{1}{4} \false, \four II  }$,   
		{$\obsN(\m)$}   is the subdistribution {$\{\true^{\frac{3}{8}},  \false^\four  \}$, \ie~$\mu(\true) = \frac{3}{8}$, $\mu(\false) = \frac{1}{4}$}.

	\item $\obs_{\pn}$  observes the  probability that   $\m$ has reached  a normal form. For example, with $\m$ as above,  $ \obs_{\pn} (\m)=\frac{5}{8}$.

\end{itemize}
%
%\begin{itemize}
%	\item $ \obsN: \MDST{\PLambda} \to  \DST{\Nnfv} \quad \m \mapsto \mu $\\
%which   extracts  from $\m=\mset{p_iM_i}_{\iI}$ the \emph{(sub)distribution $\mu$ of normal forms}, where $\mu$ is defined as follows
%\[ \forall U\in \Nnf:~~~ \mu(U) = 
%\left\{
%\begin{array}{ll}
%	p& \mbox{ if } p=\sum_{\iI}p_i \mbox{ s.t. }   M_i=U\\
%	0&\mbox{otherwise;}
%\end{array}
%\right.
%\]
%
%\begin{example}\label{ex:pnf}
%	Let  $\m=\mset{\four \true, \frac{1}{8} \true, \frac{1}{4} \false, \four II  }$.
%	Then 
%	$\obs(\m)$   is the subdistribution $\{\true^{\frac{3}{8}},  \false^\four  \}$. 
%\end{example}
%
%
%\item $\obs_{\pn}: \MDST{\Lambda} \to  [0,1]   \quad \m \mapsto \norm{\mu}$  which  observes the  probability that   $\m$ has reached  a normal form. For example, taking \Cref{ex:pnf},  $ \obs_{\pn} (\m)=\frac{5}{8}$.
%
%\SLV{}{\item   $\obs_{\pv}: \MDST{\Lambda} \to  [0,1]$,
%which   observes   the  probability that  $\m$ has reached  a value, is    defined in a similar way.}
%\end{itemize}
%
In \textbf{CbN}, events of interest are the set  %$\Nnfn$ 
of normal forms (w.r.t. $\redb \cup \redo$), and     the set $\Hnf$ of head normal forms. The corresponding observations are  defined in the obvious way.

% !TEX root = main_2022.tex

\subsection{Asymptotic Normalization for Probabilistic \texorpdfstring{$\lambda$}{lambda}-Calculi}\label{sec:P_Anormalization}\label{sec:PCbV}

We can now  revisit  the probabilistic calculi {$\PLambda^\cbv$}  and {$\PLambda^\cbn$} as  QARS, and define  for them an asymptotically normalizing strategy.   
%Similar results hold for for CbN, taking into account in the definition of $\llred$ that $\redbv$ is replaced by  $\redb$ and surface reduction is $\hred$.
%	The method  and proofs are exactly the same.}
%
We develop explicitly only the CbV case, but similar definitions and  results hold for  CbN, 
%taking into account that surface reduction is $\hred$ 
 taking into account  that  $\redbv$ is replaced by  $\redb$ and surface reduction is $\hred$.
Method  and proofs are exactly the same.

\medskip
%

%Similar results hold for for CbN, taking into account in the definition of $\llred$ that $\redbv$ is replaced by  $\redb$ and surface reduction is $\hred$.
%The method  and proofs are exactly the same.

%
%\RED{
%\paragraph{Surface and Least-Level Reduction on $\PLambda$.} 
%Given  a step $M\redc M' \in \PLambda \times \PLambda$, where  $\redc\in \{\redb,\red_{\oplus_1},\red_{\oplus_2}\}$,
%we write $M\xredx{\surf}{\gamma} M'$ if the reduction is surface,   and  $M\lsredx{\gamma} M'$ if 
%the reduction is least-level. \RED{w.r.t. \emph{both} $\oplus$-redexes (\ie terms of the form $\oplus M N$) and $\beta$-redexes.
%We write $M\nlsredx{\gamma} M'$ if the reduction is internal.}
%%Note that $\lev{\oplus CM} = \lev C+1 = \lev{\oplus MC} $. 
%}
%The calculus introduced in  \cite{FaggianRonchi}  is a conservative extension of the $\lam$-calculus, because the $\beta$-rule is unrestricted. To have confluence and therefore a good definition of result, the probabilistic reduction is restricted to depth 0. 

%\paragraph{{$\MPLambda^\cbn$} and {$\MPLambda^\cbv$} as QARS.}

%The  property of interest  here  is to reach    normal form. 
%We observe  either the probability to be in normal form, or    a (sub)-distribution on normal forms. 
%With respect to this observation, we define a reduction which  is  asymptotically complete and in fact induces  an asymptotically  normalizing strategy. 

%\subsubsection{Asymptotic Normalization for Probabilistic CbV $\lam$-calculus}\label{sec:PCbV}

The QARS framework allows us to  express and analyze   the asymptotic behaviour of the calculus $\PLambda^\cbv = (\MPLambda, \Red)$. Here we  are interested in   % which   associates to   each multidistribution $\m$ its distribution over  normal forms
$	\obsN:\MPLambda\to \DST{\Nnf} $
as defined  in \Cref{fig:obs_proba}.
It is immediate  that $\m\red \m'$ implies $\obsN(\m)\leq \obsN(\m')$. So,
%\begin{center}
%	$\QQ^\cbv=\brack{(\PLambda^\cbv,\Red), \obsN}$ is a QARS.
%\end{center}
%\begin{center}
$ 	\brack{\PLambda^\cbv, \obsN}$ is a QARS. 
%\end{center}
%
\SLV{}{\begin{remark}
	  $\obsN$ is finer than $\obs_\pr$, and  $\obsN(\r)=\obsN(\s)$ implies $\obs_\pr(\r)=\obs_\pr(\s)$;  we therefore only focus on $\obsN$.	
\end{remark}
}
We prove (\Cref{thm:main_CbV}) that $\PLambda^\cbv$ satisfies  the following properties:
\SLV{
	(1) the result $\sem \m $ of computing $\m$ is well defined; 
	(2) there exists a strategy that is guaranteed to produce  $\sem \m $.
%	there is a strategy that is guaranteed to produce  $\sem \m $.
}{
\begin{enumerate}
	\item  the result $\sem \m $ of computing $\m$ is well defined;
	\item there exists a strategy that is guaranteed to produce  $\sem \m $.
\end{enumerate}
}

%We  follow the method presented in \refsec{QARS}.%:  we first study asymptotic completeness, and then establish uniqueness of the result.

\subparagraph{Beyond the surface.}%$M\in \PLambda$ is $\surf$-normal if $M\not\sred$ (\ie $M\not\redo$ and $M\not\sredbv$). 
We define a  reduction
 $\llred\subseteq \PLambda\times \MPLambda$ which performs surface steps  ($\sred = \sredbv \cup \redo$, see  \Cref{sec:PLambda}) as much as  possible, and then  iterates the process on the subterms.  
 There are two subtleties here. First: $M\not \sredx{}$ if and only if  ($M\not\sredbv$ and $M\not \sredo$).
 Second: an occurrence of   $\oplus$-redex can only be fired when it is a surface redex. 
 By keeping this into account,  \Cref{def:liberal} updates as follows. We denote by $\Snf$ the set of \textbf{$\sred$-normal forms}. 
	\begin{Def}[Unbiased evaluation $\llRed$, $\llfull$ ]\label{def:E}
		
		 \begin{itemize}
			\item The relation $\llred\subseteq \PLambda \times \MPLambda$ is  defined by the following rules, depending if  $M\not \in \Snf$ 
			or 
			 $M \in \Snf$.  The relation $\lsredx{\betav}$ is as in \Cref{def:liberal}.
		\[\infer[\mathsmaller{(M \not \in \Snf)}]{M\llred\ \m}{M \sred \ \m}   \quad\quad \infer[\mathsmaller{(M \in \Snf)}]{M \llred \mdist{M'}}{  M\not \sred   & M \lsredx{\betav}M'}\]	
\item  	$\llRed, \llfull \subseteq \MPLambda \times \MPLambda$ are respectively   the lifting  and  full lifting of $\llred$ (\Cref{fig:lifting,fig:lifting_full}).
		\end{itemize}
	\end{Def}

Clearly, $\llred \subseteq \red$, and 
moreover	$\red$ and $\llred$ have the \emph{same normal forms}.
%	\begin{fact}
%		 $M\not \red$ iff $M\not \xredx{\los}{}$
%	\end{fact}

\ArX{}{\begin{remark} Notice that   $\llred  \  \not= \  \lsredbv \cup \redo$,  which is not diamond, check $(\Delta\oplus  \Delta\Delta)(x(Iz))$.
	\end{remark}
}

\begin{remark}\label{rem:complete}  
	$\llRed$ is $\obsN$-complete, but  not $\obsN$-normalizing for $\Red$. Indeed, the  sequence  $\m = \mdist{II\oplus \Delta\Delta} \llRed\mdist {\two II, \two \Delta\Delta}\llRed \mdist{\two II, \two \Delta\Delta} \llRed \dots $  never fires $II$.
	% and has for limit the $\zero$ distribution, while $\xLim{\obsN}({\m,\Red}) $ also contains  $\{ I^{\two}\}\in \MDST{\Nnf}$. 
	The solution is to move to $\llfull$, which forces all non-normal terms to reduce. 
	Note  that $\llfull$ does not factorize $\Red$.
\end{remark}
%
%\subsection{CbV Finitary properties}\label{sect:finitary}
%
We show that 
$\llfull$ is  an $\obsN$-normalizing strategy  for $\Red$. 
The pillars   of our  construction are  $\EX$-factorization  and weighted Random Descent. 
The former holds for  $\llRed$\SLV{}{(obtaining that it is  $\obsN$-complete)},  the latter  for $\llfull$.
\SLV{}{(obtaining that $\llfull$ has a unique limit). } 
\begin{prop}[Factorization and $\obsN$-neutrality]\label{prop:Pfactorization} \hfill
	\begin{enumerate}
		\item \emph{$\llsmall$-factorization:}~~	$\m \Red^*\n$ implies $\m \llRed^*\cdot \iRed^* \n$.
		\item \emph{$\obsN$-neutrality:}~~	 $\m\iRed \n$ implies  $\obsN(\m) =\obsN(\n)$.
	\end{enumerate}
\end{prop}
\begin{prop}[Diamond]\label{prop:fulldiamond}$\llfull$ is  $\obsN$-diamond.
\end{prop}
\SLV{}{
	\paragraph{Factorization and Neutrality}
\begin{prop}[$\llsmall$-Factorization of $\Red$.]\label{prop:Pfactorization} In $\PLambda^\cbv$:
		$\m \Red^*\n$ implies $\m \llRed^*\cdot \iRed^* \n$
\end{prop}
\begin{proof}
By  surface factorization of $\Red$ and   $\lsredbv$ factorization (\Cref{prop:out-factorization}).
\end{proof}

\begin{prop}[neutrality]\label{prop:Pneutral} If 
	$\m\iRed \n$ then  $\obsN(\m) =\obsN(\n)$.
\end{prop}
\begin{proof}By the fact that $\nllred \subseteq \nsred$ preserves the shape of terms, and the redexes.
\end{proof}
\paragraph{Diamonds.}\label{sec:diamonds}
The reduction $\llfull$ is $\obsN$-diamond. 
\Cref{prop:fulldiamond} follows from the following
\begin{lemma}[Pointed Diamond]
	\label{l:diamond}
	Let $\a,\c\in \{\beta,\oplus\}$.
	Assume  $M$ has two distinct redexes, such that $M \lsredx{\a} \m_1 $ and $M \lsredx{\c} \m_2$. 
	Then  \begin{enumerate}
		\item exists $\s$ such that $\m_1\llfull{}_{\c}\ \s$ and $\m_2\llfull{}_{\a}\ \s$.
		
		\item Moreover, no $M_i$ in  $ \m_1=\mdist{ p_i M_i}_{i} $  and no  $M_j $ in $\m_2=\mdist{q_jM_j}_{j}$ is normal.
	\end{enumerate}
\end{lemma}
\begin{prop}[Diamond]\label{prop:fulldiamond}$\llfull$ is  $\obsN$-diamond.
\end{prop}
}
%
%%%%%%%%%%%%%%%%%%%%%%%%%%%%%%%
%
\SLV{}{\subsubsection{Asymptotic Normalization}}
We are now ready to  prove that, in $\brack{\PLambda^\cbv, \obsN}$,  the  reduction $\llfull\subseteq \Red$ (\ie, the full lifting of $\llred$) is guaranteed to compute  the best possible result from each $\m \in \MPLambda$.
\subparagraph{Asymptotic Completeness.}
%\SLV{}{Recall  that   $\llfull$ is \emph{asymptotically complete} in $\brack{(\MPLambda,\Red),\obs_N}$ if for each $\m$
%\begin{center}
%	$\m\Red\conv{} \cpo{\r}$ implies 	($\m\llfull\conv{} \cpo{\s}$ and $\cpo{\r} \leq \cpo{\s}$).
%\end{center}
%}
%
We have  that $\llRed$ is  asymptotically complete for $\Red$, because  it satisfies the  conditions % (i) and (ii) 
of   \Cref{thm:ACompl} (by \Cref{prop:Pfactorization}).
\begin{lemma}\label{prop:RedACompl}
If	$\m\Red\conv{} \cpo{\r}$ then 	$\m\llRed\conv{} \cpo{\r}$.
\end{lemma}
%\begin{proof}$\Red ^*=(\llRed \cup \iRed)^*$ satisfies the  conditions % (i) and (ii) 
%	of   \Cref{thm:ACompl}, by \Cref{prop:Pfactorization} %and \ref{prop:Pneutral}.
%%Condition (i),  $\llsmall$-Factorization, is proved in \refsec{Pfactorization}. Condition (ii) , $\int$-neutrality, is  \refprop{Pneutral}. Hence we have established 
%\end{proof}
In turn,  $\llfull$ is  asymptotically complete for $\llRed$ (immediate). % that is 	$\m\llRed\conv{} \cpo{\r}$ implies 	($\m\llfull\conv{} \cpo{\s}$ and $\cpo{\r} \leq \cpo{\s}$).
So  via  \Cref{prop:RedACompl}  we have:
\begin{theorem}\label{thm:complete}$\llfull$ is asymptotically complete for $\Red$: if
		$\m\Red\conv{} \cpo{\r}$ then $\m\llfull\conv{} \cpo{\s}$ and $\cpo{\r} \leq \cpo{\s}$.
\end{theorem}
%\begin{proof}
%	By Lemmas \ref{prop:RedACompl} and \ref{prop:fullACompl}.
%\end{proof}
%
%\vspace*{-8pt}
\subparagraph{Unique Result.}
All $\llfull$-sequences  from $\m$ converge to the same limit, 	by \Cref{prop:fulldiamond,thm:diamond}. 
\begin{theorem}\label{thm:unique} 
	$\xLim{\obsN}(\m,\llfull)$ contains a unique element.
\end{theorem}
\SLV{}{
 All $\llfull$-sequences  have the \emph{same} limit, because ($\obsN$-random descent)   
 all $\llfull$-sequences $\seqr \r$, $\seqr \s$ from the same $\m$ are indistinguishable 
w.r.t $\obsN$ :  \SLV{\quad	$\seqr {\obsN~\r} = \seqr{\obsN~\s}$.}{
\begin{center}
	$\seqr {\obsN~\r} = \seqr{\obsN~\s}$.
\end{center}
}
}
%\vspace*{-8pt}
\subparagraph{Asymptotic Normalization.}
By \Cref{thm:main},
the  main result   follows from   \Cref{thm:complete,thm:unique}.
%With the  notations  of \Cref{sec:QARS}, we have
\begin{theorem}[Main, probabilistic \CbV]\label{thm:main_CbV} 
	%For each $\mdist{M}=\m\in \MPLambda$:
	For each $\m\in \MPLambda$:
	\begin{enumerate}
		\item $\sem \m$ is  defined;
		
		\item $\m \llfull\conv{} \cpo \r$ if and only if $\cpo \r = \sem \m$.
	\end{enumerate}
Hence $\llfull$ is an  \emph{$\obsN$-normalizing strategy} for $\Red$ (see \Cref{rem:Anormalizing}).
\end{theorem}

%\begin{cor}
%	$\llfull$ is a \emph{normalizing strategy} for $\Red$.
%\end{cor}

\SLV{}{
\subsection{Asymptotic Normalization for Probabilistic CbN $\lam$-calculus}\label{sec:PCbN}
%\subsection{CbN Probabilistic Lambda Calculus}\label{sec:PCbN}

{Similar results hold for for CbN, taking into account in the definition of $\llred$ that $\redbv$ is replaced by  $\redb$ and surface reduction is $\hred$.
	The method  and proofs are exactly the same.}
}

Some simple examples will help to see how the normalizing strategy works, and how it differs from surface reduction.  
\begin{example}\label{ex:CBVPfull}
	Recall that $\betav$-reduction is unrestricted, so for example  $M= \lam z. (Iz) \redbv \lam z.z$. Instead, $ \lam z. (Iz) \not\sred$, because surface reduction cannot fire under abstraction. So 
	\emph{surface reduction is not a complete strategy} w.r.t. $\betav$-normal forms.
	
	A direct  consequence is that surface reduction is not informative about normalization, as it  produces ``false positive''. For example,  $N= \lam z. \Delta\Delta$ is  diverging  w.r.t. $\betav$-reduction, but it is a surface normal form.
	Let us now incept probability (with the terms  $M$ and $N$ as above).
	\begin{enumerate}		
		\item  Let $R= (\lam x. M \oplus xx)(\lam x. M \oplus xx) $.  Then		
		$\mdist{R} \llfull  \mdist{M \oplus R} \llfull \mdist{ \two M, \two R} \llfull\mdist{ \two I, \two M \oplus R}  \llfull  \mdist{ \two I, \four M,  \four R} \llfull \cdots $. At the limit,  $R$ converges with probability $1$ to   $I$, as wanted. 
		%that is $\mdist{R} \llfull\conv{}   \{I^\two\}$.
		
		\item The term $ S= (\lam x.  N \oplus xx)  (\lam x. N \oplus xx)$  converges to  normal form with probability $0$.  One can easily check that  
		$\mdist{S} \llfull\conv{}  \zero$.
		
		\item The term $ S'= (\lam x.  (N \oplus I) \oplus xx)  (\lam x. (N\oplus I )\oplus xx)$  converges with probability $\two$ to the normal form $I$.  One can easily check that  
		$\mdist{S'} \llfull\conv{}  \{ I ^ \two\}$.
		
		%	\item $( \lam x. (N \oplus I) \oplus xx)  (\lam x. (N\oplus I )\oplus xx) \sRed \mdist{\two M, \two N}$  so this term converges with probability $1$.
		
		\ArX{}
		{	\item  Notice  that  from $S'$, the reduction $\llRed$ (instead of $\llfull$) would  admit also reduction sequences which produces no normal form, 		for example  $S' \llRed^* \mdist{\four N, \two I, S'} \llRed \mdist{\four N, \four I, \two S'} \llRed \mdist{\four N, \four  I,  \two S'} \cdots$. By persistently reducing only $N$, we have 	$\mdist{S'} \llRed\conv{}  \zero$. }
	\end{enumerate}
\end{example}		

\begin{example}\label{ex:CBVPfull2}
	One can easily build probabilistic terms with a more interesting behaviour than those in \Cref{ex:CBVPfull}. First, observe that for $F=\lam x.I$ (encoding  the boolean \emph{false}), we have that $(\lam z. FF) \redbv F $.  	 	
		Now let   $U=\lam xy. \big(y\oplus xx(\lam z. yy)\big)$ and consider the term $UUF$, which converges with probability $1$ to $F$. Indeed 	$  \mdist{UUF}   \llfull\conv{}  \{F^1\}  $.  In contrast, surface reduction converges to a distribution over  \emph{countably many different surface normal forms}, since  each iteration produces a new \snf:  $  \two F, \four \lam z.FF, 
		\frac{1}{8} \lam z.  (\lam z.FF)(\lam z.FF) , ... $. 
\end{example}

%\section{Discussion and Conclusions}
%\paragraph{Summary}

% !TEX root = main_2022.tex

\section{Asymptotic Normalization:  More Case Studies}\label{sec:output}
%\vspace*{-15pt}
Our method applies---uniformly---to the other examples in  \Cref{sec:examples}. 
In this section we consider a CbV $\lam$-calculus extended with an output operator. For the sake of a compact presentation, 
we take as output  not a string, but simply an integer (think of it as a string on a single character). 
Albeit simple, this case study allows us to illustrate the subtleties related to limits with output calculi, and the use of our method.
In a similar way, one can revisit B{\"o}hm Trees as the limit of a specific asymptotic strategy---we leave  this to Appendix~\ref{sec:BT}.

\paragraph{$\lambda$-calculus with output: the payoff calculus.} \label{sec:payoff}

% 	\begin{itemize}
%	\item ${\ww\hole{\tick.P} \red_{\tick}  \ww\hole P} \\
%		 {\cc \hole{(\lambda x.M)V }\redbv \cc\hole{ M \subs x V} }\\
%	{\ww \hole{(\lambda x.M)V }\wredbv \ww\hole{ M \subs x V} }$\\
%		
%	
%	\item 	$ \red\,:=\, \redbv \cup \red_{\tick} $\quad
%	$ \wred\,:=\, \wredx{\betav} \cup \red_{\tick} $
%\end{itemize}

% 	\begin{itemize}
%	\item $\pair n {\ww\hole{\tick.P}} \red_{\tick}  \pair {n+1}{\ww\hole P} \\
%	\pair n {\cc \hole{(\lambda x.M)V }}\redbv \pair n  {\cc\hole{ M \subs x V} }\\
%	\pair n {\ww \hole{(\lambda x.M)V }}\wredbv \pair n {\ww\hole{ M \subs x V} }$\\
%	
%	
%	\item 	$ \red\,:=\, \redbv \cup \red_{\tick} $\quad
%	$ \wred\,:=\, \wredx{\betav} \cup \red_{\tick} $
%\end{itemize}

\begin{figure}
	\centering
	 $\infer {\opair n {\ww\hole{\tick.P}} ~\wredx{\tick} ~  \opair {n+1}{\ww\hole P}}{} $
	\quad
	$ \infer {\opair n M \wredbv \opair n{M'}}{M\wredbv M'}$
	\quad
	$ \infer {\opair n M\redbv \opair n {M'}}{M\redbv M'}
	$
	%\\[4pt]
%	\item 	$ \red\,:=\, \redbv \cup \wredx{\tick} $\quad
%	$ \wred\,:=\, \wredx{\betav} \cup \red_{\tick} $

\caption{Payoff reductions  $\wredx{\tick},{\redbv},{\wredbv} \subseteq (\Nat\times \Lambda_{\tick})  $ }\label{fig:payoff}

% 	{\scriptsize 	
% 	$\infer{\opair 0 M \pwred  \opair 0 {M'} }{M\Root{\betav} M'}  
% 	\quad 
% 	\infer{\opair 0 M \pwred  \opair 1 {M'} }{M\Root{\tick}M'}
% 	\quad 
% 	\infer   {\opair 0 V \pwred  \opair 0 {V} }{} 
% 	\quad
% 	\infer{\opair 0 {P_1P_2}  \pwred    \opair {k_1+k_2} {P_1P_2}  }
% 	{  \opair 0 {P_1}  \pwred    \opair {k_1} {P_1} & \opair 0 {P_2}  \pwred    \opair {k_2} {P_2} }
% 	\quad
% 	\infer{ \opair {n} {M}  \pwred    \opair {n+k} {M'}}{n>0 & \opair 0 {M}  \pwred \opair {k} {M'} }
% 	$}
	\begin{gather*}
	\infer{\opair 0 M \pwred\!  \opair 0 {M'} }{M\Root{\betav} M'}  
	\qquad 
	\infer{\opair 0 M \pwred\!  \opair 1 {M'} }{M\Root{\tick}M'}
	\qquad 
	\infer{\opair 0 V \pwred\!  \opair 0 {V} }{} 
	\\
% 	\infer{\opair 0 {P_1P_2}  \pwred\!    \opair {k_1+k_2} {P_1P_2}  }
% 	{  (\opair 0 {P_i}  \pwred\!    \opair {k_i} {P_i})_{i \in \{1,2\}} } 
	\infer{\opair 0 {P_1P_2}  \pwred\!    \opair {k_1+k_2} {P_1P_2}  }
	{  \opair 0 {P_1}  \pwred    \opair {k_1} {P_1} & \opair 0 {P_2}  \pwred    \opair {k_2} {P_2} }
	\qquad
	\infer{ \opair {n} {M}  \pwred\!   \opair {n+k} {M'}}{n>0 & \opair 0 {M}  \pwred    \opair {k} {M'} }
	\\[-25pt]
	\end{gather*}
\caption{Parallel weak reduction in the payoff calculus}\label{fig:parallel_payoff}
\end{figure}

	The  payoff $\lam $-calculus   (called cost $\lam $-calculus in \cite{LagoG19,GavazzoF21}) extends the $\lam$-calculus with a 
	ticking operation. Its intrinsic purpose is to  facilitate an  intensional analysis of programs,   endowing  terms with   constructs  to  perform cost analysis.

Let	$\Lambda_{\tick}$ denote the set  of $\lam$-terms extended  with a unary  operator $\tick$. 
	The elements of the payoff calculus are  pairs  $\m = \opair n M$ of a counter $n\in \Nat$ and a closed term $M\in \Lambda_{\tick}$. Intuitively, 
	the term   $\tick   (P)$ increments the counter by $1$, and continues as $P$.
	Following \cite{GavazzoF21}, in  $ (\Nat\times \Lambda_{\tick})$ we define the full reduction $\red$ and the weak reduction $ \wred $ as follows:
\begin{center}
		$ \red \;\defeq\; \redbv \cup \; \wredx{\tick} $\qquad
	$ \wred \;\defeq\; \wredx{\betav} \cup \; \wredx{\tick} $
\end{center}
	where $\wredx{\tick}$,$\redbv$, $\wredbv \subseteq   (\Nat\times \Lambda_{\tick})  $ are  given in \Cref{fig:payoff}. 
Note that  weak effectful reduction
$\wredx{\tick}$ is   the closure under weak  context $\ww$ of the rule $(\tick.P)\Root{\tick}P$ (effects are only allowed under weak context). 
Left and right reductions $\lred$ and  $\xredx{\rsym}{}$ %on $\Nat\times \Lambda_\tick $
can  be  defined  similarly.

	\SLV{}{
	\paragraph{Payoff of convergent computations.}
	Assume we observe (in $\Nat^\infty$) the payoff  associated to the computation of
	a value
	\[
	\left\{
	\begin{array}{ll}
		\obs_{v} (n:V) \deff n &  \\
		\obs_{v} (n:M) \deff 0 & \mbox{ if } M \not\in \Val 
	\end{array}
	\right.\]
	
	%\[\obs_{v} (n:V) =n  \quad\quad  \obs_{v} (n:M) =\bot  \mbox{ if } M \not\in
	%\Val \]
	%
	It is easy to check that the  reduction $\wred$ satisfies the  $\obs_v$-diamond property.  Hence
	\begin{fact}Let $M \in \ \Lambda_{\tick}$. Assume 
		$M:0\wred^k (n:V)$, then \emph{every} maximal  $\wred$-reduction ends  after $k$
		steps in the same pair $(n:V)$.
	\end{fact}
}
	
%	\paragraph{Payoff of diverging computations.}

	%\begin{figure*}
	%	\[\infer{\opair 0 M \pwred\ \opair 0 {M'} }{M\Root{\betav} M'}  \quad \infer{\opair 0 M \pwred\ \opair 1 {M'} }{M\Root{\tick}M'}
	%	\quad \infer   {\opair 0 V \pwred\ \opair 0 {V} }{}\quad
	%	\infer{\opair 0 {P_1P_2}  \pwred\   \opair {k_1+k_2} {P_1P_2}  }
	%	{  \opair 0 {P_1}  \pwred\   \opair {k_1} {P_1} & \opair 0 {P_2}  \pwred\   \opair {k_2} {P_2} } \quad
	%	\infer{ \opair {n} {M}  \pwred\   \opair {n+k} {M'}}{ \opair 0 {M}  \pwred\   \opair {k} {M'}}
	%	\]
	%	\caption{}
	%\end{figure*}

The pair $ 	\brack{(\Lambda_\tick,\red), \obs}$ is a QARS where we   observe the payoff, \ie~$ 	\obs \opair  n M = n $.
We  now prove (using \Cref{thm:ACompl})  that $\wred=\wredbv \cup \; \wredx{\tick}$ is  asymptotically complete for $\red$.

\begin{lemma}\label{lem:payoff}	For every pair $\m = \opair n M$, $\m \ \red\conv{} \cpo{n}$ implies 	$\m \ \wred\conv{} \cpo{n}$, because 
	\begin{itemize}
		\item $\weak$-factorization of $\red$: if $\m \ \red^* \ \n$ then $\m \ \wred^*\cdot \nwred^* \ \n$;
		\item $\obs$-neutrality : if	$\m \ \nwred \ \m'$ then  $\obs {(\m)} = \obs {(\m')}$.
	\end{itemize}
\end{lemma}

Weak reduction $\wred$ however does \emph{not} have a unique limit, as \Cref{ex:payoff_lim} below illustrates. An unsatisfactory  solution would be to  fix   a deterministic evaluation order (left or right, as in point 1. below), making the limit easy to predict but also rather \emph{arbitrary}.
%$\wred$ has not a unique limit.
	\begin{example}\label{ex:payoff_lim}
		Consider  $M=(\Delta\Delta)(\dtick\dtick)$, where $\dtick= \lam x.\tick(xx) $, and let
		$\m=\opair{0}{M}$. \begin{enumerate}
			\item By fixing  left  (resp.~right)  evaluation,   $\wLim (\m,
			\lred)=\{0\} $  (resp.~$\wLim (\m, \xredx{\rsym}{})=\{\infty\} $).
			\item By choosing   a redex in 
			unspecified order, we have an uncountable number of $\wred$-sequences,
			leading to $\wLim (\m, \wred)=\{0,1, \dots\,\infty\} = \Nat^\infty $.
		\end{enumerate}
	\end{example}
A way out is to  proceed somehow similarly to  \Cref{sec:PCbV}. 
If we examine  more closely 
	the  set  of limits associated with $\wred$, we realize that $\xLim{\obs}{(\m,\wred)}$ does have a greatest element.
	Thus  $\sem {\m}$ can naturally be defined as   the best possible payoff from $\m$.
We prove that  parallel reduction  $\xredx{\pw}{}$  (given in \Cref{fig:parallel_payoff})
is a 
 (\emph{multistep}) strategy  which is guaranteed to  compute  $\m$. Indeed,  
	it is easy to verify that  $\xredx{\pw}{}$ is asymptotically complete for $\wred$. By composing with  \Cref{lem:payoff} we have that $\xredx{\pw}{} $ is asymptotically complete for $\red$ (point 1. below). 
	\begin{lemma}\label{l:completeness-uniqueness-tick}
		\begin{enumerate}
			\item \emph{Asymptotic Completeness.} 
			If $\opair{k}{M} \;\wred\conv{} \cpo{ n}$ then $\opair{k}{M} \xredx{\pw}{} \conv{} \cpo {n'}$ and $\cpo{n} 	\leq \cpo{n'}$. That is, $\xredx{\pw}{} $ is asymptotically complete for $\wred$  and  (by  \Cref{lem:payoff}) for $\red$.
			\item \emph{Unique Limit. } The reduction $ \xredx{\pw}{}  $ is deterministic.
		\end{enumerate}
	\end{lemma}
Since (by points 1. and 2. in \Cref{l:completeness-uniqueness-tick}) both conditions of \Cref{thm:main} are verified, we have:

	\begin{theorem}[Main, payoff] 
	Given the   QARS $\brack{(\Nat\times \Lambda_{\tick}, \red), \obs}$, for each pair  $\m = \opair{k}{M}$, 
%\begin{center}
		$\sem {\m}$ 
	is  defined, and $\m \xredx{\pw}{}  \conv{}  \sem {\m} $.
%\end{center}
Hence, 	multistep reduction 
	$\xredx{\pw}{}$ is  asymptotically normalizing   for $\wred$ and for $\red$ (\Cref{rem:Anormalizing}).
\end{theorem}

	%%%%%%%%%%%%%%%%%
	
	\SLV{}{
	\begin{restatable}{proposition}{prop:payoff}
		\begin{enumerate}
%			\item \emph{$\wred$ is Asymptotically Complete.}   	$M \red\conv{} \cpo{n}$ implies 	$M \wred\conv{} \cpo{n}$, because 
%			\begin{itemize}
%				\item $\weak$-Factorization of $\red$: $M \red^*N$ implies $M \wred^*\cdot \nwred^* N$
%				\item $\obs$-neutrality : 	$M\nwred M'$ then  $\obs {M} = \obs {M'}$.
%			\end{itemize}
		
				\item \emph{$ \xredx{\pw}{} $ is Asymptotically Complete.}   
			$(k:M) \wred\conv{} \cpo{ n}$ implies 	\big($k:M \xredx{\pw}{}\ \conv{} \cpo {n'}$ and $\cpo{n} 	\leq \cpo{n'}$\big).
			\item \emph{Unique Limit. } The reduction $ \xredx{\pw}{}  $ is deterministic.
		\end{enumerate}
	\end{restatable}

\begin{thm}[Main, payoff]	Given the   QARS $\brack{(\Nat\times \Lambda_{\tick}, \red), \obs}$,
	%	 it is the  limit of the  $\xredx{\pw}{}$-sequence. 
	%
	multistep reduction 
	$\xredx{\pw}{}$ is  asymptotically normalizing   for $\wred$ and  for $\red$.
\end{thm}
}

\ArX{}{}{
	\paragraph{$\lam $-calculus with outputs.} The calculus in   \Cref{ex:output} can be formalized in a similar way to the payoff calculus.
	We can define	$\obs (\opair s M) \eq s$.
	As already noted, $\wred$ is not confluent,  and given a pair $\m$,  the set of limits may contain uncountably many different   elements. Still,  the reduction  has interesting    properties, {which appear when looking not directly at the string $s$ itself, but  at its length $|s|$.}  This way, 
	one  can  transfer the  results from the payoff calculus.}

% !TEX root = main_2022.tex
\section{Conclusions}\label{sec:conclusions}
We  propose a method to study  completeness and  normalization when the result of computation is \emph{asymptotic}.
Our  techniques  abstract from    details specific to the calculus under study---they are therefore of \emph{general} application.
The robustness of the method is witnessed by its ability to deal with different  settings and \emph{different notions of asymptotic computation}. 
%In particular we illustrate its use  with  two diverse notions of asymptotic limit: distributions on normal forms in the setting of CbV probabilistic $\lam$-calculus, 
%and infinite normal forms (B\"ohm Trees) in the classical setting of  CbN $\lam$-calculus.

The application to {probabilistic $\lam$-calculus} yields a result of independent interest: a theorem of \emph{asymptotic normalization},  both for CbV and  CbN probabilistic $\lambda$-calculi. Remarkably, the same definitions and proof techniques apply \emph{uniformly} to both. In the paper 
we prefer to give  the details for the  CbV calculus, which is arguably a  more
natural one in presence of effects.
%
%The method is robust, being able to work with differet notions of asymptotic computation, and  both \CbN and \CbV evaluations.
%To  illustrate this, we study two different notions of asymptotic limit: distributions on normal forms in the setting of probabilistic $\lam$-calculus, 
%and infinite normal forms (B\"ohm Trees) in the setting of (classical) CbN $\lam$-calculus.

%Remarkably, all results we give in this paper hold for both call-by-value and call-by-name
%evaluation, but we prefer to give all the details of the a system of the former kind, arguably a more
%natural one in presence of effects.

%introduces a strategy for probabilistic $\lam$-calculus which we show to by asymptotically normalizing, \ie complete for  \emph{normal forms}, in contrast to previous work which is limited to weak (CbV) or head (CbN) normal forms. 

\paragraph{Related work.}  QARS, proposed  in \cite{parsLMCS} in the setting of probabilistic rewriting, refine Ariola and Blom's  ARSI \cite{AriolaBlom02}. 
The techniques  in \Cref{sec:strategies} are an original contribution of this paper. 
% Our \Cref{thm:ACompl} refines for arbitrary observations and generalizes  to asymptotic computation an ARS  technique for finitary normalization studied in \cite{AccattoliFG19,OostromT16,LOrevisited}.
Our \Cref{thm:ACompl} generalizes an ARS  technique for finitary normalization 
(studied in \cite{AccattoliFG19,OostromT16,LOrevisited}) 
to asymptotic computation, refining it for arbitrary observations.
%(studied in \cite{AccattoliFG19,OostromT16,LOrevisited} and also used in \cite{FaggianGuerrieri21}) 
%non citerei Fossacs, perche' qui stiamo parlando di tecniche ARS 

%We here adopt the  notion of set of  limits (\Cref{sec:QARSdefs})---and  of greatest limit as analog of unique normal form---introduced in \cite{pars}, in the setting of probabilistic rewriting. 

	The study of  reduction  strategies in a probabilistic $\lam$-calculus where the notion of reduction is  general---rather than simply fixing  a deterministic reduction---started in \cite{FaggianRonchi} (CbV and CbN) 
and \cite{Leventis19} (CbN).  {Asymptotic completeness} is  there  established only for \emph{surface}  normal forms (values in closed CbV, hnf's in CbN).
Strategies that are complete for  \emph{full normal forms} (which we treat and solve here) are more difficult to study than head or weak reduction, especially in the CbV setting. The question of defining  such a strategy was left open in \cite[Remark 27]{FaggianRonchi}. 
We stress that our  technique would also yield  a simpler proof  of  the results in \cite{FaggianRonchi,Leventis19}, where  confluence is  used  to establish that a greatest limit exists. 
The (non-trivial) proofs {there} use   properties that are \emph{specific} to  probability distributions.
The method we propose here avoids technical issues, it  is much  simpler,  and it is general, in that it can be applied to other settings.

%
%To some extent,  in the CbN setting, one can  see 
%Leventis construction of Probabilistic B{\"o}hm trees \cite{Leventis18} as \emph{implicitly} yielding a CbN  asymptotic normalization. Our approach  is   direct, and \emph{operational}. Moreover, it is  \emph{uniform} for both  CbN and CbV.

Finally, we mention that forgoing confluence and studying uniqueness of normal forms via a complete subreduction is a route already employed  in the context of infinitary $\lam$-calculi \cite{BerarducciIntrigila96, BarendregtKlop}.

\newpage

\bibliography{biblioPARS}

\clearpage
\appendix
% !TEX root = main_2022.tex
\appendix
\section*{APPENDIX}
We include some  proofs and  details that have been omitted in the article.

\ArX
{
	\paragraph{One more example.}\label{app:QARS}
%	\paragraph{\Cref{sec:QARS}, QARS.}
	 $\wLim(t,\red)$ may have a lub but not a maximum---similarly to $\Nat$. 
	 \begin{example}[\Cref{sec:QARS}, QARS]\label{ex:nomax}We revisit 
	 	\Cref{ex:strings},  allowing full reduction $\redbv$. Let $\obs_p(\opair s M)=s$ if $M\in \Val$, $\bot$ otherwise.
	 	The pair  $\m \eq \opair  \epsilon  {(\lam z. I)(\lam z.\dout 0\dout 0)}$ has countably many limits, but not a greatest one, because all strings in  $ \xLim{\obs_p}(\m, \redbv)$ are finite.
	 \end{example}
%	\begin{example}\label{ex:nomax}Let us revisit 
%		\Cref{ex:strings}, now allowing full reduction $\redbv$. We define $\obs_p(\opair s M)=s$ if $M\in \Val$, $\bot$ otherwise.
%		The pair  $\m \eq \opair  \epsilon  {(\lam z. I)(\lam z.\dout 0\dout 0)}$ has countably many limits, but not a greatest one, because all strings in  $ \xLim{\obs_p}(\m, \redbv)$ are finite.
%	\end{example}
}{}

\ArX{	}{
		\paragraph{Some  more examples.}\label{app:QARS}
	%	\paragraph{\Cref{sec:QARS}, QARS.}
	$\wLim(t,\red)$ may have a lub but not a maximum---similarly to $\Nat$. 
	\begin{example}[\Cref{sec:QARS}, QARS]\label{ex:nomax}We revisit 
		\Cref{ex:strings}, now allowing full reduction $\redbv$. Let $\obs_p(\opair s M)=s$ if $M\in \Val$, $\bot$ otherwise.
		The pair  $\m \eq \opair  \epsilon  {(\lam z. I)(\lam z.\dout 0\dout 0)}$ has countably many limits, but not a greatest one, because all strings in  $ \xLim{\obs_p}(\m, \redbv)$ are finite.
	\end{example}
	
%	\paragraph{\Cref{sec:wRD},  Weighted Random Descent.}
	\begin{example}[\Cref{sec:wRD},  Weighted Random Descent, strings] By fine tuning the notion of observation, one can establish  non trivial properties of the   calculus $\OutC {\alphabet}$  in \Cref{ex:strings} (here reduction is  CbV and weak), such as the following, which is also an example of weighted RD.
		If $\m=\opair s M$ has a terminating computation   $\m  \ \wred^* \ \opair {s'} V$ for some value $V$, then \emph{all maximal $\wred$-sequences from $\m$ terminate} (in the same number of steps) 
		and all end exactly with \emph{the same value $V$ and  a string of the same length as $s'$}.	
		%	all maximal weak reduction sequences from the same pair $\m=(s:M)$ have the same number of steps, and---if $\m$ has a terminating computation (\ie  $\m  \wred^* \ s':V$)--- \emph{all maximal sequences from $\m$ end exactly with the same value $V$ and  a string of the same length as $s'$}.		
	\end{example}
}
}

\SLV{}{
	\paragraph{\Cref{sec:P_Anormalization}: CbV asymptotic normalization.}
	\begin{example}\label{ex:CBVPfull}Some simple examples will help to see how the normalizing strategy works, and how it differs from surface reduction.  
		Recall that $\betav$-reduction is unrestricted, so for example  $M= \lam z. (Iz) \redbv \lam z.z$. Instead, $ \lam z. (Iz) \not\sred$, because surface reduction cannot fire under abstraction. So 
		\emph{surface reduction is not a complete strategy} w.r.t. $\betav$-normal forms.
		
		A direct  consequence is that surface reduction is not informative about normalization, as it  produces ``false positive''. For example,  $N= \lam z. \Delta\Delta$ is  diverging  w.r.t. $\betav$-reduction, but it is a surface normal form.
		Let us now incept probability (with the terms  $M$ and $N$ as above).
		\begin{enumerate}		
			\item  Let $R= (\lam x. M \oplus xx)(\lam x. M \oplus xx) $.  Then		
			$\mdist{R} \llfull  \mdist{M \oplus R} \llfull \mdist{ \two M, \two R} \llfull\mdist{ \two I, \two M \oplus R}  \llfull  \mdist{ \two I, \four M,  \four R} \llfull \cdots $. At the limit,  $R$ converges with probability $1$ to   $I$, as wanted. 
			%that is $\mdist{R} \llfull\conv{}   \{I^\two\}$.
			
			\item The term $ S= (\lam x.  N \oplus xx)  (\lam x. N \oplus xx)$  converges to  normal form with probability $0$.  One can easily check that  
			$\mdist{S} \llfull\conv{}  \zero$.
			
			\item The term $ S'= (\lam x.  (N \oplus I) \oplus xx)  (\lam x. (N\oplus I )\oplus xx)$  converges with probability $\two$ to the normal form $I$.  One can easily check that  
			$\mdist{S'} \llfull\conv{}  \{ I ^ \two\}$.
			
			%	\item $( \lam x. (N \oplus I) \oplus xx)  (\lam x. (N\oplus I )\oplus xx) \sRed \mdist{\two M, \two N}$  so this term converges with probability $1$.
			
			\ArX{}
			{	\item \RED{ Notice  that  from $S'$, the reduction $\llRed$ (instead of $\llfull$) would  admit also reduction sequences which produces no normal form, 		for example  $S' \llRed^* \mdist{\four N, \two I, S'} \llRed \mdist{\four N, \four I, \two S'} \llRed \mdist{\four N, \four  I,  \two S'} \cdots$. By persistently reducing only $N$, we have 	$\mdist{S'} \llRed\conv{}  \zero$. }}
		\end{enumerate}
	\end{example}
}

\paragraph{Surface reduction.}
Everywhere in the appendix, we fix surface reduction to be as follows. 
\begin{itemize}
	\item CbN ($\betab=\beta$): $\surf =\head$   (the contextual closure of $\hh$).
	\item CbV ($\betab=\betav$): $\surf= \weak$   (the contextual closure of $\ww$).
\end{itemize}

	\section{Properties of surface normal forms}
	We will use extensively the following easy fact.
	\begin{lemma}[Surface normal forms]\label{lem:ex_nf}\label{lem:snf} 
		$M$ is $\weak$-normal (resp. $\head$-normal) if there is no redex $R$ such that 	$M=\ww\hole R$ (resp. $M=\hh\hole R$).
		\begin{enumerate}
			\item \CbV. 
			Assume 
			$M \nwredbv M'$.
			%	{$M$ contains no $\weak$-redex  $\Leftrightarrow$ $M'$ contains no  $\weak$-redex.}	
			$M$ is   $\weak$-normal  $\Leftrightarrow$ $M'$ is     $\weak$-normal.
			
			\item \CbN. Assume $M \nhredb M'$.	
			%{$M$ contains no $\head$-redex  $\Leftrightarrow$ $M'$ contains no  $\head$-redex.}
			$M$ is   $\head$-normal  $\Leftrightarrow$ $M'$ is     $\head$-normal.
			
			%	\item Point 1. and 2. both  hold also taking for  $\R$ all the  terms of shape  $\opp{\dots M_i \dots}$
		\end{enumerate}
	\end{lemma}

%\ArX{}{\input{99_Appendix_Surface}}
% !TEX root = main_2022.tex

\ArX{\section{ \Cref{sec:normalization}:   properties of unbiased reduction}\label{app:normalization}
 The properties of $\lsred$ (	\Cref{thm:Ufactorization},
 \Cref{lem:Uproperties}, \Cref{thm:Unormalization},  and those stated  here) are proved in \cite{long}.

%	With the same assumptions as in \Cref{def:liberal}, let $\betab\in \{\beta,\betav\}$. 
%	The relation $\nlsredbb$ is the complement of unbiased reduction $\lsredbb$, \ie~$\nlsredbb \ = \ \redbb \smallsetminus \lsredbb$.
%	\begin{theorem}[$\los$-Factorization]
%		\label{thm:Ufactorization} 
%		$ 	M\redbb^*  N  \timplies M \lsredbb^* \cdot  \nlsredbb^* N $
%	\end{theorem}

\begin{lemma}[Normal forms]\label{lem:nfU}\label{lem:Unormal}
	If   $U\nlsredbb N$, then $N$ is not  $\betab$-normal.
\end{lemma}

}{}

\ArX{}{
\section{ \Cref{sec:normalization}:   properties of unbiased reduction}\label{app:normalization}

With the same assumptions as in \Cref{def:liberal}, let $(\betab, \surf)\in \{(\beta, \head),(\betav, \weak)\}$. 
{The relation $\lsredbb$ is as in \Cref{def:liberal}.}
The relation $\nlsredbb$ is the complement of $\lsredbb$, \ie   $\nlsredbb \ = \ \redbb \smallsetminus \lsredbb$.

We start with a basic remark.
\begin{remark}
	\label{rem:Usurf}
	%	{Let $(\betab, \surf)\in \{(\beta, \head),(\betav, \weak)\}$}.
	Fixed $(\betab, \surf)\in \{(\beta, \head),(\betav, \weak)\}$, 	we have  $\sred \subseteq \lsredbb$. Hence, $\nlsredbb \subseteq \nsred$.
\end{remark}

\paragraph{Diamonds.}

\begin{lemma}[RD-diamond]\label{lem:Udiamond} With the same assumptions as in \Cref{def:liberal}, let $(\betab, \surf)\in \{(\beta, \head),(\betav, \weak)\}$. $\lsredbb$ has the  RD-diamond property of \Cref{fact:diamond}.
\end{lemma}

\begin{proof}We already know that $\sred$ ($\surf\in \{\head,\weak\}$) has the $RD$-diamond property.
	If  $M$ has $\surf$-redexes, then   $M_1 \xrevredx{\surf}{} M \sred M_2$, and the claim holds.
	Otherwise, if  $M$ has no $\surf$-redexes, then  $M$, $M_1$ and $M_2$ have the same shape. We have either  $M=\lam x. P$ or $M= PQ$. 
	\begin{itemize}
		\item Case  $M=\lam x. P$. Then $M_1=\lam x.P_1$, $M_2=\lam x.P_2$, and $P_1 \xrevredx{\los}{} P \lsred P_2$, and we conclude by \ih.
		\item Case  $M=PQ$.  Three cases are possible:
		\begin{itemize}
			\item  $P_1Q \xrevredx{\los}{} PQ \lsred P_2Q$, and we conclude by \ih.
			\item  $PQ_1 \xrevredx{\los}{} PQ \lsred PQ_2$, and we conclude by \ih.
			\item  $P'Q \xrevredx{\los}{} PQ \lsred PQ'$, where $P\lsred P'$ and $Q\lsred Q'$. Then we conclude
			$P'Q \sred PQ' \xrevredx{\los}{}  PQ'$.
		\end{itemize}
		
		\item 	 Case  $M=\opp{ P_1 \dots P_k}$. 
		\begin{itemize}
			\item If $ \opp{\dots P_i'\dots } \xrevredx{\los}{} \opp{\dots P_i \dots } \lsred \opp{\dots P_i'' \dots }$,  we conclude by \ih.
			\item If $ \opp{\dots P_i' \dots P_j \dots} \xrevredx{\los}{} \opp{\dots P_i \dots P_j \dots } \lsred \opp{\dots P_i \dots P_j' \dots} $, then\\ 
			$ \opp{\dots P_i' \dots P_j \dots}    \sred  \opp{\dots P_i' \dots P_j' \dots}  \xrevredx{\los}{}   \opp{\dots P_i \dots P_j' \dots} $
			\qedhere
		\end{itemize}
	\end{itemize}

\end{proof}

\paragraph{Factorization.}We only sketch the proof of 
 $\los$-factorization, which  is straightforward to establish from  $\surf$-Factorization, using Mitschke argument \cite{Mitschke79}.

\begin{theorem}[$\los$-Factorization]
	\label{thm:Ufactorization} 
	With the same assumptions as in \Cref{def:liberal}, let $(\betab, \surf)\in \{(\beta, \head),(\betav, \weak)\}$. Then:
	\begin{equation*}\tag{$\los$-Factorization }
		M\redbb^*  N  \timplies M \lsredbb^* \cdot  \nlsredbb^* N
	\end{equation*}
\end{theorem}

\begin{proof}By induction on the term $N$. 
	From $M \redbb^* N$, 
	by $\surf$-Factorization, we have that 
	\[M\sredbb^* U  \nsredbb^* N.\]
	Since $\nsredbb$-steps preserve the shape of terms,   all terms in the sequence $U\nsredbb^* N$ have the same shape. \SLV{}{(\Cref{fact:shape}).}
	\begin{itemize}
		\item If $U$ contains any $\surf$-redex, so does every term in the sequence $U\nsredbb^* N$   by \Cref{lem:snf}, and so $U\nlsredbb^* N$.  Hence the claim.

		\item Assume $U$ contains no $\surf$-redex.
		We examine the possible shape of $N$, and conclude by \ih.
	\end{itemize}		
\end{proof}

\SLV{
}{\begin{theorem}[$\los$-Factorization]\label{thm:normalizing_strategy}
		Assume $M\redbv^*  N $. Then  \begin{center}
			$M \lsredbv^* \cdot  \nlsredbv^* N$.
		\end{center}
	\end{theorem}

	\begin{proof}
		By induction on the term $N$. 
		
		From $M \redbv^* N$, 
		by $\weak$-Factorization, we have that 
		\begin{center}
			$M\wredbv^* U  \nwredbv^* N$.
		\end{center} Since $U\nwredbv^* N$,  all terms in the $\nwredbv$-sequence have the same shape.
		
		\begin{itemize}
			\item If $U$ contains any $\weak$-redex , so does every term in the sequence $U\nwredbv^* N$  (by Result \ref{lem:surface_nf}), and so $U\nlsred^* N$.  Hence the claim.

			\item Assume $U$ contains no $\weak$-redex .
			We examine the possible shapes of $N$:
			\begin{enumerate}
				\item $N=x$. Then  $U=x=N$, and trivially $M\lsredbv^* M$.
				\item $N=(\lam x.N_0)$. Then  $U=(\lam x. U_0)$, hence  $U_0 \redbv^*  N_0$. 
				
				By \ih  $U_0\lsredbv^* \cdot \nlsredbv^* N_0$. Observe that $\lam x. U_0$ is $\weak$-normal, hence  $\lam x.U_0\lsredbv^* \cdot \nlsredbv^* \lam x.N_0$, and so
				\[		M\lsredbv^* (\lam x. U_0)\lsredbv^*\cdot \nlsredbv^* (\lam x.N_0)=N.\]
				
				\item $N=N_1N_2$. 
				Then $U=U_1U_2$, with $U_1\redbv^* N_1$ and $U_2\redbv^* N_2$. 	
				
				By \ih 
				$U_1 \lsredbv^* S_1  \nlsredbv^* N_1$ and $U_2\lsredbv^* S_2  \nlsredbv^*N_2$. 
				
				Since $U$ contains no $\weak$-redex, then  neither $U_1$ nor $U_2$ contain any $\weak$-redex, and $U_1U_2$ is not a redex. By 
				repeatedly using Result \ref{lem:ex_nf} (from $U_1U_2$, forwards), we have that---in both sequences---each step  is a $\nwred$-step.
				Therefore all terms in the sequence $\fseq{U_1 \dots S_1  \dots N_1}$ have the same shape, all terms in the sequence $\fseq{U_2 \dots S_2  \dots N_2}$ have the same shape, and 
				no term in either sequence contains any $\weak$-redex.	
				
				It follows that in the  four  sequences \begin{center}
					$\fseq{U_1U_2\dots  S_1U_2}$,  $ \fseq{S_1U_2 \dots S_1U_2}$, $\fseq{S_1S_2  \dots N_1S_2}$, and $\fseq{N_1S_2  \dots  N_1N_2}$
				\end{center}
				each step is also  a $\nwred$-step, and (by Result \ref{lem:ex_nf}), and each term contains 
				no $\weak$-redex. 
				Therefore  
				\[M\lsredbv^* ~U_1U_2~ \lsredbv^*  S_1U_2 \lsredbv^*S_1 S_2 \nlsredbv^*  N_1S_2  \nlsredbv^*  N_1N_2=N.\]

				\item $N=\opp{ N_1 \dots N_k }$. Then $U=\opp{U_1\dots U_k  }$, with $U_i\redbv^* N_i ~(\forall \ik)$.	
				By \ih 
				$U_i \lsredbv^* S_i  \nlsredbv^* N_i ~(\forall \ik)$. We conclude 
				\[\opp{U_1\dots U_k } \lsredbv^* \opp{S_1 \dots S_k }  \nlsredbv^* \opp{N_1\dots N_k } \]

			\end{enumerate}
		\end{itemize}	
\end{proof}}

\paragraph{Normalization.}

\begin{lemma}[Normal forms]\label{lem:nfU}\label{lem:Unormal}
	If   $U\nlsredbb N$, then $N$ is not  $\betab$-normal.
\end{lemma}

\begin{proof} By induction on the shape of $U$, observing that $N$ and $U$ have the same shape.
	
	%	 By \Cref{lem:snf}), both $N$ and $U$ are $\surf$-normal, or neither is.
	
	\begin{itemize}
		\item Assume $U$ is not $\surf$-normal.   Then  (by \Cref{lem:snf}) $N$ is not $\surf$-normal,  and a fortiori not  $\betab$-normal.
		\item Assume $U$ is  $\surf$-normal. We examine its shape.
		\begin{itemize}
			\item Case  $U=U_1U_2$ and 	$N=N_1N_2$. Then either (i) $U_1\redbb N_1$ or (ii) $U_2\redbb N_2$. Consider (i). 	
			Necessarily, 
			$U_1\nlsredbb N_1$ (because  	$U_1\lsredbb N_1$ would  imply  $U_1U_2\lsredbb N_1N_2$ by \Cref{def:U}) and therefore  by \ih  $N_1$  is not $\betab$-normal,  and neither is $N=N_1N_2$. Similarly for (ii).
			\item The other cases are similar (and simpler).
			\qedhere
		\end{itemize}
	\end{itemize}
	
\end{proof}

\Uproperties*

\begin{proof}
	\begin{enumerate}
		\item This is \Cref{lem:Udiamond}.
		
		\item Since  $\lsredbb \subseteq \redbb$, every $\redbb$-normal form is $\lsredbb$-normal.
		Conversely, if $M$ is $\lsredbb$-normal, then $M$ is  $\sred$-normal and a straightforward induction on $M$ shows that $M$ is $\redbb$-normal.
		
		\item Proof 1:	 By \Cref{thm:Ufactorization}, $M \lsredbb^* U \nlsredbb^* N$.
		Since $N$ is $\betab$-normal, by \Cref{lem:nfU}   $P = N$.
		
		Proof 2:  By induction on the shape of $N$,  using $\surf$-factorization.
		\qedhere
	\end{enumerate}
\end{proof}

The fact that $\lsred$ is a normalizing strategy (\Cref{thm:Unormalization}) follows  from \Cref{lem:Uproperties}, because
$\lsred$ is complete w.r.t. normal forms, and it is  RD-diamond (and hence uniformly normalizing).

}

\subsection{A  parallel variant of \liberal reduction.} \label{sec:parallelU}
%\paragraph{A  parallel variant of \liberal reduction.} 

%
%\SLV{Parallel  (multi-step) reduction $ \xredx {\pd}{}$ is easily defined.}

Given $(\Lambda,\redbb)$ and $\sred$ as  in \Cref{def:liberal}, a parallel   version $ \xredx {\pd}{\betab}$ is easily defined. The idea here is that once a term  is  $\sred$-normal,  iteration of the reduction process can be performed in any arbitrary order,  or in \emph{parallel}. Recall that here  $(\betab, \surf)\in \{(\beta, \head),(\betav, \weak)\}$.
\begin{enumerate}
	\item If  $M\sred{} M'$ then $M \xredx {\pd}{\betab} M'$ ($M$ is not $\surf$-normal);
	\item If $M \not \red$ then $M \xredx{\pd}{\betab} M$ ($M$ is $\to$-normal);
	\item Otherwise: \\
	{\small 
		$\infer{M \defeq \lam x. P \xredx{\pd}{\betab}   \lam x.P'}{P \xredx{\pd}{\betab} P'} 
		\qquad
		\infer{M \defeq P_1P_2 \xredx{\pd}{\betab}   P_1'P_2'}{ P_1\xredx{\pd}{\betab}   P_1' & P_2\xredx{\pd}{\betab}   P_2'}
		\qquad
		\infer{M \defeq \opp{ P_1, \dots, P_k}\xredx{\pd}{\betab}   \opp{ P_1', \dots, P_k'}}{ (P_i\xredx{\pd}{\betab}  P_i')_{1 \leq i \leq k}} 
		$
	}
\end{enumerate}
%\begin{prop}$\xredx{\pd}{}  $ is a (multi-step) normalizing strategy for $\red$.
%\end{prop}
Rule 2. makes the relation reflexive on normal forms and \emph{only} on normal forms---this is a harmless shortcut in order to give a compact and  neat formulation.
%It essentially fires \emph{all} the redexes that are not contained in other redexes.

The  (multistep) reduction $\xredx{\pd}{\betab}$  is guaranteed to reach the $\redbb$-\nnf, if any exists.

\begin{lemma}\label{lem:parallelU}Let $\betab\in\{\beta,\betav\}$
	\begin{enumerate}
		\item  If  $M\xredx{\pd}{\betab} \ N $ then $M\lsredbb^*N$. Therefore, $M\xredx{\pd}{\betab}^* \ N $ implies $M\lsredbb^*N$.
		\item  If $M\lsredbb^*N$ then  there exists $N'$ such that $M\xredx{\pd}{\betab}^*N' $ and $N \lsredbb^*N'$.
	\end{enumerate}
\end{lemma}
\ArX{}{
	\proof\begin{enumerate}
		\item Easy induction on $M$. % on the number of $ \xredx{\pd}{\betab} $ steps.
		\item By induction on $N$, using  $\surf$-factorization  and the fact that non-surface steps preserve the shape of terms.  By the former $M\sredbb^* U \nsredbb^* N$.  By the latter, $U$ has the same shape as $N$; we then  examine the possible cases, and   conclude by \ih.
		\begin{itemize}
			\item $N=\lam x. N_0$ and $U=\lam x. U_0$. By definition of $\lsred$, $U_0 \lsredbb^*N_0$, and we conclude by \ih.
			\item $N= N_1N_2$ and $U=U_1U_2$. It holds $U_1 \lsredbb^* N_1$ and $U_2 \lsredbb^* N_2$. By \ih, for $i\in \{1,2\}$
			$U_i\xredx{\pd}{\betab}^*N_i' $ and $N_i \lsredbb^*N_i'$. 
			We now extend the two $ \xredx{\pd}{\betab} $-sequences so that they have the same length (recall that $ \xredx{\pd}{\betab} $ is reflexive on $\bbeta$-\nnf), obtaining $U_i\xredx{\pd}{\betab}^*N_i'  \xredx{\pd}{\betab}^*N_i'' $ .  We then can conclude that 
			$U_1U_2\xredx{\pd}{\betab}^*N_1'N_2'  \xredx{\pd}{\betab}^*N_1''N_2'' $.
			
			By using point 1., we also have 
			$N_i' \lsredbb^* N_i''$, and so $U_1U_2\lsredbb^*N_i''N_2''$.  Hence the claim.
		\end{itemize}
	\end{enumerate}
}

\begin{cor}[$\Nnf$-completeness] Let $\betab\in\{\beta,\betav\}$ and $N$ be $\redbb$-normal.  ~$M\redbb^*N$  if and only if $M\xredx{\pd}{\betab}^*N $.
\end{cor}
%\begin{proof}$ \Leftarrow $:  immediate.  $\Rightarrow$: by induction on $N$. The claim is again an 
%	easy consequence of Surface Factorization and the fact that non-surface step preserve the shape of terms.  By the former $M\sredbb^* U \nsredbb^* N$.  By the latter, $U$ has the same shape as $N$; we then  examining the possible cases, and   conclude by \ih.
%\end{proof}

% !TEX root = main_2022.tex

\section{Proofs of \Cref{sec:P_Anormalization}: Asymptotic Normalization for \texorpdfstring{$\PLambda^\cbv$}{LambdaOplusCbV}}
%\noindent
%\textbf{The reduction}  $\llred$.  Notice that the definition of $\llred$ is the same in CbV and CbN, taking into account that 
%in CbN  $\betav$ is replaced with $\beta$ and $\surf$ is $\head$.

Notice also that the definition of the reductions $\llred, \llRed$ can be given in the same way also in CbN, by replacing 
 $\betav$  with $\beta$ (surface steps are here head steps).

\paragraph{Properties.} We freely use the following  fact.   
\begin{fact}\label{fact:trans}
	$\mdist{M} \Redbv\mdist{M'} \quad  \text{ iff  }  \quad M  \redx{\betav}^= {M'}$ ,
	where \begin{itemize}
		\item on the l.h.s.  we have $(\MDST{\PLambda}, \Red)$,  and 
		\item on the r.h.s.  the  CbV  $\lam$-calculus $(\PLambda, \redbv)$, as defined in \Cref{sec:lambda}.
	\end{itemize}
%Notice in particular that  $\mdist{M}\nsRed  \mdist{M'} \text{ iff } M \nsredbv^= M'$.

%	With a slight abuse of notation, we have written in the same way these two relations; the abuse is justified by the  equivalence above.
%	Moreover,  $\mdist{M}\sRedbv  \mdist{M'} \text{ iff } M \sredbv^= M'$, and similarly for the other relations.
%$\mdist{M}\nsRedbv  \mdist{M'}  \text{ iff } M \nsredbv^= M'$.
\end{fact}

%\begin{fact}\label{fact:trans}
%\[M \nsredx{\ \betav} \mdist{M'} \quad  \text{ \iff }  \quad M  \nsredx{\ \betav} {M'} \quad  (\#)\]	
%where on the l.h.s.  we have the 
%relation $\redbv \subseteq \PLambda \times \MDST{\PLambda} $ as defined in \Cref{fig:reductions} and on the r.h.s.  the usual 
%$\betav$-reduction in the CbV  $\lam$-calculus $(\PLambda, \redbv)$, as defined in \Cref{sec:lambda}.
%With a slight abuse of notation, we have written in the same way these two relations; the abuse is justified by .
%\end{fact}

\paragraph{Factorization and Neutrality.}  
Recall that  $\Snf$ denotes the set of the \textbf{surface normal forms} of $\red \eq (\redbv\cup \redo)$. $M\in \PLambda$ is $\surf$-normal if $M\not\!\sred$. That is $M\not\redo$ and $M\not\sredbv$. 

\begin{lemma}[\snf propagation]\label{lem:snf_propagation} If  $M$ is $\surf$-normal and $M \red M'$, then $M'$ is $\surf$-normal.
\end{lemma}
\ArX{}
{\begin{proof}
If $M$ is $\surf$-normal and $M \red M'$, then  $M \nsredbv M'$. By \cref{lem:snf}, $M'$ is $\surf$-normal.
\end{proof}}

	\begin{prop}[$\llsmall$-Factorization of $\Red$.] In $\PLambda^\cbv$:
		$\m \Red^*\n$ implies $\m \llRed^*\cdot \iRed^* \n$
	\end{prop}
	\begin{proof} 	
%		Let us  write $\Snf_\oplus := \{S \in \PLambda \mid S\not \redo\}$, $\Snf_{\betav} := \{ S \in \PLambda \mid S\not \sredx{\betav} \}$ The set of surface normal forms is $\Snf= \Snf_\oplus \cap \Snf_{\betav}$. 
%		
		In the proof, we use freely \Cref{fact:trans}.
	By  surface factorization of $\Red$ (proved in \cite{FaggianRonchi}), $\m \Red^*\n  \timplies \m \sRed^*  \td  \nsRed^* \n$ for some $\td$.		
	From this we have: \begin{itemize}
			\item $\m \llRed^*\td$. Because if  $M \sred \ \r$ then also  $  M \llred \ \r $. 
			\item $  \td  \xRedx{\neg\surf}{\ \betav}^* \n $. Because if $M \nsred \ \r$ then necessarily $M {\nsredx{ \betav}} \r$. Moreover,   $\r=\mdist{M'}$.  		
		\end{itemize}	
	
	Let $\td = \mdist{\dots p_iT_i \dots }_{\iI}$.  Then necessarily, $\n=\mdist{\dots p_iN_i \dots }_{\iI}$ and     
	$\mdist{T_i }\nsRedbv^* \mdist{ N_i}$
	and so also  ${T_i }\nsredbv^* { N_i}$. For each $T_i$, we examine if $T_i$ is  $\surf$-normal or not ($\Snf$ being the set of \snf's).
		\begin{enumerate}
				\item $T_i \in \Snf$.  
				By  ${T_i }\nsredbv^* { N_i}$ and   $\los$-factorization of $\redbv$ (\Cref{thm:Ufactorization}),			${T_i }\lsredbv ^* U_i \nlsredbv^* { N_i}$.
			By \Cref{lem:snf_propagation}, each term  in the sequence  $ {T_i }\lsredbv ^* U_i \nlsredbv^* { N_i} $ is  $\surf$-normal. Hence, by \Cref{def:E}
			(since only the second rule   can apply), we  conclude that 
			$\mdist{T_i }\llRedx{\betav} ^* \mdist{U_i} \iRedx{\betav}^* \mdist{N_i}$.
		
			\item $T_i\not \in \Snf$.  By \Cref{lem:snf}, each term  in the sequence  ${T_i }\nsredbv^*{ N_i}$ is not $\surf$-normal. 
		 By definition of $\llred$ 
		(since only the first rule can apply) we  conclude that  $\mdist{T_i }\iRedbv^*  \mdist{ N_i}$.
%			\item $T_i  \not \in \Snf_{\oplus}$ or  $T_i   \not\in \Snf_{\betav}$.  By \Cref{lem:snf}, each term  in the sequence  ${T_i }\nsredbv^*{ N_i}$ is not $\surf$-normal. 
%			We conclude that $\mdist{T_i }\iRedbv^*  \mdist{ N_i}$
%				\item  $T_i   \not\in \Snf_{\betav}$ By \Cref{lem:snf}, each term  in the sequence  ${T_i }\nsredbv^*{ N_i}$ is not $\surf$-normal. 			We conclude that $\mdist{T_i }\iRedbv^*  \mdist{ N_i}$
%				
		
		\end{enumerate}
	Let  us partition $\td$ into two multi-distributions, collecting in $\td_1$ the terms of case 1. and in $\td_2$ the terms of case 2.
	We partition $\n$ so that $\td_1 \nsRed^* \n_1$ and $\td_2 \nsRed^* \n_2$.
	We have $\td_1 \llRedx{\betav} ^* \ud \iRedx{\betav}^* \n_1$ and $\td_2 \iRedbv^*  \n_2$. Therefore
	$\m \llRed^*(\td_1 + \td_2) \llRedx{} ^*( \ud + \td_2)  \iRed^* ( \n_1 + \n_2)=\n.$
	which proves the claim.
	\end{proof}
	
	\begin{prop}[neutrality]\label{prop:Pneutral} If 
		$\m\iRed \n$ then  $\obsN(\m) =\obsN(\n)$.
	\end{prop}
	\begin{proof} Consequence of the fact that if  $U\nlsredbv N$, then $N$ is not  $\betav$-normal \ArX{}{\Cref{lem:nfU}}.  
		Indeed 
		 $\nllred \subseteq \nsred$ and so  $M \nllred \r$ \iff ( $M \nlsredbv M'$ with  $\r = \mdist {M'}$).
	\end{proof}

	\paragraph{Diamonds.}\label{sec:diamonds}	\Cref{prop:fulldiamond}
	(the relation $\llfull$ is $\obsN$-diamond) 
 follows from the following key lemma.  Notice that	Point (2.) implies that $\obsN (\m_1) = \obsN (\m_2)$.
	\begin{lemma}[Pointed Diamond]
		\label{l:diamond}
		Let $\a,\c\in \{\beta_v,\oplus\}$.
		Assume  $M$ has two distinct redexes, such that $M \llredx{\a} \m_1 $ and $M \llredx{\c} \m_2$. 
		Then  \begin{enumerate}
			\item exists $\td$ such that $\m_1\llfull{}_{\c}\ \td$ and $\m_2\llfull{}_{\a}\ \td$.
			
			\item Moreover, no $M_i$ in  $ \m_1=\mdist{ p_i M_i}_{i} $  and no  $M_j $ in $\m_2=\mdist{q_jM_j}_{j}$ is $\red$normal.
		\end{enumerate}

	\end{lemma}
	\begin{proof}\begin{itemize}
			\item If $M$ is $\surf$-normal, then by definition of $\llred$,   $M \lsredx{\betav} \m_1 $ and $M \lsredx{\betav} \m_2$,
		and we conclude by using  \Cref{fact:trans} and \Cref{lem:Uproperties}, point 1.
		
\item 	If $M$ is $\surf$-reducible,	then by definition of $\llred$, $M \sredx{\betav} \m_1 $ and $M \sredx{\betav} \m_2$. We easily conclude by case analysis.
	\end{itemize}
	\end{proof}

	\begin{remark}[$\llred  \  \not= \  \lsredbb \cup \redo$]\label{rem:Edelicate} It is useful to notice that   $\llred  \  \not= \  \lsredbv \cup \redo$. Such a relation is neither diamond nor confluent.
		\begin{itemize}
			\item The lifting of  $\lsredbv \cup \redo $ is  neither diamond nor confluent.  	Consider  $ (\Delta \oplus  \Delta\Delta)(\lam z.Iz) $. Then
			$\m_1 = \mdist{\two \Delta (\lam z.Iz), \two  (\Delta\Delta) {(\lam z.\underline{Iz})}} \xrevredx{}{\oplus
			} (\Delta \oplus  \Delta\Delta)(\lam z.Iz) \lsredbv\mdist{(\Delta \oplus  \Delta\Delta)(\lam z.z)}=\m_2$.   The elements $\m_1$ and $\m_2$ cannot join, because  no $\lsredbv$-step can fire the underlined $(Iz)$.
			\item Similarly in CbN, for the lifting of $\lsredb \cup \redo$.	 Consider  $(\Delta \oplus  \Delta\Delta)(x(Iz))$.
		\end{itemize}
	\end{remark}%
\ArX{}{%
\begin{remark}
	Notice that the problem in \Cref{rem:Edelicate} does not happen. Indeed \\
	 $M= (\Delta \oplus  \Delta\Delta)(\lam z.Iz) \not\llred 
	\mdist{(\Delta \oplus  \Delta\Delta)(\lam z.z) } $
	because $M$ is not $\surf$-normal.
\end{remark}}
%	\begin{prop}[Diamond]\label{prop:fulldiamond}$\llfull$ is  $\obsN$-diamond.
%	\end{prop}

% !TEX root = main_2022.tex
\section{Details for  \Cref{sec:output}: more case studies}
\subsection{Asymptotic Normalization for a calculus with outputs}

\subparagraph{Proof of \Cref{lem:payoff}.} The $\weak$-factorization of $\red$ is proved in \cite{GavazzoF21},  where it is called surface factorization, and  proved in general for all CbV  monadic calculi, including the payoff calculus which we discuss here. $\obs$-neutrality is straightforward  to verify, by case analysis.

\subsection{Asymptotic Normalization and  B{\"o}hm Trees}\label{sec:BT}

We   show that the  B{\"o}hm Tree of a term $M$  is the (unique) limit of an asymptotically normalizing strategy, \ie the
limit of a \emph{single} reduction sequence.
\paragraph{B{\"o}hm Trees and Partial Normal Forms}
Following \cite{AmadioCurien}, the  B{\"o}hm Tree of a term $M$ is (the downward closure of) the set of the partial
normal forms of all reducts of $M$. 

\begin{Def}[Partial Normal Forms and B{\"o}hm Trees ]\label{def:BT}
			 
		The set  $\PNF$ of \textbf{partial normal forms} is defined as follows; 
		%\SLV{it has the order induced by $\Omega \leq P$.}{}
		\[\infer{\Omega\in \PNF}{} \qquad \infer {\lam  x_1\dots x_n.xA_1\dots A_n\in
			\PNF}{A_1\in \PNF \ \dots \ A_n\in \PNF}
		\]	
		$\PNF$ is a subset of the set of partial $\lam$-terms, defined by $P:=\Omega\mid x\mid
			PP\mid \lam x.P$, % and inherits its \textbf{order},	 which is defined by	
			and inherits its \textbf{order} $\leq$,	 which is generated by the following rules:
			\[\infer{\Omega \leq P}{} 
			\qquad   
			\infer{P_1P_2\leq P'_1P'_2}{P_1\leq P'_1 & P_2\leq P'_2}
			\qquad \infer{\lam x.P\leq \lam x.P'}{P\leq P'}\]
		 The  elements of  the ideal completion $\PNF^\infty$ of $\PNF$ are called
		\textbf{B{\"o}hm Trees}. Precisely: 
		\begin{enumerate}	
		\item 
		The function  $\omega: \Lambda\red \PNF$  associates to each 
		term $M\in \Lambda$  its \emph{partial normal form} $\pnf M$:
		\[\pnf M = 
		\left\{
		\begin{array}{ll}
			\Omega& \mbox{ if } M\not\in \Hnf \\
			\lam \vec x. x \pnf {M_1}\dots \pnf {M_p}&\mbox{ if } M=\lam \vec x.
			xM_1...M_p
		\end{array}
		\right.\]
		\item The \textbf{B{\"o}hm Tree of }$M$ is defined as below For a set $\mathcal S$, $\downarrow \mathcal S=\{ Q\in \PNF\mid Q\leq
		S\in \mathcal S \}$.
		\[\BT M := \bigcup\limits_{M\red^*N}  \downarrow \{\pnf N \}= \downarrow \{\pnf
		N \mid M\red^*N\}	\]
		
	\end{enumerate}
\end{Def}
\SLV{}{
	The following properties are standard to check
	\begin{lemma}\label{lem:basics}Let $M,M'\in \Lambda$.
		\begin{itemize}
			
			\item[i] $M\red M'$ implies $\pnf M \leq \pnf{M'}$
			\item[ii] $\pnf M \leq M$
			\item[iii]  $Q\in \PNF$ and $ Q\leq M$ implies $Q\leq \pnf M$
		\end{itemize}
	\end{lemma}
}
%	\begin{proof}By induction on the shape of $M$.
	%	\end{proof}

The following property is standard and easy-to-check (see \cite[Lemma 2.3.2]{AmadioCurien}.)
\begin{lemma}\label{lem:basics}
	Let $M,M'\in \Lambda$. If $M\redb M'$ then $\pnf M \leq \pnf{M'}$.
	%	\begin{enumerate}
		%		\item If $M\redb M'$ then $\pnf M \leq \pnf{M'}$.
		%		\item $\pnf M \leq M$.
		%		\item If $Q\in \PNF$ and $ Q\leq M$, then $\pnf Q \leq \pnf M$.
		%	\end{enumerate}
\end{lemma}
\Cref{lem:basics} guarantees that $((\Lambda, \redb), \obs)$ is a QARS where $\obs \colon \Lambda\to \Nnf_\omega^\infty$ is defined as $\obs(M)= \downarrow
\{\pnf  {M}\}$.
%Note that the elements of $\obs(M)$ are finite.

\paragraph{Asymptotic Normalization.}
Let us define $\obs: \Lambda\to \Nnf_\omega^\infty$ as $\obs(M)= \downarrow
\{\pnf  {M}\}$. It is easily checked that $\brack{(\Lambda, \redb),\obs }$	  is a QARS. We show that 
the  B{\"o}hm Tree of a term $M$  can be obtained by asymptotic normalization, 
as the limit a $\xredx{\pd}{\beta}$ reduction sequence, which is an $ \obs$-normalizing strategy for $\redb$ (\Cref{thm:mainBT}).

The $\obs$-limit of a reduction sequence 
$\seq M$ is then
$   \sup_i\{\obs(M_i)\} = \bigcup_i  \downarrow \{\pnf  {M_i}\}	$.
$ \BT M$ is clearly the sup of the set $\wLim{(M,\redb)}$. We show that  $\BT M$ \emph{belongs} to that set, by proving that 
$\wLim{(M,\redb)}$ has a greatest element $\sem M$; this necessarily is  $ \BT M$.
%, because $\wLim{(M,\redb)}$ has a greatest element.   
%\paragraph{Asymptotic Completeness}
%\RED{\begin{prop}[$\los$-Factorization of $\redb$.]In $(\Lambda, \redb)$:\hfill
		%	\begin{center}
			%		$M \redb^*N$ implies $M \lsred^*\cdot \nlsred^* N$
			%	\end{center}
		%\end{prop}}
		%
		%
		
		We first  show that  $\lsredb$ is  $\obs$-complete for $\redb$ (Point 1 in \Cref{prop:Oasymptotic_BT} below). Reduction $\lsredb$ is not $\obs$-normalizing  for $\redb$ (for example, it admits the  sequence $x(\Delta\Delta)(Iz) \lsred x(\Delta\Delta)(Iz) \lsred \dots$) but its parallel version $\xredx{\pd}{}$ (\Cref{sec:parallelU}) is. 
		
		We  proceed similarly to \Cref{sec:PCbV} (think $\llRed$ vs $\llfull$). We consider the reduction 
		$\xredx{\pd}{}$ (the  explicit definition is in \Cref{sec:parallelU}) which has  $\obs$-Random Descent (trivially) and is 
		asymptotically complete for $ \lsred $ (Point 2 in \Cref{prop:Oasymptotic_BT} below), and so for $\redb$.
		
		% because of 
		%$\los$-Factorization  of $\redb$ (\Cref{prop:out-factorization}) and $\neg\out$ neutrality.
		%\begin{lemma}[$\neg\out$ neutrality ]\label{prop:neutral_BT} If 
		%	$M\nlsred M'$ then  $\pnf {M} = \pnf {M'}$.
		%\end{lemma}
		%
		%\begin{proof}By induction on $M$, observing that $M\in \Hnf$ iff  $M'\in \Hnf$,
		%	and that the two terms have the same shape.
		%\end{proof}
		%
		%

		\begin{lemma}\label{lem:neutral_BT}
			If	$M\nlsredb M'$ then  $\pnf {M} = \pnf {M'}$.
		\end{lemma}
\proof
			First,  observe that $M\nhredb M'$, because	$\nlsredb \ \subseteq \ \nhredb$. \ArX{}{(See \Cref{rem:Usurf}).}		
			\begin{itemize}
				\item 	If $M$ is not $\head$-normal ($M\not \in \Hnf)$, neither is  $M'$, by  
				(\Cref{lem:snf}).
				Therefore, $\pnf {M} = \Omega = \pnf {M'}$.
				
				\item 	Otherwise,  $M$ is $\head$-normal, that is, $M =  \lam \vec x. 	xM_1 \dots M_p$. 
				%	Hence, $\pnf M = \lam \vec x. x \pnf {M_1}\dots \pnf {M_p}$.
				As $M \nhredb M'$, necessarily \SLV{}{(by using \Cref{fact:shape})} $M' =  \lam \vec x. 	xM_1 \dots M_i' \dots M_p$ (which is head normal) and  $M_i \redb M_i'$ for some $1 \leq i \leq p$.
				It is impossible that $M_i \lsredb M_i'$, otherwise $M \lsredb M'$ according to the definition of $\lsredb$ (\Cref{def:liberal}).
				Therefore, $M_i \nlsredb M_i'$ and so, by \ih, $\pnf{M_i} = \pnf{M_i'}$.
				Thus, $\pnf M = \lam \vec x. x \pnf {M_1}\dots \pnf{M_i} \dots \pnf {M_p} = \lam \vec x. x \pnf {M_1}\dots \pnf{M_i'} \dots \pnf {M_p} = \pnf{M'}$.
			\end{itemize}

		\begin{restatable}[Asymptotic Completeness]{proposition}{AComplBT}\label{prop:Oasymptotic_BT}\label{prop:ACompl_BT}
			\begin{enumerate}
				\item 	$M \redb\conv{} \cpo{\r}$ implies 	$M \lsred\conv{} \cpo{\r}$, because 
				\begin{itemize}
					\item $\los$-Factorization of $\redb$: $M \redb^*N$ implies $M \lsred^*\cdot \nlsred^* N$
					\item $\obs$-neutrality : 	$M\nlsred M'$ then  $\obs {M} = \obs {M'}$.
				\end{itemize}
				\item 	$M \lsred\conv{} \cpo{\r}$ implies 	($M \xredx{\pd}{}\conv{} \cpo{\s}$ and
				$\cpo{\r} \leq \cpo{\s}$) \qedhere
			\end{enumerate}
		\end{restatable}

\proof
		\begin{enumerate}
			\item Factorization  is that for $\lsredb$ (details of the proof are in \cite{long}). $\obs$-neutrality is immediate consequence of  
			\Cref{lem:neutral_BT}.
			\item It follows by \Cref{lem:parallelU}, and the fact that $\obs$ is  monotonic. 
			\qedhere
		\end{enumerate}
{	\begin{remark}[Unique Limit]
 If we take for head reduction the standard  one (as in \cite{Barendregt}), then $\xredx{\pd}{}$ is deterministic. Otherwise, 
if we take $\hredb$ as defined in \Cref{sec:lambda},  it is easily verified that $\xredx{\pd}{}$ has the $\obs$-diamond property. 
	\end{remark}
}
		$  \xredx{\pd}{} $ is $\obs$-complete for $\redb$ (by Points 1. and 2.). % and moreover is  $ \obs $-diamond.
		Hence       we conclude by \Cref{thm:ACompl}:
		\begin{theorem}[Main, B{\"o}hm Trees]\label{thm:mainBT}
			$ \xredx{\pd}{} $ is  a (multi-step) $\obs$-normalizing strategy for $\redb$, and $M \xredx{\pd}{}\conv{}  \BT{M}$.
		\end{theorem}
		
		%\begin{theorem}[Main]\label{}The   QARS
		%	$\brack{(\Lambda,\redb), \obs}$ satisfies, 
		%	for each $M \in \Lambda$:
		%	\begin{enumerate}
			%		\item[i] $\sem M$ is  defined 	
			%		\item[ii] $M \xredx{\pd}{}\conv{}  \sem {M}$
			%		
			%	\end{enumerate}
		%	Moreover   $\sem M =\BT{M}$.
		%\end{theorem}

\end{document}